\documentclass[11pt,a4paper,thmsb,ukenglish,fleqn]{article}
\usepackage{epsf,  amsmath, amssymb, graphicx}
\usepackage{amsthm}
\usepackage[sort]{natbib}
\usepackage{a4wide}
\usepackage{subfig, caption, multicol} 

\usepackage{stix, mathdots}

\usepackage{xcolor}
\usepackage{pdflscape}
\usepackage{hyperref}

\usepackage{algorithm, algpseudocode}
\usepackage{fancyhdr}
\newcommand{\mycomment}[1]{}

\newcommand{\R}{\ensuremath{\Bbb{R}}}
\newcommand{\Z}{\ensuremath{\Bbb{Z}}}
\newcommand{\N}{\ensuremath{\Bbb{N}}}
\newcommand{\E}{\ensuremath{\Bbb{E}}}
\newcommand{\Var}{\ensuremath{\mathrm{Var}}}
\newcommand{\Cor}{\ensuremath{\mathrm{Cor}}}
\newcommand{\Cov}{\ensuremath{\mathrm{Cov}}}

\newcommand{\leb}{\ensuremath{\mathrm{Leb}}}

\newcommand{\Log}{\ensuremath{\mathrm{Log}}}

\newcommand{\drift}{\ensuremath{\zeta}}
\newcommand{\lev}{\ensuremath{\ell}} 

\newcommand{\indicator}{\ensuremath{\mathbb{I}}}

\def\Levy{L\'{e}vy }
\def\levy{L\'{e}vy}

\newcommand{\ind}{\ensuremath{\Bbb{I}}}

\newtheorem{theorem}{Theorem}

\newtheorem{assumption}{Assumption}

\newtheorem{definition}[theorem]{Definition}
\newtheorem{example}{Example}

\newtheorem{lemma}{Lemma}

\newtheorem{proposition}{Proposition}
\newtheorem{remark}{Remark}

\allowdisplaybreaks
\begin{document}
\title{
Periodic trawl processes: Simulation, statistical inference and applications in energy markets
}

 \author{ \textsc{Almut E.~D.~Veraart}\\
 \textit{Department of Mathematics, Imperial College London}\\
 \textit{ 180 Queen's Gate, 
 London, SW7 2AZ, 
 UK}  \\
 \texttt{a.veraart@imperial.ac.uk}
 }
\maketitle

\begin{abstract}
This article introduces the class of periodic trawl processes, which are continuous-time, infinitely divisible, stationary stochastic processes, that allow for periodicity and flexible forms of their serial correlation, including both short- and long-memory settings.  We derive some of the key probabilistic properties of periodic trawl processes and present relevant examples. Moreover, we show how such processes can be simulated and establish the asymptotic theory for their sample mean and sample autocovariances. Consequently, we prove the asymptotic normality of a (generalised) method-of-moments estimator for the model parameters. 
We illustrate the new model and estimation methodology in an application to electricity prices. 
\end{abstract}
\noindent
{\it MSC2020 subject classifications:} Primary 62M15, 
62F99, 
60G10; 
 secondary 
 62P99, 
 60F05, 
  60G57, 
  60E07. 
\maketitle{}

\section{Introduction}
Time series with periodic behaviours are ubiquitous and are present in many applications such as climate science, meteorology, oceanography, hydrology, economics, and communications networks. Our article is particularly motivated by applications in energy, carbon, and temperature markets. 
Here, we need to consider different types of seasonalities and periodicities when, for example, modelling (renweable) energy production, demand and prices that are influenced by the seasonal weather patterns, as well as different daily and weekly fluctuations and economic cycles.

Traditionally, stochastic models for electricity prices focussed on mean-reverting processes since spot prices can be regarded as equilibrium prices balancing supply and demand.  \cite{BNBV2013} argued that, in order to achieve a weaker version of mean reversion, one could model prices directly by a stationary process. We follow the same idea in this article and are interested in developing a flexible class of stationary stochastic processes which can allow for periodic behaviour rather than employing widely-used seasonal-trend  decomposition models. 

There is a broad literature on cyclostationary processes (and their extensions), which are typically constructed by combining stationary processes with periodic functions.
The textbook by 
 \cite{HurdMiamee2007} contains detailed historical background and 
 \cite{GARDNER2006} provides an excellent review; the name \emph{cyclostationary} processes  first appears in  
\cite{Bennett1958}, whereas other authors referred to them as \emph{periodically correlated} processes, see e.g.~\cite{G1961}, 
\cite{Hurd1969}, and \cite{GentonHall2007}, 
\cite{DasGenton2021} for more recent work.
There has been much interest in developing suitable statistical methodologies for periodic and cyclostationary processes and their extensions. For instance,
\cite{QuinnThompson1991} focuses on  estimating the frequency of a periodic function, whereas	
\cite{HallReimannRice2000} presents a nonparametric estimator of a periodic function.

Moreover, from a stochastic modelling perspective, it is relevant to construct models which can combine periodicity with either short- or long-memory settings, with the latter being typically more challenging. However, there is a long 
tradition in time series analysis tackling long-memory settings and incorporating periodicities: Fractional differencing, introduced by \cite{Hosking1981} and further developed by \cite{Andel1986} lead to long memory time series models and fractionally integrated ARMA (FARIMA) processes, in particular. The processes were further generalised by 
\cite{GZW1989}, who introduced Gegenbauer and Gegenbauer ARMA (GARMA) processes, see also \cite{WCG1998} for a {$k$}-factor GARMA long-memory model and 
\cite{ELORM2015} for extensions to Gegenbauer random fields.
The concept of cyclical long-memory processes and the appropriate statistical methodology has been further explored in  
\cite{AR2000}, \cite{GHR2001}, \cite{FG2001} and \cite{A2002}.
For example, \cite{HS2004} study the estimation of the location and exponent of the spectral singularity of a long-memory process. They maximise the periodogram to estimate the location of the singularity, in the spirit of \cite{Y2007}, who proved consistency in the Gaussian case and proved consistency under weaker (non-Gaussian) assumptions \cite{HS2004}.
They used a GPH estimator, see \cite{GPH1983}, for the long-memory parameter. Also, \cite{H2005} discusses the semiparametric estimation of stationary processes whose spectra have an unknown pole.
\cite{Maddanu2022} presents a harmonically weighted filter for cyclical long-memory processes.
In practice, one often estimates the parameters of a cyclical/seasonal long-memory process in two steps: First, one estimates the location of the singularity in the spectral density function, second, with the location parameter now fixed, one estimates the long memory parameter, see \cite{FG2001} and  \cite{ELORM2015}.
More recently, a joint estimation of the periodicity and the long-memory parameter have been developed in 
\cite{AAFO2020,AFNO2022}.

While there is already extensive literature on discrete-time periodic processes with a short  or long memory, much less is known in the continuous-time setting. To fill this gap, this article introduces the class of \emph{periodic trawl processes}, which are continuous-time, infinitely divisible, stationary stochastic processes, which allow for periodicity and flexible forms of their serial correlation, including both short- and long-memory settings. 

Periodic trawl processes extend the class of trawl processes first introduced under this name in \cite{BN2011} but were earlier mentioned (under a different name) in  
\cite{WolpertTaqqu2005} and \cite{WolpertBrown2012}. 
Motivated by applications in high-frequency financial data,  
\cite{BNLSV2014} and \cite{VERAART2019} introduced integer-valued trawl (IVT) processes in a univariate and multivariate framework, respectively. 
The statistical methodology for such processes was further developed in 
\cite{BLSV2023,SheYan,Shephard2016}.
Moreover, it has been shown that trawl processes can be applied in 
 hierarchical models for extremes, see  \cite{Noven,CourgeauVeraart2022,Bacro}.
There is also a growing interest in discrete-time versions of trawl processes, which were first introduced by \cite{DoukJakLopSurg18}, see also \cite{DRR2020},  and limit theorems for discrete- and continuous-time trawl processes, see for instance  \cite{GRAHOVAC2018235,Paulau,TalTre2019,PPSV2021}.
   
\subsection{Outline and main contributions of the article}
The outline for the remainder of the article is as follows.
Section \ref{sec:background} briefly reviews mixed moving average processes (MMAPs) and then defines (periodic) trawls as special cases of MMAPs. 
It  presents some of the key probabilistic properties of periodic trawl processes and relevant examples. 
Next, Section \ref{sec:simulationalg} describes how (periodic) trawl processes can be simulated and the corresponding {\tt R}  code is 
available via the {\tt R} package {\tt ambit} on CRAN, see \cite{ambit}.
Next, we tackle the question of how the parameters of periodic trawl processes can be estimated. 
In order to develop a suitable inference methodology for periodic trawl processes, we first develop a general asymptotic theory for MMAPs in Section \ref{sec:asymptotictheory}.  
Here we present central limit theorems for the sample mean, sample autocovariance and sample autocorrelations of MMAPs. Next,
Section \ref{sec:inference} tailors the asymptotic theory to the case of (periodic) trawl processes and develops suitable moment-based estimators. Our main focus here is on estimating the parameters of the kernel function in the periodic trawl process, which determine the decay of the autocorrelation function as well as the periodicity. This theory will be developed under the assumption that the periodicity of the process is known.
Next, the more general setting is presented in Section \ref{sec:inferenceGMM}, where the asymptotic theory is developed for a generalised method-of-moment approach.
The new estimation methodology is illustrated in an empirical study of electricity prices in Section \ref{sec:empirics}; 
the code used in the empirical study is
 available on GitHub
and archived in Zenodo, see \cite{PeriodicTrawl-Energy}.
 Section \ref{sec:conclusion} concludes.
 The proofs of the theoretical results, additional examples and detailed discussions of the technical assumptions are relegated to the Appendix, see  Section \ref{ap:proofs}.

\section{Mixed moving average  and periodic trawl processes}\label{sec:background}
Let us start off by reviewing the definition of the broader class of mixed moving average processes (MMAPs), see for instance \cite{SRMC1993, F2005, FuchsStelzer2013}; we will  then introduce periodic trawl processes as a special case of MMAPs. 

\subsection{Background}
We briefly review the concept of 
\Levy bases. 
To this end, consider a Borel set 
 $S\subseteq \R$  and denote by $\mathcal{S}=\mathcal{B}(S)$  the Borel $\sigma$-algebra on $S$ and  by $\mathcal{B}_b(S)$  the bounded Borel sets of $\mathcal{S}$.

\begin{definition}\label{def:LevyBasis}
	Let $ L=\{ L(A)\,:\,A\in\mathcal{B}_b(S)\}$ denote a collection of $\R$-valued  random variables. Then, $L$ is called an $\R$-valued \emph{L\'{e}vy basis}  on $S$ if  the following three conditions hold: 
1) Random measure: For any
		sequence $A_{1},A_{2},\dots$ of disjoint elements of 
		$\mathcal{B}_b(S)$ 
		satisfying $\cup _{j=1}^{\infty }A_j\in \mathcal{B}_b(S)$
		it holds that   $L\left( \cup _{j=1}^{\infty }A_j\right) =\sum_{j=1}^{\infty }
		 L(A_{j}) $ a.s..
	2) Independent scatteredness: For any
		sequence $A_{1},A_{2},\dots$ of disjoint elements of 
		$\mathcal{B}_b(S)$, the random variables  
		$L(A_{1}), L(A_{2}),\dots$ are independent.
		3) Infinite divisibility: For any
		$A\in \mathcal{B}_b(S)$, the law of $L(A)$ is infinitely divisible (ID).
	\end{definition}

Recall that a \Levy basis admits 
a \levy-Khintchine representation, see  \citet[Proposition 2.1 (a)]{RajputRosinski1989}. 
\begin{proposition}\label{prop:Cha5_LevKhi}
	Let $L$ denote a \Levy basis, and let $\theta \in \R$ and $A \in \mathcal{B}_b(S)$, then the cumulant function of $L$ is given by $C(\theta; L(A))=\Log\left(\mathbb{E}(\exp(i \theta L(A))\right)$, where
	\begin{align}\begin{split}\label{eq:LevKhi}
			C(\theta; L(A))
			& = i \theta \drift^*(A) - \frac{1}{2}\theta^2 a^*(A) + \int_{\mathbb{R}}\left(e^{i\theta \xi}-1-i\theta \xi \mathbb{I}_{[-1,1]}(\xi)\right) n(d\xi, A), 
	\end{split}\end{align}
	where   
	$\drift^*$ is a signed measure on $\mathcal{B}_b(S)$, $a^*$ is a measure on $\mathcal{B}_b(S)$, and 
	$n(\cdot,\cdot)$ is the generalised L\'{e}vy measure, i.e.~$n(d\xi, A)$ is a L\'{e}vy measure on $\mathbb{R}$ for fixed $A \in \mathcal{B}_b(S)$ and a measure on $\mathcal{B}_b(S)$ for fixed $dx$.
\end{proposition}
 Here $\Log$  denotes the \emph{distinguished logarithm}, see \citet[p.~33]{Sato1999}.

 In the following, we shall restrict our attention to homogeneous \Levy bases. 
\begin{definition}[Homogeneous \Levy basis]
	A \Levy basis $L$ with L\'{e}vy-Khintchine representation \eqref{eq:LevKhi} is called homogeneous if we have the following representations for ${\bf z}\in \R^2$:
\begin{align}\label{eq:CQ}
	\drift^*(d{\bf z})=\drift d{\bf z}, &&
a^*(d{\bf z})= a d{\bf z}, &&
	n(d\xi, d{\bf z}) = \lev(d\xi) d{\bf z},
\end{align}
for constants $\drift \in \R, a\geq 0$ and a \Levy measure $\lev$. We call $(\drift, a, \lev)$ the characteristic triplet associated with $L$.
\end{definition}
In the case of a homogeneous \Levy basis with characteristic triplet $(\drift, a, \lev)$
the \levy-Khintchine representation %
 simplifies to  
	\begin{align}\begin{split}\label{eq:LevKhiHom}
			C(\theta; L(A))
			& = \leb(A)\left[i \theta \drift  - \frac{1}{2}\theta^2 a  + \int_{\mathbb{R}}\left(e^{i\theta \xi}-1-i\theta \xi \mathbb{I}_{[-1,1]}(\xi)\right) \lev(d\xi)\right],
	\end{split}\end{align}
for $\theta \in \R$ and $A \in \mathcal{B}_b(S)$.
In the following, we shall denote by $L'$ the \emph{\Levy seed} associated with the \Levy basis with characteristic triplet $(\drift, a, \lev)$, i.e.~a random variable with cumulant function given by 
	\begin{align}
		\label{eq:LevSeed}
		C(\theta; L')=\Log\left(\mathbb{E}(\exp(i \theta L'))\right)
		 = i \theta \drift  - \frac{1}{2}\theta^2 a  + \int_{\mathbb{R}}\left(e^{i\theta \xi}-1-i\theta \xi \mathbb{I}_{[-1,1]}(\xi)\right) \lev(d\xi),
\end{align}
for $\theta \in \R$.
We recall that a homogeneous \Levy basis admits a L\'{e}vy-It\^o representation, cf.~\citet[Proposition 4.5]{Ped03}:
\begin{proposition}\label{pro:LevIto}
	For a homogeneous \Levy basis $L$ with a characteristic triplet $(\drift, a, \lev)$  there exists a modification $L^*$ with the same characteristic triplet that has the following L\'{e}vy-It\^{o} decomposition. Let $A\in \mathcal{B}_b(S)$, then 
	\begin{align}\label{eq:Cha5_LevyIto}
		L^*(A) 
		&= \drift \leb(A)+ W(A) +
		\int_{\{|y|\leq 1\}} y(N- n)(dy,A)
		+ \int_{\{|y|> 1\}} y N(dy,A), 
	\end{align}
		for a Gaussian basis $W$, with characteristic triplet $(0, a, 0)$, i.e.~$W(A)\sim N(0,a\leb(A))$, and a Poisson basis $N$ (independent of $W$) with compensator
	$n(dy;A)=\mathbb{E}(N(dy;A))$, where $n(dx,d{\bf z})= \lev(dx) d{\bf z}$.
\end{proposition}

\subsection{Definition of a mixed moving average process}
\begin{definition}[Mixed moving average process]\label{def:MMA}
	Let $L$ denote a homogeneous \Levy basis with characteristic triplet $(\drift, a, \lev)$, $f: \R \times \R\to \R$ be a deterministic function. The stochastic process $Y=(Y_t)_{t\geq 0}$ with 
\begin{align}\label{eq:PT}
	Y_t 
	&=
	\int_{\R\times \R}f(x,t-s) L(dx,ds),
\end{align}
is called a mixed moving average  process (MMAP).
\end{definition}
An application of  \citet[Theorem 2.7]{RajputRosinski1989} leads to the following result.
\begin{proposition}\label{pro:exMMA}
The MMA process defined in \eqref{eq:PT} is well-defined, if and only if the following integrability conditions hold almost surely:
	\begin{multline}\label{eq:intcondMA}
			\int_{\R \times \R} | V_1(f(x,t-s))|dxds < \infty, \qquad
					a\int_{\R\times \R} |f(x,t-s)|^2 dx ds < \infty, \\				\int_{\R\times \R} V_2(f(x,t-s))dxds < \infty,
		\end{multline}
			where for $\varrho(\xi) := \xi\indicator_{[-1, 1]}(\xi)$, we define
			$V_1(u):= u \drift +\int_{\mathbb{R}}\left(\varrho(\xi u) -u\varrho(\xi)\right)\lev(d\xi)$ and
			$V_2(u) := \int_{\mathbb{R}} \min(1,|\xi u|^2)\lev(d\xi)$.
			In that case, the cumulant function of the MMA process is given by 
			\begin{align*}
	C(\theta; Y_t)	&=\Log(\mathbb{E}(\exp(i \theta Y_{t})))
	= \int_{\R \times \R}C( \theta f(x, t-s); L')dx ds\\
	&= \int_{\R \times \R}C( \theta f(x,u); L')dx du\\
		&=i \theta \drift \int_{\R \times \R}f(x,u)dxdu  - \frac{1}{2}\theta^2 a  \int_{\R \times \R}f^2(x,u)dxdu \\
	&+ \int_{\R \times \R}  \int_{\mathbb{R}}\left(e^{i\theta \xi f(x,u)}-1-i\theta \xi f(x,u) \mathbb{I}_{[-1,1]}(\xi f(x,u))\right) \lev(d\xi)dxdu,
\end{align*}
for $\theta \in \R$, 
		which implies that the process is stationary and  infinitely divisible. 
\end{proposition}

\subsection{Periodic trawl processes}
We will now show that trawl processes and the new class of periodic trawl processes fall into the MMAP framework. 
First, recall the definition of a periodic function. 
\begin{definition}
A function $p:[0, \infty)\to \R$ is called periodic with period $\tau>0$ if $p(x+\tau)=p(x)$ for all $x\geq 0$. 
\end{definition}
Now, we have all the ingredients for defining a periodic trawl process. 
\begin{definition}[Periodic trawl process]\label{def:PT}
	Let $L$ denote a homogeneous \Levy basis with characteristic triplet $(\drift, a, \lev)$, $p:[0, \infty)\to \R$ a periodic function  with period $\tau>0$ and $g:[0, \infty)\to \R$ a continuous, monotonically decreasing function. The stochastic process $Y=(Y_t)_{t\geq 0}$ with 
\begin{align}\label{eq:PT}
	Y_t 
	&=
	\int_{\R\times \R}f(x,t-s) L(dx,ds),
\end{align}
with  
$f(x,t-s)=p(t-s)\ind_{(0,g(t-s))}(x)\ind_{[0,\infty)}(t-s)$
is called a $\tau$-periodic trawl process. 
\end{definition}
\begin{remark}
	In the case where $p\equiv 1$, we obtain a monotonic trawl process without periodicity. 
	Note that we use a slightly different notation than that used in \cite{BNLSV2014}. Here $g$ is defined for non-negative arguments.
\end{remark}

\begin{remark}\label{rem:alt-trawl-main}
	\cite{BNLSV2014} proposed adding a periodic function as a multiplicative factor to $g$ rather than as a kernel function as in \eqref{eq:PT}. However, we find that the version proposed above appears to be more analytically tractable, as we shall discuss in more detail in the appendix; see Remark \ref{rem:alt-trawl}.
	Note also that if one is interested in integer-valued trawl processes, then they can be obtained by choosing $p$ to be integer-valued. 
\end{remark}

The following result is a special case of Proposition \ref{pro:exMMA}.
\begin{proposition}
	The $\tau$-periodic trawl process defined in Definition \ref{def:PT} is well defined, in the sense that the stochastic integral exists, if and only if the
	integrability conditions stated in \eqref{eq:intcondMA} hold almost surely. 
In that case, the cumulant function of the periodic trawl process is given by 
			\begin{align*}
	C(\theta; Y_t)	&=\Log(\mathbb{E}(\exp(i \theta Y_{t})))
	= \int_{-\infty}^t\int_{\R}C( \theta p(t-s)\ind_{(0,g(t-s))}(x); L')dx ds
 \\
&
 = \int_0^{\infty}\int_{\R}C( \theta p(u)\ind_{(0,g(u))}(x); L')dx du\\
 	&
 =i \theta \drift \int_{0}^{\infty}g(u) p(u)du  - \frac{1}{2}\theta^2 a  \int_{0}^{\infty}g(u) p^2(u)du \\
	&+ \int_{0}^{\infty}g(u) \int_{\mathbb{R}}\left(e^{i\theta \xi p(u)}-1-i\theta \xi p(u) \mathbb{I}_{[-1,1]}(\xi p(u))\right) \lev(d\xi)du,
\end{align*}
for $\theta \in \R$, 
		which implies that the process is stationary and  infinitely divisible. 
\end{proposition}

\begin{remark}
	A sufficient condition for the integrability conditions to hold is that 
 $p$ is continuous and $\int_0^{\infty}g(u)du<\infty$. This is indeed the scenario we typically consider for applications. 
\end{remark}

\subsubsection{Second-order properties of periodic trawl processes}
Let us now study the second-order properties of periodic trawl processes. The proofs of our results are given in  Appendix \ref{ap:SP}. 

\begin{proposition}\label{prop:SecondOrder}
	Let $Y$ denote a $\tau$-periodic trawl process as defined in Definition \ref{def:PT}.
	For $t\geq 0$, we have  
	$\E(Y_t)= \E(L')\int_0^{\infty}p(u)g(u)du$,$	\Var(Y_t)
			= \Var(L')\int_{0}^{\infty}p^2(u)g(u)du$,
			$\Cov(Y_{0},Y_{t})
		= \Var(L')\int_{0}^{\infty}p(u)p(t+u)g(t+u)du$
		and
		\begin{align*}
		\Cor(Y_{0},Y_{t})
		= \frac{\int_{0}^{\infty}p(u)p(t+u)g(t+u)du}{\int_{0}^{\infty}p^2(u)g(u)du}.
	\end{align*} 
\end{proposition}

We get the following important result.

\begin{proposition}\label{prop:Cor}
		Let $Y$ denote a $\tau$-periodic trawl process as defined in Definition \ref{def:PT}. Suppose that the periodic function $p$ is continuous.  Then there exists a continuous $\tau$-periodic function $c:[0,\infty)\to \R$, which is proportional to $p$ and satisfies $c(0)=1$,  such that 
		\begin{align*}
			\Cor(Y_{0},Y_{t})
			&=c(t)\frac{\int_{0}^{\infty}g(t+u)du}{\int_0^{\infty}g(u)du}.
		\end{align*}
\end{proposition}

We note that the $\tau$-periodic function $c$ can be represented via its corresponding  Fourier series representation, which is useful if a parametric model for $c$ is required.

\begin{remark}
The corresponding results for a trawl process follow by setting $p\equiv c\equiv 1$.
\end{remark}

\subsubsection{Examples}
Let us study some examples of periodic trawl processes. 

\begin{example}
	Consider an exponential trawl function given by  
$g(x)=\exp(-\lambda x)$ for $\lambda>0, x\geq 0$ and the $\tau$-periodic function $p(x)=\sin(2\pi x/\tau)$.
Then
\begin{align*}
\E(Y_t)&
=\E(L')\frac{2\pi \tau}{\lambda^2 \tau^2 +4\pi^2},\\ 
\Var(Y_t)
&=\Var(L')\frac{8\pi^2}{\lambda^3 \tau^2+16\pi^2 \lambda},\\
\Cov(Y_{0},Y_{t})
&=  \Var(L') 2\pi\frac{e^{-\lambda t}\left[\sin\left(\frac{2\pi t}{\tau}\right)\tau \lambda +4\pi\cos\left(\frac{2\pi t}{\tau}\right)\right]}{\lambda(\lambda^2\tau^2 + 16\pi^2)},\\
\Cor(Y_{0},Y_{t})
&= e^{-\lambda t} c(t),
\quad \mathrm{with} \quad c(t) = \frac{1}{4\pi}\left[\sin\left(\frac{2\pi t}{\tau}\right)\tau \lambda +4\pi\cos\left(\frac{2\pi t}{\tau}\right)\right],
\end{align*} 
where $c$ is a $\tau$-periodic function.

\end{example}
Next, we consider an example based on the supGamma trawl function,  which can allow for both short and long memory.
\begin{example} Consider the supGamma trawl function defined as  $g(x)=\left(1+\frac{x}{\alpha}\right)^{-H}$, for $\alpha>0, H>1, x \geq 0$. As before, let $p(x)=\sin(2\pi x/\tau)$. Then
\begin{align*}
\Cov(Y_{0},Y_{t})
&= \Var(L')\int_{0}^{\infty}p(u)p(t+u)g(t+u)du
\\
&= \Var(L')\int_{0}^{\infty}\sin\left(\frac{2\pi u}{\tau}\right)\sin\left(\frac{2\pi(t+u)}{\tau}\right) \left(1+\frac{(t+u)}{\alpha}\right)^{-H}du. 
\end{align*} 
We note that
	$\int_0^{\infty}g(t+u)du=
\int_{0}^{\infty} \left(1+\frac{(t+u)}{\alpha}\right)^{-H}du=\frac{\alpha^H}{H-1}(t+\alpha)^{1-H}$.
According to Proposition \ref{prop:Cor}, there exists a $\tau$-periodic function $c$ such that 
	\begin{align*}
	\Cor(Y_{0},Y_{t})
	&=c(t)\frac{\int_{0}^{\infty}g(t+u)du}{\int_0^{\infty}g(u)du}
	=c(t)\left(1+\frac{t}{\alpha}\right)^{1-H}.
\end{align*}
For $H\in(1,2]$, we are in the long-memory case and for $H>2$ in the short-memory case.
\end{example}

\begin{example}
Let us briefly consider the general case of a superposition trawl. Let 
\begin{align*}
g(x)=\int_0^{\infty}e^{-\lambda x}f_{\lambda}(\lambda)d\lambda,
\end{align*}
for a density $f_{\lambda}$. If 
\begin{align*}
f_{\lambda}(\lambda)=\frac{1}{\Gamma(H)}\alpha^H\lambda^{H-1}e^{-\lambda \alpha},
\end{align*}
for $\alpha>0, H>1$. 
Then 
\begin{align*}
g(x)=\int_0^{\infty}e^{-\lambda x}\frac{1}{\Gamma(H)}\alpha^H\lambda^{H-1}e^{-\lambda \alpha}d\lambda =\left(1+\frac{x}{\alpha}\right)^{-H},
\end{align*}
i.e.~we are in the case of the previous example. 
Then 
\begin{align*}
\Cov(Y_{0},Y_{t})
&= \Var(L')\int_{0}^{\infty}p(u)p(t+u)g(t+u)du,
\end{align*}
where

\begin{align*}
&\int_{0}^{\infty}p(u)p(t+u)g(t+u)du\\
&=\int_{0}^{\infty}p(u)p(t+u)\int_0^{\infty}e^{-\lambda (t+u)}f_{\lambda}(\lambda)d\lambda du\\
&=\int_0^{\infty}  e^{-\lambda t} \int_{0}^{\infty}p(u)p(t+u)e^{-\lambda u} duf_{\lambda}(\lambda)d\lambda \\
&= \int_0^{\infty} 2\pi\frac{e^{-\lambda t}\left[\sin\left(\frac{2\pi t}{\tau}\right)\tau \lambda +4\pi\cos\left(\frac{2\pi t}{\tau}\right)\right]}{\lambda(\lambda^2\tau^2 + 16\pi^2)} f_{\lambda}(\lambda)d\lambda\\
&= 2\pi \tau\sin\left(\frac{2\pi t}{\tau}\right) \int_0^{\infty}\frac{e^{-\lambda t}  }{(\lambda^2\tau^2 + 16\pi^2)} f_{\lambda}(\lambda)d\lambda
+
 8\pi^2 \cos\left(\frac{2\pi t}{\tau}\right) \int_0^{\infty} \frac{e^{-\lambda t}}{\lambda(\lambda^2\tau^2 + 16\pi^2)} f_{\lambda}(\lambda)d\lambda.
\end{align*} 

From Proposition \ref{prop:Cor}, we deduce that there exists a $\tau$-periodic function $c$ such that 
\begin{align*}
	\Cor(Y_0,Y_t)=c(t)\frac{\int_{0}^{\infty}g(t+u)du}{\int_{0}^{\infty}g(u)du}.
\end{align*}
Note that 
\begin{align*}
\int_{0}^{\infty}g(t+u)du&=\int_{0}^{\infty}\int_0^{\infty}e^{-\lambda (t+u)}f_{\lambda}(\lambda) d\lambda du =
\int_{0}^{\infty} e^{-\lambda t} f_{\lambda}(\lambda) \int_0^{\infty} e^{-\lambda u}  du  d\lambda\\
&= \int_{0}^{\infty} \frac{1}{\lambda} e^{-\lambda t} f_{\lambda}(\lambda)   d\lambda
=\mathbb{E}(e^{-\Lambda t}/\Lambda),
\end{align*}
where $\Lambda$ is a continuous random variable with density $f_{\lambda}$.
Also, 
\begin{align*}
	\int_{0}^{\infty}g(u)du
	&= \int_{0}^{\infty} \frac{1}{\lambda}  f_{\lambda}(\lambda)   d\lambda
	=\mathbb{E}(1/\Lambda).
\end{align*}
That is, we have 
\begin{align*}
	\Cor(Y_0,Y_t)=c(t)\frac{\mathbb{E}(e^{-\Lambda t}/\Lambda)}{\mathbb{E}(1/\Lambda)}.
	\end{align*}
\end{example}

\section{Simulation of periodic trawl processes}\label{sec:simulationalg}
We will now address the question of how a periodic trawl process can be simulated efficiently. A natural first choice might be to develop a grid-based method, using a suitable spatio-temporal grid. 
However, it turns out that we can adapt a \emph{slice-based} simulation algorithm which was developed for trawl processes, see \cite{Noven2016} and also \cite{LeonteVeraart2022}, and extend it to allow for any additional temporal kernel function, which in our case, will be a $\tau$-periodic function $p$.

The advantage of this slice-based approach is that we can have a coarser, but exact, approximation in the spatial domain, and the simulation error only appears through the discretisation in time, whereas a grid-based approach would result in a discretisation error both in the spatial and the temporal domain.

Suppose that we would like to simulate $Y$ on the grid $t_0, \ldots, t_n$ with $t_i=i\Delta$ for $i=0, \ldots, n$ for $\Delta>0$ and $t_n=n\Delta=T$. 

\subsection{Slice-based simulation for trawl processes}
The idea behind the slicing, see \cite{Noven2016,LeonteVeraart2022}, is that we consider the partition of $\cup_{i=0}^nA_{t_i}$, where $A_{t_i}=\{(x,s): s\leq t, 0\leq x \leq g(t-s)\}$, which is obtained when considering all the disjoint sets obtained from the intersections of the various trawl sets.
We illustrate this idea in Figure \ref{fig:slices} for the case when $n=4$ and spell out the mathematical details in Algorithm \ref{alg:slices}.

\begin{figure}[htbp]
\centering
\includegraphics[trim = 100 520 0 100, clip,scale=0.75]{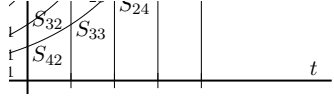}
\caption{Slices for $n=4$.}
\label{fig:slices}
\end{figure}

We create a $(n+1)\times (n+1)$-dimensional slice matrix $S=(S_{ij})$. Capital letters refer to slices themselves, which are Borel sets of finite Lebesgue measure,  and small letters $s=(s_{ij})$ to the Lebesgue measures of the associated slices ($s_{ij}=Leb(S_{ij})$). 
We will then draw independent random variables $L(S_{i,j})$ whose distribution has  cumulant function $s_{i,j}C(\theta;L')$
and simulate periodic trawl processes of the form
	$Y_t 
 =\int_{(-\infty,t]\times \R}p(t-s)\ind_{(0,g(t-s))}(x)L(dx,ds)$,
using the following approximation, for $k=0, \ldots, n$,
\begin{align*}
	Y_{k\Delta_n} & 
	=\sum_{j=1}^{k+1}p((k+2-j)\Delta_n)\sum_{i=k+2-j}^{n+2-j}L(S_{i,j}).
\end{align*}

We will now describe the simulation algorithm for periodic trawl processes in more detail by splitting the task at hand into three algorithms.

\subsubsection{Computing the matrix of slices}

We only consider slices $S_{ij}$ (or $s_{ij}$) with $j\leq n+2-i$ and store them in the matrix of slices of the form
\begin{align*}
    \begin{pmatrix}
S_{11}&S_{12}&\dots &S_{1n}& S_{1 (n+1)}\\
S_{21}&S_{22}&\dots& S_{2n}& \hrectangle\\
\vdots &\vdots & \iddots &\iddots&\vdots\\
S_{n1}&S_{n2}&\hrectangle \dots & \dots& \hrectangle\\
S_{(n+1)1}& \hrectangle & \dots & \dots & \hrectangle
    \end{pmatrix},
\end{align*}
see 
also Algorithm \ref{alg:slices}.

\begin{algorithm}[H]
\caption{Computing the slices $(s_{i,j})$, for $i\leq n+1, j\leq n+2-i$, using the trawl function $g$, $n$,  and the grid width $\Delta_n$ as inputs. \label{alg:slices}}
\begin{algorithmic}[1]
	\Procedure{Slices}{$g$, $n$, $\Delta_n$} 
	\State  $b \gets \mathrm{numeric}(n+1)$ \Comment{Create four vectors of zeros}
	\State  $c \gets \mathrm{numeric}(n)$
	\State  $d \gets \mathrm{numeric}(n+1)$ 
	\State  $e \gets \mathrm{numeric}(n)$
	\For{k in 1:(n+1)}
	\State $b[k] \gets \int_0^{\Delta_n}g(k\Delta_n-x)dx$ 
	\State $d[k] \gets \int_{-\infty}^{0}g(x-(k-1)\Delta_n)dx$
	\EndFor
	\For{k in 1:n}
	\State $c[k] \gets b[k]-b[k+1]$
	\State $e[k]=d[k]-d[k+1]$
	\EndFor
	\State  $s\gets \mathrm{matrix}(0, n+1, n+1)$ \Comment{Create slice matrix of zeros}
	\For { k in 1:n}
	\State $s[k, 1:(n+1-k)] \gets  \mathrm{replicate}(c[k], n+1-k)$
	\State $s[k,(n+1-k+1)] \gets b[k]$
	\EndFor
	\State $s[,1] \gets  c(e, d[n+1])$
	\Comment{Compute the column of first slices}
	\State \Return $s$ \Comment{Return the slice matrix}
	\EndProcedure
\end{algorithmic}
\end{algorithm}

\subsubsection{Adding the weighted slices}

\begin{algorithm}[H]
	\caption{Adding the weighted slices given a slicematrix $s=(s_{i,j})$, for $i\leq n+1, j\leq n+2-i$ and an $n$-dimensional weight vector $w$. \label{alg:addslices}}
	\begin{algorithmic}[1]
		\Procedure{AddWeightedSlices}{$s, w$}
		\State $n \gets \mathrm{nrow}(s)-1$
		\State $x \gets \mathrm{numeric}(n+1)$\Comment{Create vector of zeros}
		\State $tmp \gets 0$
		\For{k in 0:n}
		\State $tmp\gets 0$
		\For{j in 1:(k+1)}
		\State $tmp \gets tmp+w[k+2-j]\cdot \mathrm{sum}(s[(k+2-j):(n+2-j),j])$
		\EndFor
		\State $x[1+k]\gets tmp$
		\EndFor
		\State \Return $x$ 
		\EndProcedure
	\end{algorithmic}
\end{algorithm}

\subsubsection{Simulating a periodic trawl process}
\begin{algorithm}[H]
	\caption{Simulating a periodic trawl process on the grid $0, \Delta_n, \ldots, n \Delta_n$ given the functions $p, g$, the distribution of $L'$, and the grid width $\Delta_n$ and $n$.  \label{alg:sim}}
	\begin{algorithmic}[1]
		\Procedure{SimulatePeriodicTrawl}{$p, g, C(\cdot;L'), \Delta_n, n$}
		\State 	$s\gets \textsc{Slices}(g, n, \Delta_n)$ \Comment{Compute the slice-matrix $(s_{i,j})$}
		\State $L\gets \mathrm{matrix}(n+1,n+1)$ \Comment{Create matrix of r.v.s $L(S_{ij})$}
		\For {k in 1:n} 
		\State $L[k, 1:(n+1-k)]\gets$ vector of $(n+1-k)$ i.i.d.~r.v.s $\sim s[k,2]\cdot C(\cdot;L')$
		\State $L[k, (n+1-k+1)]\gets$ r.v.~$\sim s[k,(n+1-k+1)]\cdot C(\cdot;L')$
		\EndFor
		\For {k in 1:(n+1)}
		\State $L[k,1] \gets$ r.v.~$\sim s[k,1]\cdot C(\cdot;L')$
		\EndFor
		\State $w \gets \mathrm{numeric}(n+1)$ \Comment{Create weight vector}
		\For {k in 1:(n+1)}
		\State $w[k]\gets p(k \cdot \Delta_n)$
		\EndFor
		\State $y \gets  \textsc{AddWeightedSlices}(L, w)$
		\State \Return $y$ \Comment{Returns $Y_0, Y_{\Delta_n}, \ldots, Y_{n\Delta_n}$}
		\EndProcedure
	\end{algorithmic}
\end{algorithm}

We note that in the above algorithm, Algorithm \ref{alg:sim}, the notation "r.v.~$\sim s[k,(n+1-k+1)]\cdot C(\cdot;L')$" is a short-hand notation for a random variable whose distribution is characterised by the cumulant function $s[k,(n+1-k+1)]\cdot C(\cdot;L')$.

We note that the approximation for the first column of the slice matrix is rather rough. Hence, in practice, we should use a burn-in period in the simulation to minimise the effect of the initial approximation error.

The simulation algorithm has been implemented in the R package \texttt{ambit} available on CRAN, see \cite{ambit}. 

\subsection{A note on stochastic versus deterministic seasonality}
We note that periodic trawl processes are stationary with stochastic seasonality, which is reflected in the time-invariant mean, variance, and autocorrelation, where the autocorrelation function is a periodic function.

One might wonder how such periodic trawl processes compare to trawl processes with deterministic seasonality. 
As before, let $L$ denote a homogeneous \Levy basis with characteristic triplet $(\drift, a, \lev)$, $p:[0, \infty)\to \R$ a periodic function with period $\tau>0$ and $g:[0, \infty)\to \R$ a continuous, monotonically decreasing function. 

Consider the trawl process 
\begin{align*}
	X_t =\int_{(-\infty,t]\times \R}\ind_{(0,g(t-s))}(x)L(dx,ds)
\end{align*}
and
the periodic trawl process  
\begin{align*}
	Y_t =\int_{(-\infty,t]\times \R}p(t-s)\ind_{(0,g(t-s))}(x)L(dx,ds).
\end{align*}
Let $q:[0, \infty)\to \R$ a periodic function  with period $\tau>0$.
We consider trawl processes with additive and multiplicative seasonality.
\begin{align*}
	X^a(t):=q(t)+X_t, && X^m(t):=q(t)X_t.
\end{align*}
Then $X^a, X^m$ are not stationary and their second-order properties are given by
\begin{align*}
	\E(X^a_t) &= q(t)+\E(L')\int_0^{\infty}g(u)du
	=
	q(t)+\E(X_0)\\
	\Var(X^a_t)
	&= \Var(L')\int_{0}^{\infty}g(u)du=\Var(X_0),
	\\
	\Cov(X^a_{0},X^a_{t})
	&= \Var(L')\int_{0}^{\infty}g(t+u)du=\Cov(X_{0},X_{t}),
	\\
	\Cor(X^a_{0},X^a_{t})
	&= \frac{\int_{0}^{\infty}g(t+u)du}{\int_{0}^{\infty}g(u)du}=\Cor(X_{0},X_{t}),
\end{align*} 
for $t\geq 0$.
Also,  
\begin{align*}
 \E(X^m_t)&=q(t)\E(L')\int_0^{\infty}g(u)du=q(t)\E(X_0),\\
	\Var(X^m_t)
	&= q^2(t)\Var(L')\int_{0}^{\infty}g(u)du=q^2(t)\Var(X_0),\\
	\Cov(X^m_{0},X^m_{t})
	&= q(0)q(t)\Var(L')\int_{0}^{\infty}g(t+u)du=q(0)q(t)\Cov(X_{0},X_{t}),
	\\
		\Cor(X^m_{0},X^m_{t})
&= \frac{q(0)q(t)\int_{0}^{\infty}g(t+u)du}{q^2(t)\int_{0}^{\infty}g(u)du}=\frac{q(0)}{q(t)}\Cor(X_{0},X_{t}).	
\end{align*} 

\begin{example}
	Consider the case where $L'\sim \mathrm{N}(0,1)$, 
	$p(x)=q(x)=\sin(2\pi x/\tau)$, for $\tau=3$,
	$g(x)=\exp(-0.5x)$.
	We simulate the processes $X$, $X^a$, $X^m$, $Y$ and the function  $q$ 
	on the grid $t_i=i\Delta_n$, for $i=0,\ldots, n=499$ and 
	$\Delta_n=0.1$.

\begin{figure}[htbp]
	\captionsetup[subfigure]{aboveskip=-4pt,belowskip=-4pt}
	\subfloat[Sample path of $X$	\label{fig:X}]{	\includegraphics[width=0.45\textwidth, height=0.1\textheight]{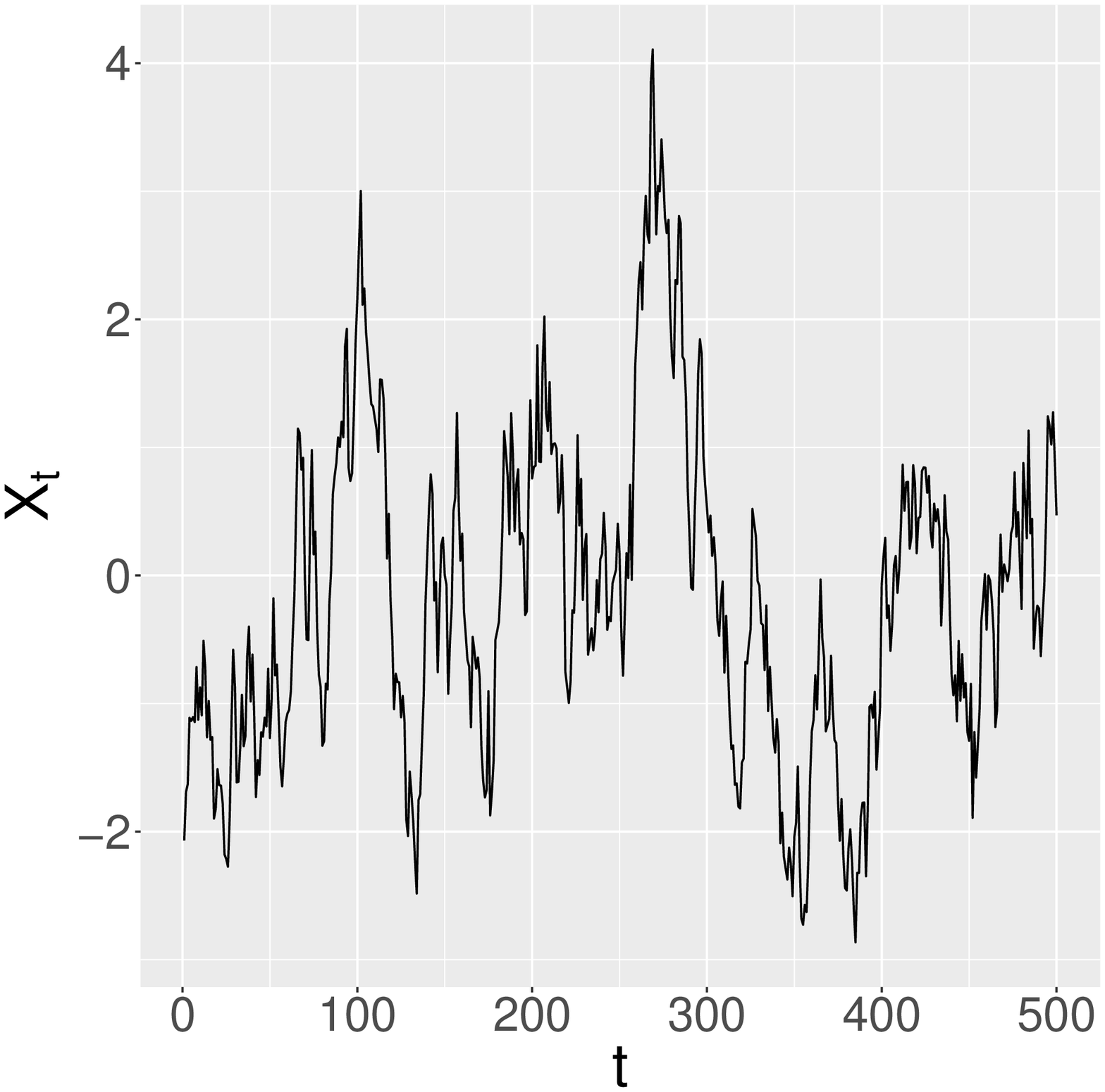} } 
	\captionsetup[subfigure]{aboveskip=-4pt,belowskip=-4pt}
	\subfloat[Empirical ACF of $X$\label{fig:ACF_X}]{	\includegraphics[width=0.45\textwidth, height=0.1\textheight]{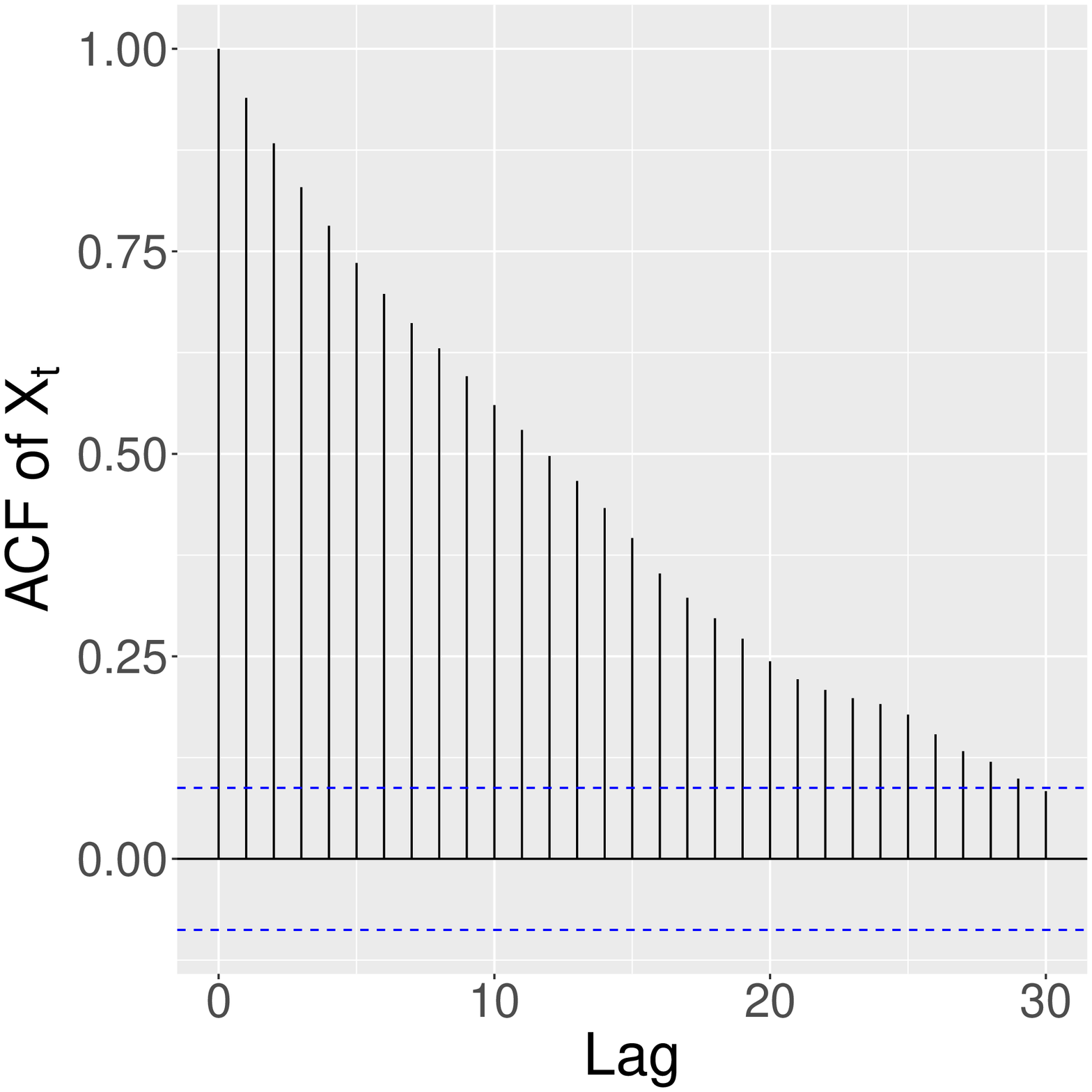}} 	\\
		\subfloat[Sample path of $X^a$	\label{fig:Xa}]{	\includegraphics[width=0.45\textwidth, height=0.1\textheight]{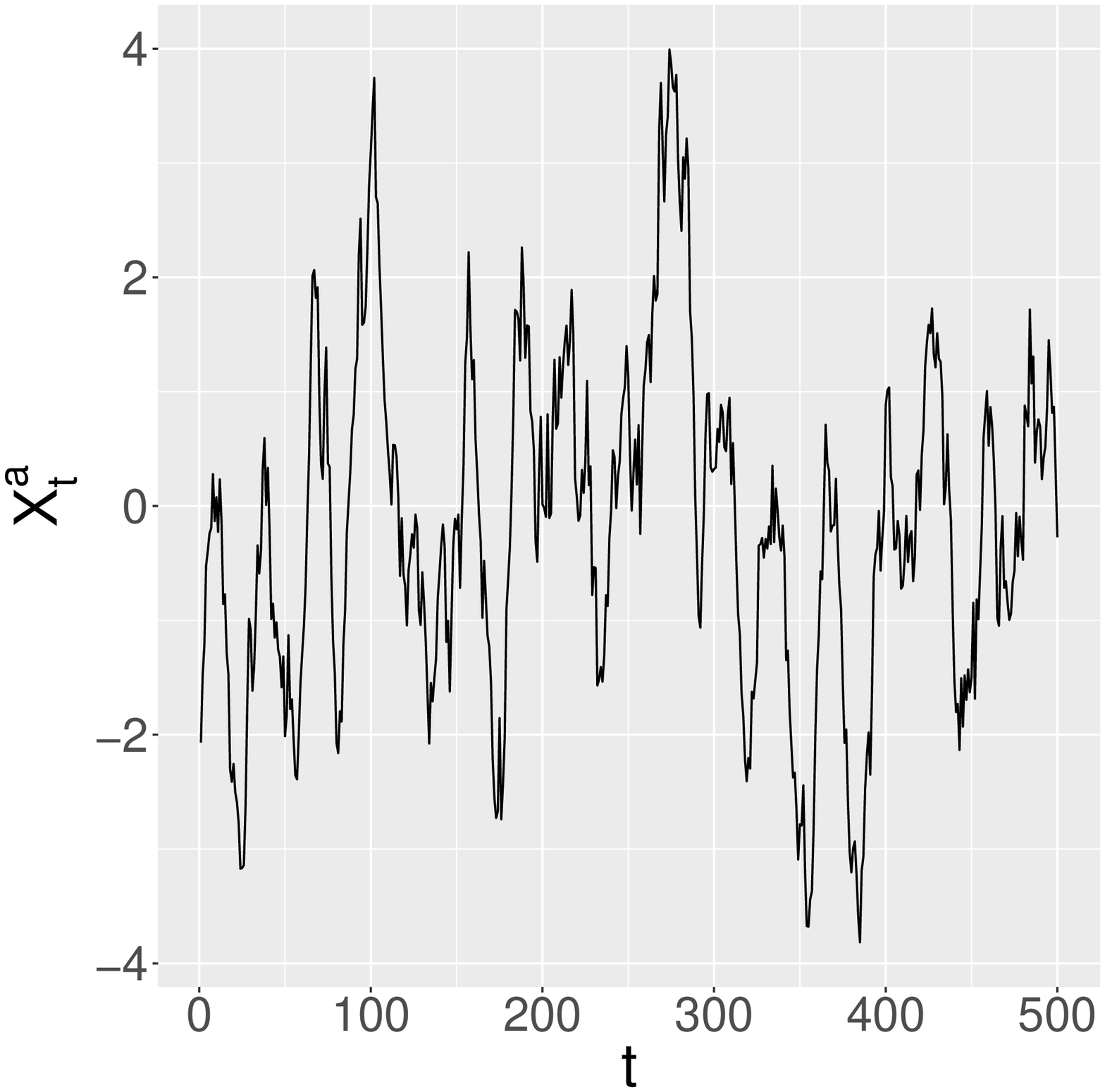} } 
	\captionsetup[subfigure]{aboveskip=-4pt,belowskip=-4pt}
	\subfloat[Empirical ACF of $X^a$\label{fig:ACF_Xa}]{	\includegraphics[width=0.45\textwidth, height=0.1\textheight]{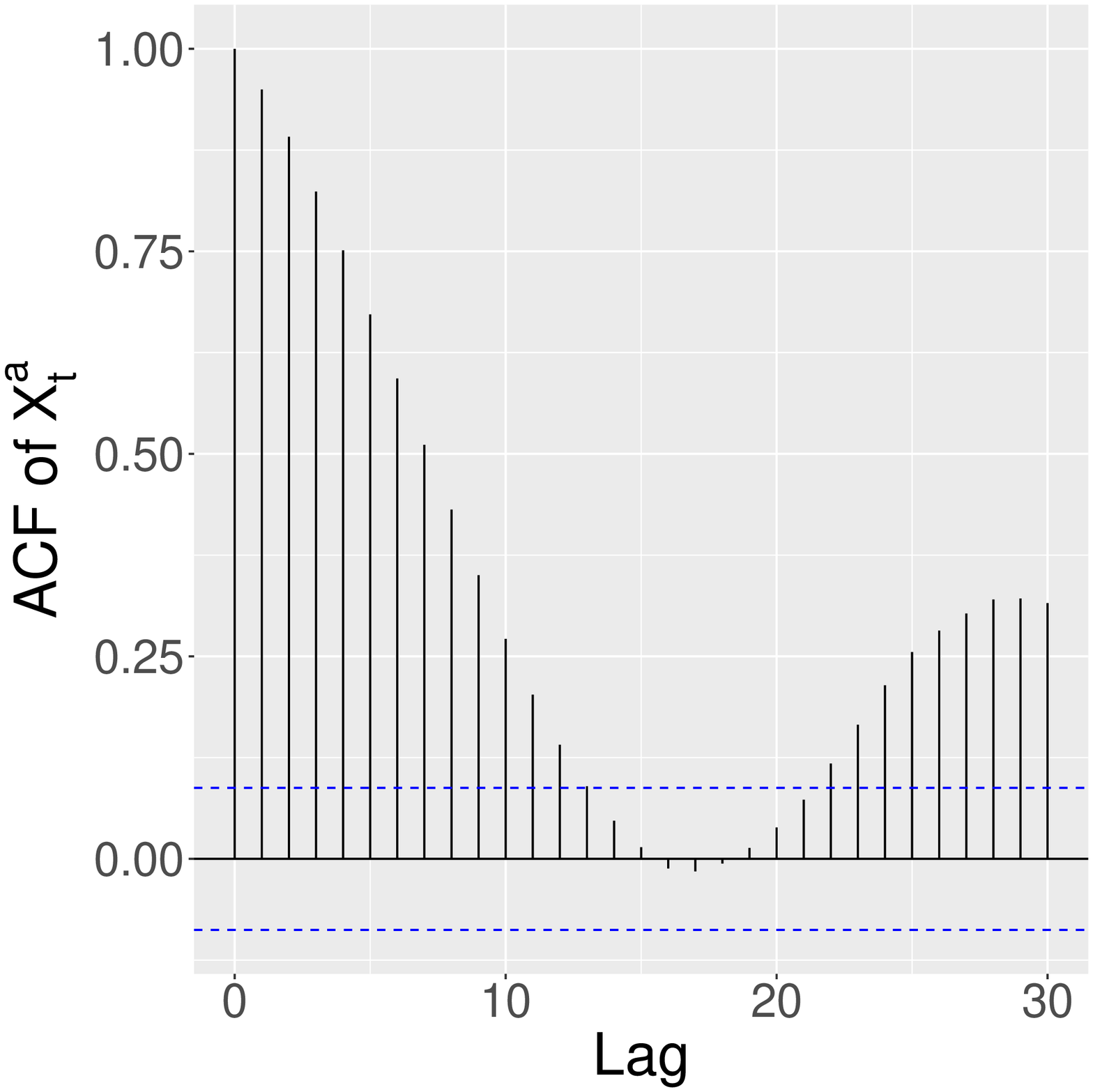}} 	\\
		\subfloat[Sample path of $X^m$	\label{fig:Xm}]{	\includegraphics[width=0.45\textwidth, height=0.1\textheight]{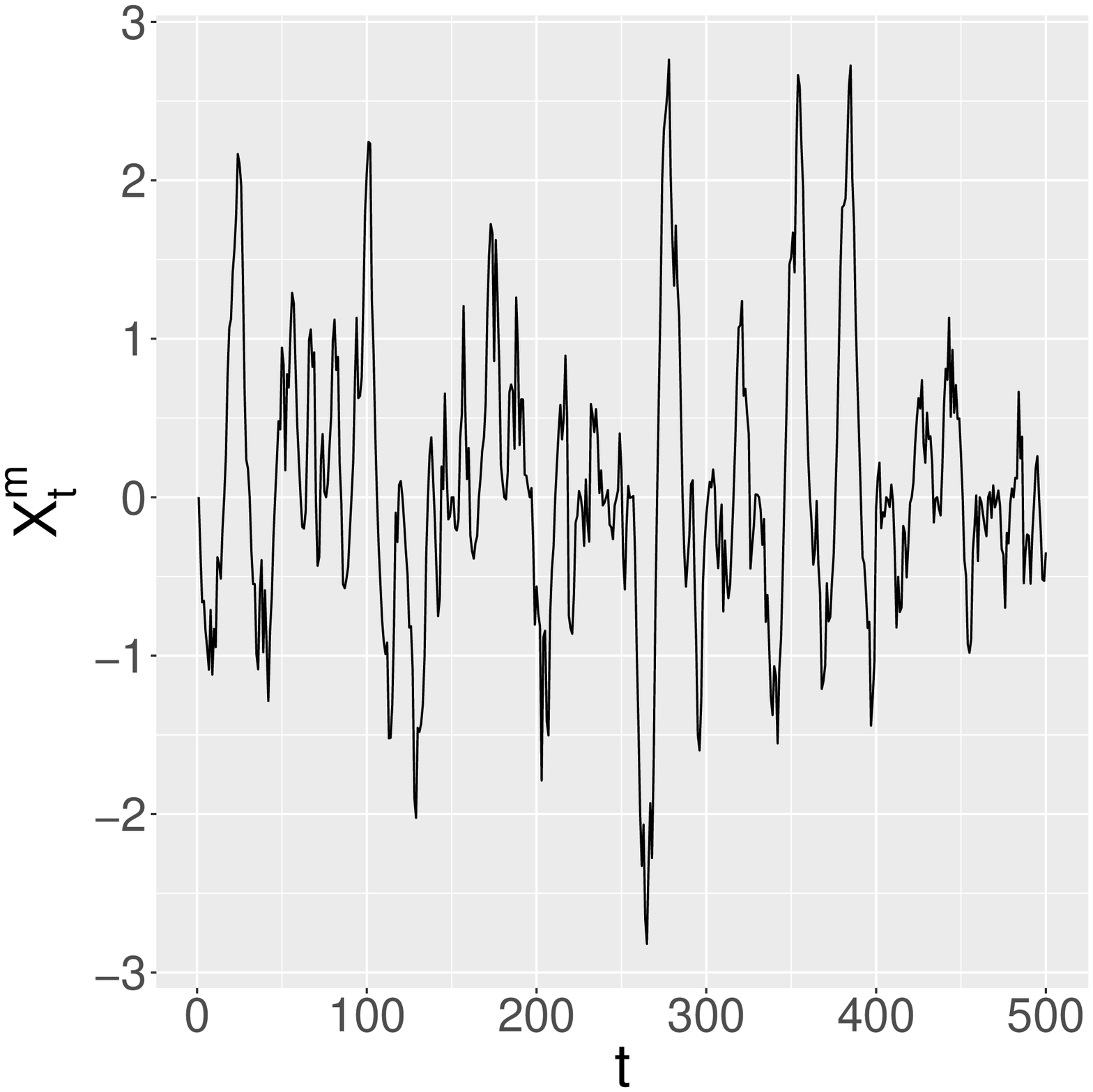} } 
	\captionsetup[subfigure]{aboveskip=-4pt,belowskip=-4pt}
	\subfloat[Empirical ACF of $X^m$\label{fig:ACF_Xm}]{	\includegraphics[width=0.45\textwidth, height=0.1\textheight]{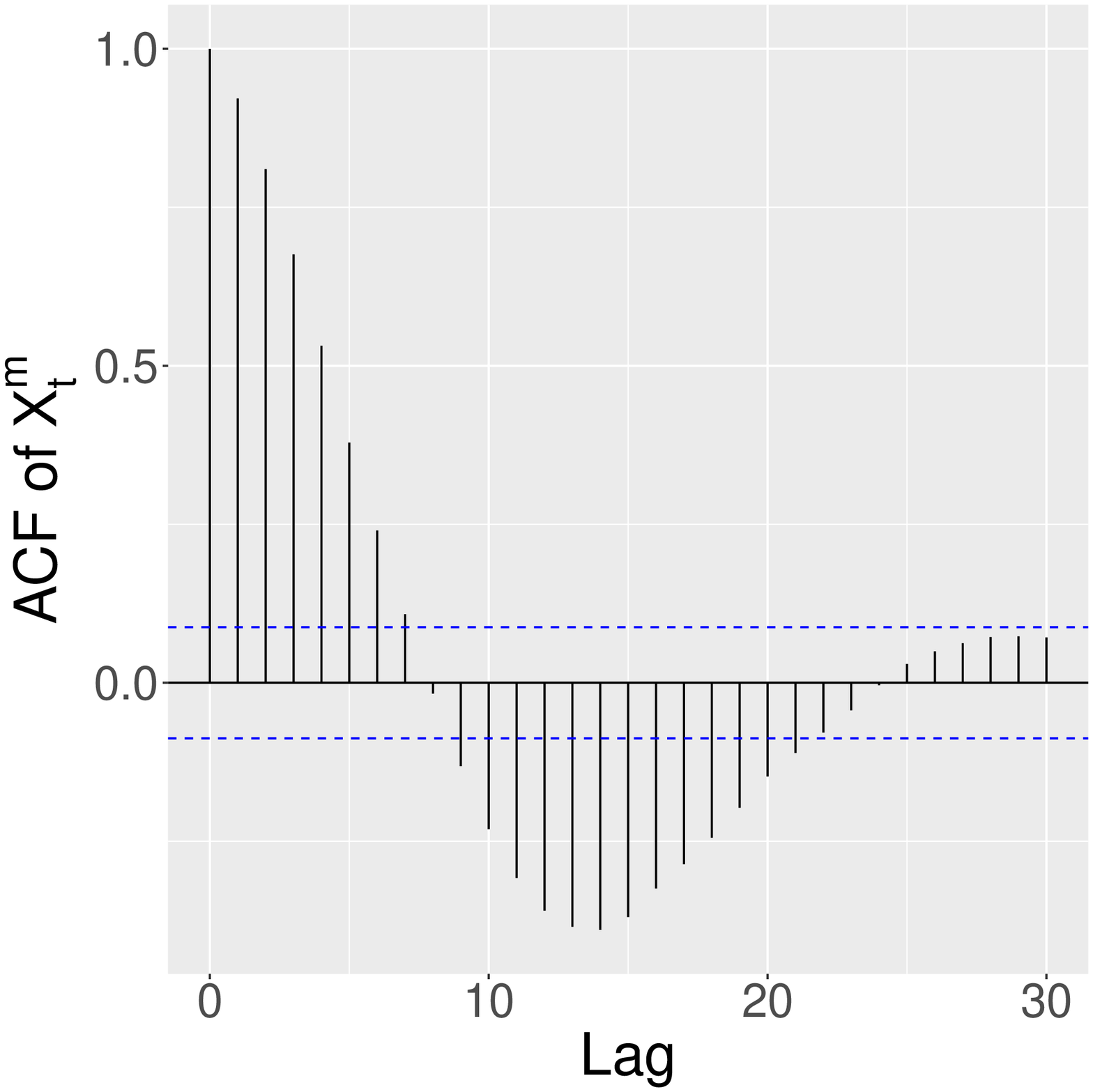}} 	\\
		\subfloat[	Sample path of $Y$\label{fig:Y}]{	\includegraphics[width=0.45\textwidth, height=0.1\textheight]{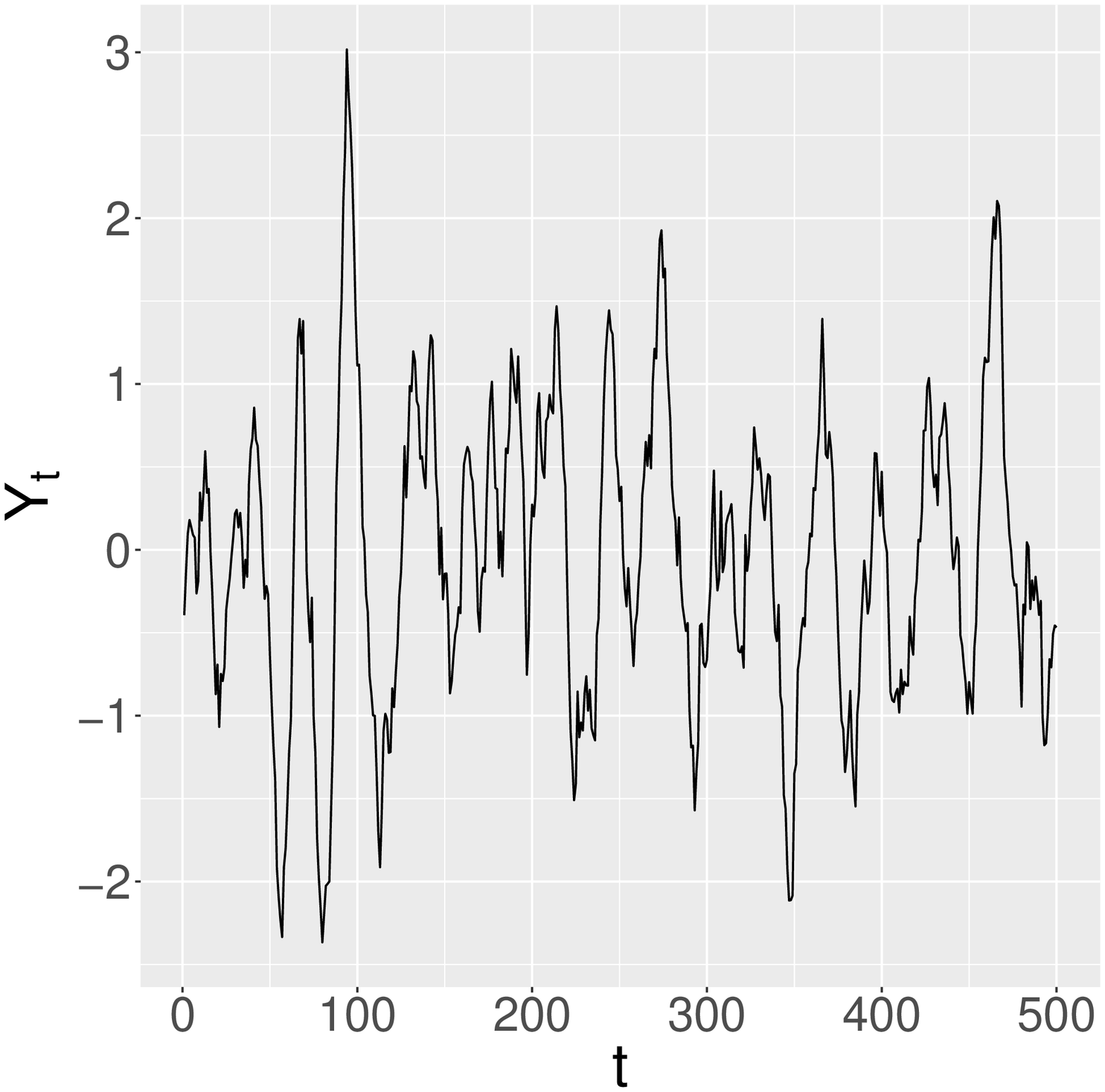} } 
	\captionsetup[subfigure]{aboveskip=-4pt,belowskip=-4pt}
	\subfloat[Empirical ACF of $Y$\label{fig:ACF_Y}]{	\includegraphics[width=0.45\textwidth, height=0.1\textheight]{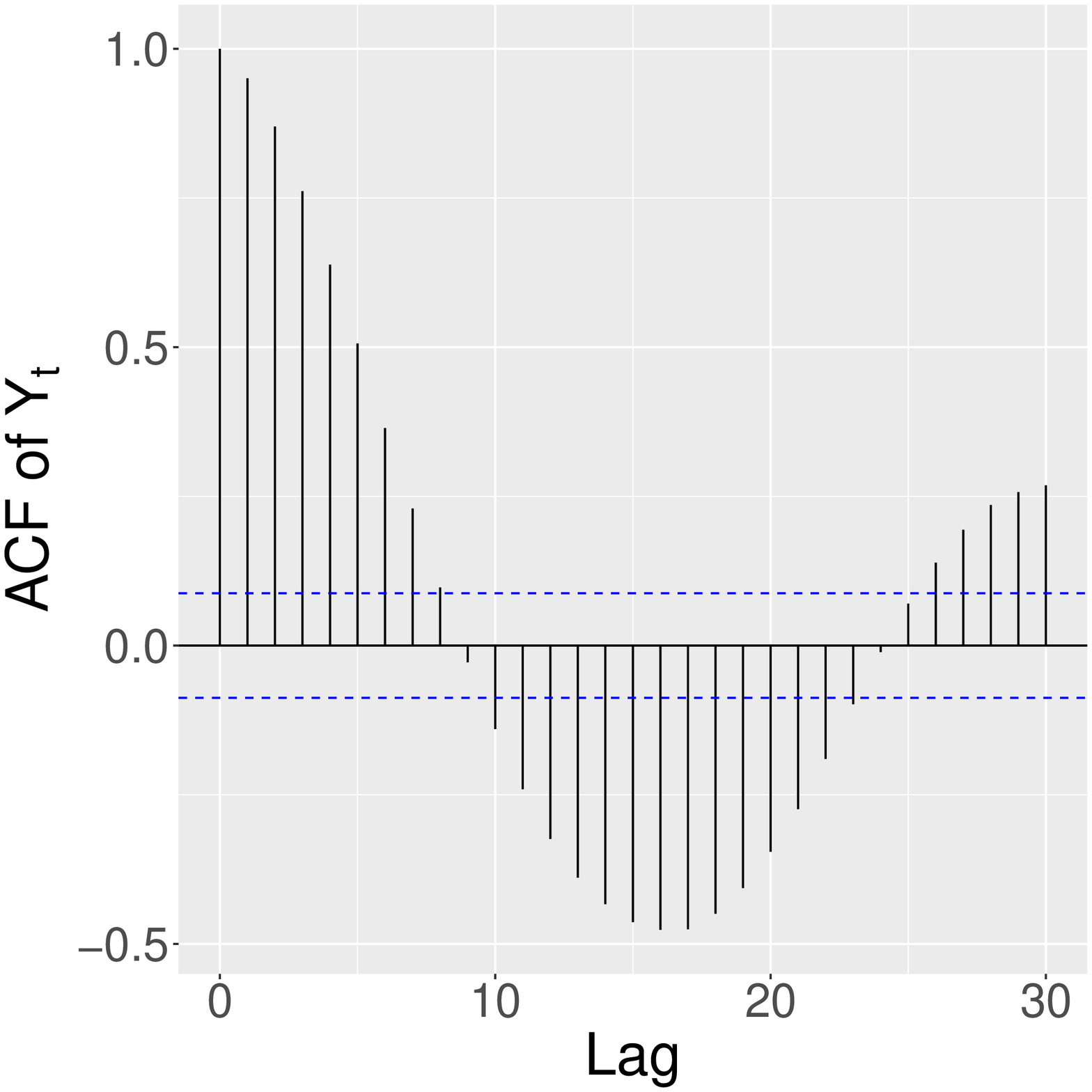}} 	\\
		\subfloat[Function $q()$	\label{fig:q}]{	\includegraphics[width=0.45\textwidth, height=0.1\textheight]{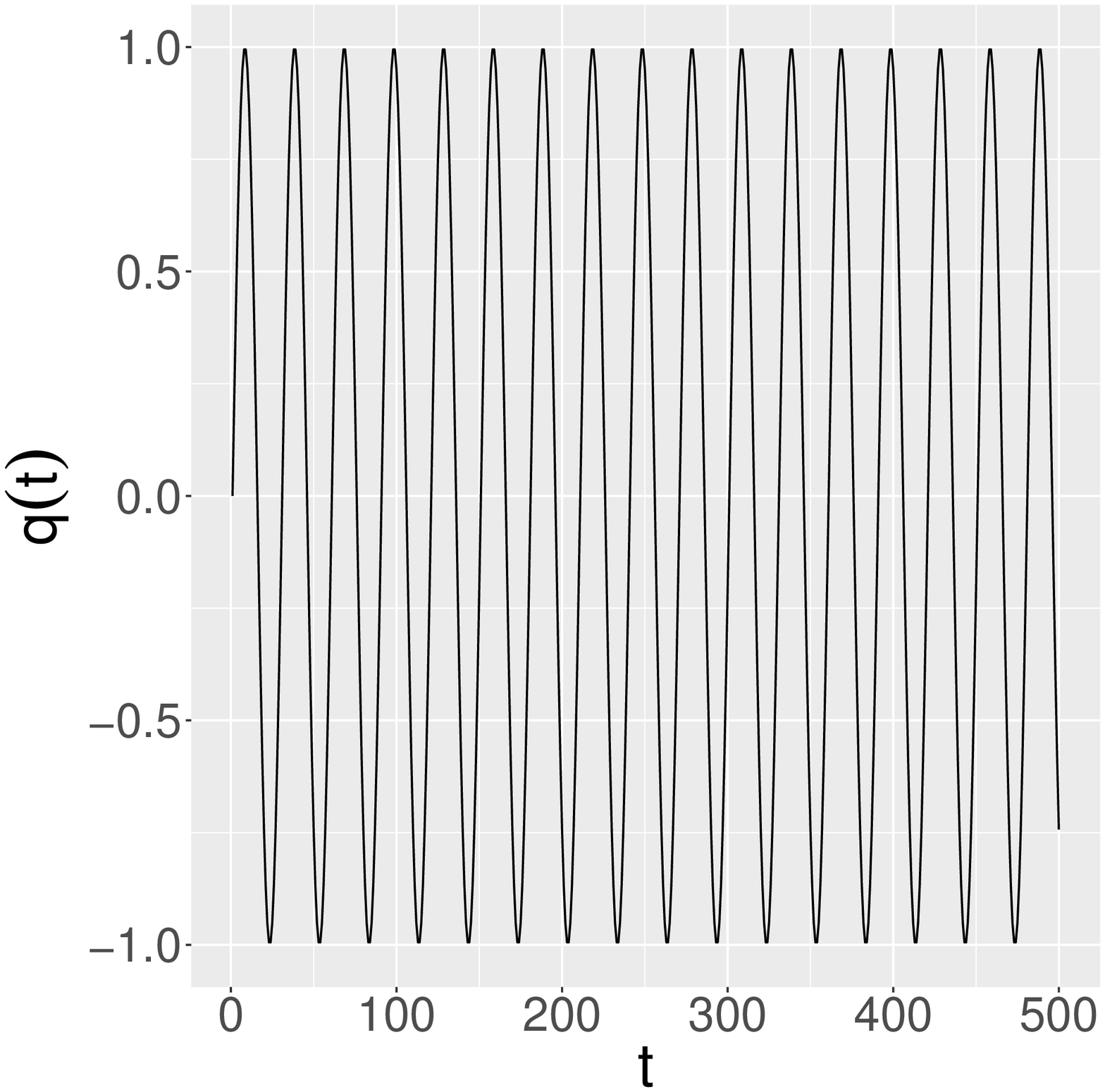} } 
	\captionsetup[subfigure]{aboveskip=-4pt,belowskip=-4pt}
	\subfloat[Empirical ACF of $q()$\label{fig:ACF_q}]{	\includegraphics[width=0.45\textwidth, height=0.1\textheight]{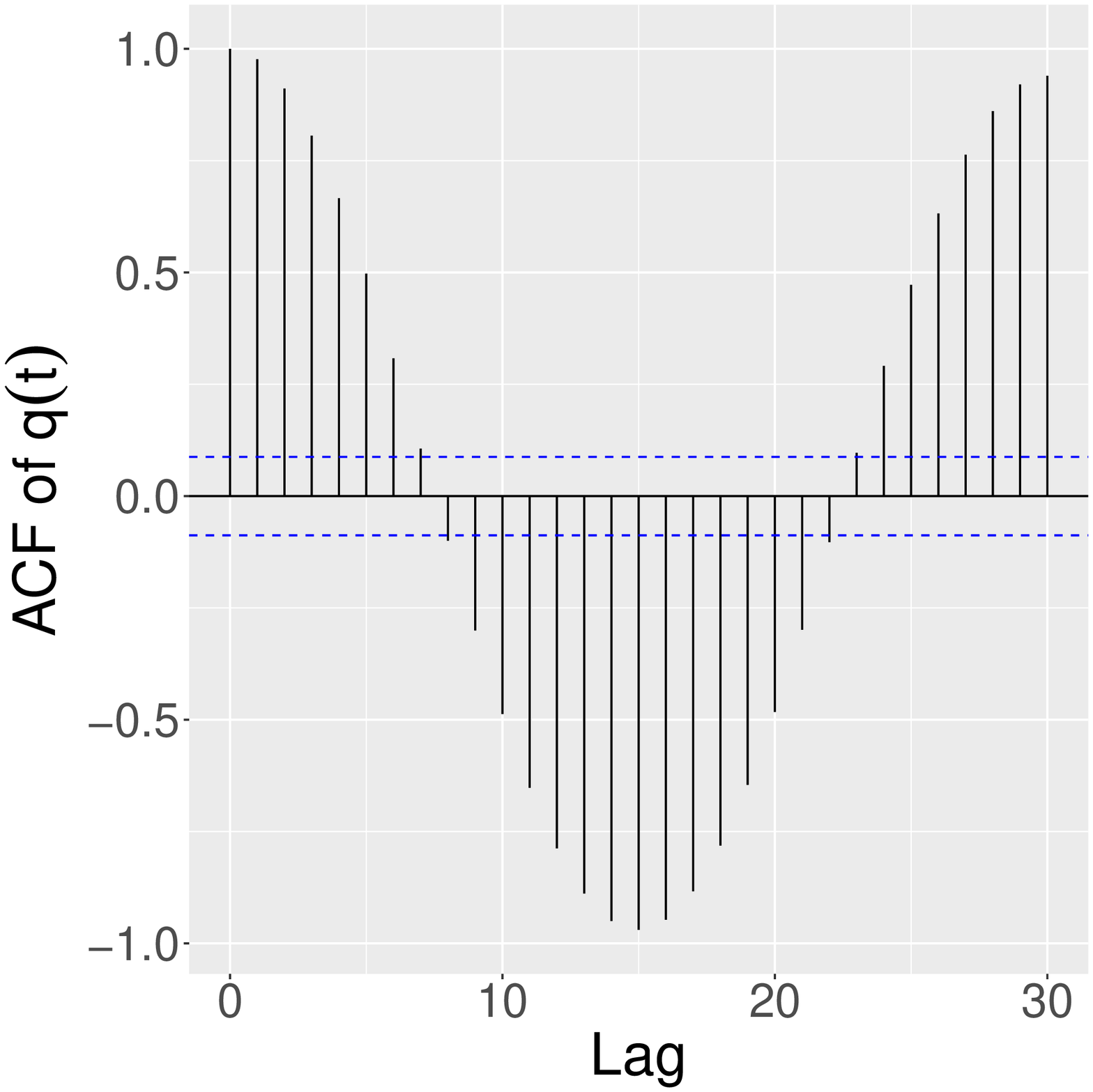}} 	\\
	\caption{\it    
		Comparing the sample paths and empirical autocorrelation functions of the  stationary trawl process $X$, the non-stationary processes $X^a$ and $X^m$, the stationary periodic trawl process $Y$ and the seasonality function $q$.}
	\label{fig:DetvsStochSeas}
\end{figure}

 We visually compare the sample paths and empirical autocorrelation functions of the stationary trawl process $X$, the nonstationary processes $X^a$ and $X^m$, the stationary periodic trawl process $Y$ and the seasonality function $q$ in Figure \ref{fig:DetvsStochSeas}, see \cite{PeriodicTrawl-Energy} for the corresponding {\tt R} code. Here we consider Gaussian processes, and the stochastic noise terms in each simulation of the various paths were kept the same to simplify the comparison. 
 
 We note that the sample paths and empirical autocorrelation functions of  $Y$ and $X^m$ look very similar, which suggests that it might be hard to distinguish between these two models in practice. Note however that the key probabilistic difference between the two processes is that $Y$ is stationary whereas $X^m$ is not.

  \end{example}

\section[Asymptotic theory for MMAPs]{Asymptotic theory for the sample mean, the sample autocovariances and the sample autocorrelations of mixed moving average processes}\label{sec:asymptotictheory}

Since MMAPs and hence also periodic trawl processes are mixing and, hence, ergodic, see \cite{FuchsStelzer2013}, we know that their corresponding moment estimators are consistent. 
In order to develop a suitable estimation theory for (periodic) trawl processes, we first derive the results in the more general setting of MMAPs.

Consider an MMAP given by 
 $Y=(Y_t)_{t\in \mathbb{R}}$ for 
\begin{align*}
Y_t =\mu +\int_{\mathbb{R}\times \mathbb{R}}f(x, t-s)L(dx, ds),
\end{align*}
for $\mu \in \R$, where $f$ satisfies the integrability conditions stated in \eqref{eq:intcondMA} almost surely. 

We are interested in the asymptotic behaviour of the sample mean and sample autocovariances/-correlations for such processes.

Note that for all $t \in \R$,
\begin{align*}
    \mathbb{E}(Y_t)=\mu + \mathbb{E}(L') \int_{\mathbb{R}\times \mathbb{R}}f(x, u)dx du,
\end{align*}
and the  autocovariance function of $Y$ is denoted by
\begin{align*}
\gamma(h)=\gamma_f(h)=\Cov(Y_0, Y_h)=\kappa_2\int_{\R \times \R}f(x, -s)f(x, h-s)dx ds,
\end{align*}
for any $h \in \R$, where $\kappa_2 =\Var(L')$.

Suppose that the process $Y$ is sampled over a fixed-time grid of width $\Delta>0$ at times $(n\Delta)_{n \in \mathbb{N}}$.

The proofs of the following results are presented in the Appendix, see \ref{ap:proofsgentheory}.
\subsection{Asymptotic normality of the sample mean}

We denote the sample mean  by
\begin{align*}
\overline{Y}_{n; \Delta}:=\frac{1}{n}\sum_{i=1}^nY_{i\Delta}.
\end{align*}
By adapting \citet[Theorem 2.1]{CohenLindner2013} to our more general setting, we get the following asymptotic result for the sample mean.

\begin{theorem}\label{thm:samplemean}
Suppose that $\mathbb{E}(L')=0, \kappa_2=\Var(L')<\infty, \mu \in \mathbb{R}$ and $\Delta >0$. 
Further, assume that 
\begin{align}\label{as:sumf2}
\left(F_{\Delta}:\mathbb{R}\times [0, \Delta] \to [0, \infty], (x, u) \mapsto F_{\Delta}(x,u)=\sum_{j=-\infty}^{\infty}|f(x, u+j\Delta)|\right) \in L^2(\R \times [0, \Delta]). 
\end{align}
Then $\sum_{j=-\infty}^{\infty} |\gamma(\Delta j)| < \infty$,
\begin{align}\label{as:sumgamma}
V_{\Delta}:=\sum_{j=-\infty}^{\infty}\gamma( \Delta j) = \kappa_2 \int_{\R \times [0, \Delta]} \left(\sum_{j=1}^{\infty}f(x, u+j\Delta)\right)^2 dx du, 
\end{align}
and the sample mean of $Y_{\Delta i}$, for $i=1, \ldots, n$, is asymptotically Gaussian as $n \to \infty$, i.e.
\begin{align*}
\sqrt{n}\left(\overline{Y}_{n; \Delta} - \mu \right)
\stackrel{\mathrm{d}}{\to} \mathrm{N}\left(0, V_{\Delta}\right), \quad \mathrm{as} \, n \to \infty.
\end{align*}
\end{theorem}

\begin{remark}
As in \cite{CohenLindner2013}, we remark that the assumption that
$\mathbb{E}(L')=0$ can be removed if the kernel function $f$ satisfies $f \in L^1(\R \times \R)\cap L^2(\R \times \R)$.
\end{remark}

\subsection{Asymptotic normality of the sample autocovariance}
We define the sample autocovariance function as
\begin{align*}
\widehat{\gamma}_{n;\Delta}(\Delta h):= n^{-1}\sum_{j=1}^{n-h}(Y_{j\Delta}-\overline{Y}_{n,\Delta})(Y_{(j+h)\Delta}-\overline{Y}_{n,\Delta}), \quad h \in \{0, \ldots, n-1\},
\end{align*}
and the sample autocorrelation function as
\begin{align*}
\widehat{\rho}_{n;\Delta}(\Delta h)=\frac{\widehat{\gamma}_{n;\Delta}(\Delta h)}{\widehat{\gamma}_{n;\Delta}(0)}, \quad h \in \{0, \ldots, n-1\}.
\end{align*}

In the case when $\mu=0$ (and we assumed that $\mathbb{E}(L')=0$), one could use the following simpler estimators
\begin{align*}
\widehat{\gamma}_{n;\Delta}^*(\Delta h):= n^{-1}\sum_{j=1}^{n}Y_{j\Delta}Y_{(j+h)\Delta}, \quad h \in \{0, \ldots, n-1\},
\end{align*}
and 
\begin{align*}
\widehat{\rho}_{n;\Delta}^*(\Delta h)=\frac{\widehat{\gamma}_{n;\Delta}^*(\Delta h)}{\widehat{\gamma}_{n;\Delta}^*(0)}, \quad h \in \{0, \ldots, n-1\}.
\end{align*}

Consider the case where $\mu=0$, for which we will study the limiting behaviour of 
$\Cov(\widehat{\gamma}_{n,\Delta}^*(\Delta p), \widehat{\gamma}_{n,\Delta}^*(\Delta q))$. For this, we need a formula for the fourth (joint) moments of the mixed moving average process, which we derive in the following lemma.

\begin{lemma}\label{lem:fourthmoment}
Let $L'$ be a non-zero L\'{e}vy seed with $\mathbb{E}(L')=0$ and $\mathbb{E}(L'^4)<\infty$. Let $\kappa_2:=\mathbb{E}(L'^2)=\Var(L')$, $\eta:= \mathbb{E}(L'^4)/\kappa_2^2$ and $\kappa_4:=(\eta-3)\kappa_2^2$. Let $\Delta>0$, and assume that $f \in L^2(\R\times \R)\cap L^4(\R \times \R)$.
Then, for $t_1, t_2, t_3, t_4 \in \R$, we have
\begin{align*}
&\mathbb{E}(Y_{t_1}Y_{t_2}Y_{t_3}Y_{t_4})\\
&= \gamma(t_1-t_2)\gamma(t_3-t_4) 
+
\gamma(t_1-t_3) \gamma(t_2-t_4)
+\gamma(t_1-t_4)\gamma(t_2-t_3) \\
& \quad + \kappa_4
\int_{\R \times \R} f(x, t_1-t_3+s) f(x, t_2-t_3)
f(x, s) f(x, t_4-t_3+s) dx ds.
\end{align*}
\end{lemma}

The following result is an extension of \citet[Proposition 3.1]{CohenLindner2013}.
\begin{proposition}\label{prop:covlimit}
Let $\mu=0$. Let $L'$ be a non-zero L\'{e}vy seed with $\mathbb{E}(L')=0$ and $\mathbb{E}(L'^4)<\infty$. Let $\kappa_2:=\mathbb{E}(L'^2)=\Var(L')$ and $\eta:= \mathbb{E}(L'^4)/\kappa_2^2$. Let $\Delta>0$, and assume that $f \in L^2(\R\times \R)\cap L^4(\R \times \R)$ and
\begin{align}\label{as:sumf2b}
\left(\mathbb{R}\times [0, \Delta] \to \R, (x, u) \mapsto \sum_{j=-\infty}^{\infty}f^2(x, u+j\Delta)\right) \in L^2(\R \times [0, \Delta]). 
\end{align}
For $q \in \mathbb{Z}$, define the function
\begin{align}\label{as:sumf2c}
\left(G_{q;\Delta}:\mathbb{R}\times [0, \Delta] \to \R, (x, u) \mapsto G_{q;\Delta}(x,u)=\sum_{j=-\infty}^{\infty}f(x, u+j\Delta)f(x, u+(j+q)\Delta)\right),
\end{align}
which is in $L^2(\R \times [0, \Delta])$ 
due to \eqref{as:sumf2b}.

Further, assume that 
\begin{align}\label{eq:sumgamma2-cov}
\sum_{j=-\infty}^{\infty}|\gamma(\Delta j)|^2<\infty.
\end{align}
Then, for each $p, q \in \mathbb{N}$, we have
\begin{align}\begin{split}\label{eq:sumT}
&\lim_{n\to \infty}n\Cov(\widehat{\gamma}_{n;\Delta}^*(\Delta p), \widehat{\gamma}_{n;\Delta}^*(\Delta q)) =v_{pq;\Delta},\\
v_{pq;\Delta}&:=(\eta-3)\kappa_2^2\int_{\R\times[0,\Delta]}G_{p;\Delta}(x,u)G_{q;\Delta}(x,u)dxdu \\
&+\sum_{l=-\infty}^{\infty}[\gamma(l\Delta)\gamma((l+p-q)\Delta)
+
\gamma((l-q)\Delta)\gamma((l+p)\Delta)
].
\end{split}
\end{align}
\end{proposition}

We now state a joint central limit theorem for the sample autocovariance, autocorrelations and their counterparts when $\mu=0$.

\begin{theorem}\label{thm:acflimits}
\begin{enumerate}
    \item Suppose that the assumptions of Proposition \ref{prop:covlimit} hold and that, in addition,
    \begin{align}\label{as:finsumsquint}
    \sum_{j=-\infty}^{\infty}\left(\int_{\R \times \R } |f(x, u)||f(x, u+j\Delta)| dxdu\right)^2<\infty.
    \end{align}
    Then, for each $h\in \mathbb{N}$, we have
    \begin{align*}
        \sqrt{n}(\gamma_{n;\Delta}^*(0)-\gamma(0), \ldots, \gamma_{n;\Delta}^*(h\Delta)-\gamma(h\Delta))^{\top} \stackrel{\mathrm{d}}{\to}\mathrm{N}(0,V_{\Delta}), \quad n \to \infty,
            \end{align*}
            where the asymptotic covariance matrix is given by $V_{\Delta}=(v_{pq, \Delta})_{p,q=0,\ldots,h}\in \R^{h+1,h+1}$ with $v_{pq; \Delta}$ defined as in   \eqref{eq:sumT}.

\item Suppose that the same assumptions as in 1.) hold and further assume that the function 
$(x,u)\mapsto \sum_{j=-\infty}^{\infty}|f(x,u+j\Delta)|$ is in $L^2(\R \times [0,\Delta])$.
Then, for each $h\in \mathbb{N}$, we have
    \begin{align*}
        \sqrt{n}(\widehat{\gamma}_{n;\Delta}(0)-\gamma(0), \ldots, \widehat{\gamma}_{n;\Delta}(h\Delta)-\gamma(h\Delta))^{\top} \stackrel{\mathrm{d}}{\to}\mathrm{N}(0,V_{\Delta}), \quad n \to \infty,
            \end{align*}
            where the asymptotic covariance matrix is given by $V_{\Delta}=(v_{pq; \Delta})_{p,q=0,\ldots,h}\in \R^{h+1,h+1}$ with $v_{pq; \Delta}$ defined as in   \eqref{eq:sumT}.
\item Suppose that the same assumptions as in 1.) hold and that $f$ is not almost everywhere equal to  zero.
  Then, for each $h\in \mathbb{N}$, we have
    \begin{align*}
        \sqrt{n}(\rho^*_{n;\Delta}(\Delta)-\rho(\Delta), \ldots, \rho^*_{n;\Delta}(h\Delta)-\rho(h\Delta))^{\top} \stackrel{\mathrm{d}}{\to}\mathrm{N}(0,W_{\Delta}), \quad n \to \infty,
            \end{align*}
            where the asymptotic covariance matrix is given by $W_{\Delta}=(w_{pq; \Delta})_{p,q=0,\ldots,h}\in \R^{h,h}$ with 
\begin{align*}
    w_{pq;\Delta}&=(v_{pq;\Delta}-\rho(p\Delta)v_{0q;\Delta}
-\rho(q\Delta)v_{p0;\Delta}
+\rho(p\Delta)\rho(q\Delta)v_{00;\Delta})/\gamma^2(0)\\
&=
\frac{(\eta-3)\kappa_2^2}{\gamma^2(0)}  \int_{\R \times [0,\Delta]}
  (G_{p;\Delta}(x,u)-G_{0;\Delta}(x,u)\rho(p\Delta))\\
  &\quad \cdot
  (G_{q;\Delta}(x,u)-G_{0;\Delta}(x,u)\rho(q\Delta))dxdu\\
  &+\sum_{l=-\infty}^{\infty}
 \left[
 \rho((l+q)\Delta)\rho((l+p)\Delta)+\rho((l-q)\Delta)\rho((l+p)\Delta) \right.\\
 &
 -2\rho((l+q)\Delta)\rho(l\Delta)\rho(p\Delta)
 -2\rho(l\Delta)\rho((l+p)\Delta)\rho(q\Delta)\\
 &\left. 
 +2\rho(p\Delta)\rho(q\Delta)\rho^2(l\Delta)\right].
\end{align*}  
If, in addition, the function 
$(x,u)\mapsto \sum_{j=-\infty}^{\infty}|f(x,u+j\Delta)|$ is in $L^2(\R \times [0,\Delta])$, then we have
\begin{align*}
        \sqrt{n}(\widehat{\rho}_{n,\Delta}(\Delta)-\rho(\Delta), \ldots, \widehat{\rho}_{n,\Delta}(h\Delta)-\rho(h\Delta))^{\top} \stackrel{\mathrm{d}}{\to}\mathrm{N}(0,W_{\Delta}), \quad n \to \infty,
            \end{align*}
            with $W_{\Delta}$ defined as above.
\end{enumerate}
\end{theorem}

\begin{remark}
Note that the assumptions used in the above theorems typically rule out long-memory settings. We discuss the detailed implications in the context of (periodic) trawl processes in the appendix, see Section \ref{sec:asscheck}. 
\end{remark}

\section{Inference for periodic trawl processes using methods of moments}\label{sec:inference}
Based on two working examples, we will now develop an estimation and inference methodology for the parameters determining the serial correlation of (periodic) trawl processes using a method-of-moments approach and present the corresponding asymptotic theory. Our methodology can in principle be applied to general trawl functions, but in order to simplify the exposition, we will concentrate on two particularly relevant specifications, namely the exponential trawl and the  supGamma trawl  function.

\subsection{Exponential trawl function}
Recall that, in the case of an exponential trawl function $g(x)=\exp(-\lambda x), \lambda >0, x\geq 0$, the autocorrelation is, for all $t>0$,  given by  
\begin{align*}
\rho(t)=\Cor(Y_0, Y_t)=c(t) \exp(-\lambda t), 
\end{align*}
which is equivalent to
\begin{align*}
\lambda =\log\left(\frac{\rho(t)}{c(t)} \right)/(-t).
\end{align*}
Setting $t=\Delta$ and assuming (for now) that $c(\Delta)$ is known (for a (non-periodic) trawl process, we would have $c\equiv 1$), we get the following estimator
\begin{align*}
\hat \lambda = -\log\left(\frac{\rho_{n;\Delta}^*(\Delta)}{c(\Delta)} \right),
\end{align*}
where $\widehat{\rho}(\cdot)$ denotes the empirical autocorrelation function.

We now establish the corresponding limit theorem.
From Theorem \ref{thm:acflimits}, assuming the corresponding assumptions hold, we deduce 
that
\begin{align*}
\sqrt{n}(\rho_{n;\Delta}^*(\Delta)-\rho(\Delta)) \stackrel{\mathrm{d}}{\to}\mathrm{N}(0,w_{11;\Delta}).
\end{align*}

Next, we define the function
$\phi:(0, \infty) \to \R,  \phi(x):=-\log(x/c(\Delta))$, which is continuously differentiable with $\phi'(x)=-\frac{1}{x}$
and set
\begin{align*}
\widehat{\lambda} := \phi(\rho_{n;\Delta}^*(\Delta)).
\end{align*}
An application of the delta method, see \citet[Proposition 6.4.3]{BrockwellDavis1987}, 
leads to 
\begin{align*}
\sqrt{n}(\widehat{\lambda}-\lambda)
=\sqrt{n}(\phi(\rho_{n;\Delta}^*(\Delta))-\phi(\rho(\Delta))) \stackrel{\mathrm{d}}{\to}\mathrm{N}(0, \Sigma_{\lambda;\Delta}),
\end{align*}
where
\begin{align*} \Sigma_{\lambda;\Delta}:=\phi'(\rho(\Delta))w_{11;\Delta}\phi'(\rho(\Delta))
=\frac{w_{11;\Delta}}{\rho^2(\Delta)}
=\frac{w_{11;\Delta}}{c^2(\Delta)\exp(-2\lambda \Delta)}=\frac{w_{11;\Delta}\exp(2\lambda \Delta)}{c^2(\Delta)}.
\end{align*}

\begin{remark}
In the case of a nonperiodic trawl, we would have $c(\Delta)= 1$, which simplifies the asymptotic result.
For periodic trawls, in the  case of unknown  $c(\Delta)$, a plug-in estimator of the following form could be considered:
$
\hat \lambda = -\log\left(\frac{\rho_{n;\Delta}^*(\Delta)}{\widehat{c(\Delta)}} \right)$.
\end{remark}

Now, suppose that $c(\Delta)$ is not known, but the period $\tau$ of the periodic function is known. I.e.~we have that $c(x+\tau)=c(x)$, for all $x$. 
Furthermore, assume that 
there exists a $\widetilde{\tau}_{\Delta} \in \N$ such that $\tau =\widetilde{\tau}_{\Delta}\Delta$.
To simplify the notation, we shall also define
$\mathcal{T}_{\Delta}:=(1+\widetilde{\tau}_{\Delta})$.
Then, we have
\begin{align*}
c(\mathcal{T}_{\Delta}\Delta)=c(\Delta).
\end{align*}
We can now construct an estimator for the memory parameter based on two empirical autocorrelations. More precisely, note that
\begin{align*}
\rho(\Delta)=e^{-\lambda \Delta}c(\Delta),&&
\rho(\mathcal{T}_{\Delta}\Delta)=e^{-\lambda \mathcal{T}_{\Delta}\Delta}c(\mathcal{T}_{\Delta}\Delta)=e^{-\lambda \mathcal{T}_{\Delta}\Delta}c(\Delta),
\end{align*}
which leaves us with two equations to estimate the two unknown parameters $\lambda$ and $c(\Delta)$.

We can solve both equations for $c(\Delta)$
and obtain
\begin{align*}
c(\Delta) = \rho(\Delta) e^{\lambda \Delta},&&
c(\Delta) = \rho(\mathcal{T}_{\Delta}\Delta) e^{\lambda \mathcal{T}_{\Delta} \Delta}.
\end{align*}
Setting the two equations equal and solving for $\lambda$ leads to 
\begin{align*}
\lambda &=\frac{1}{\Delta(1- \mathcal{T}_{\Delta})}\log\left( \frac{\rho(\mathcal{T}_{\Delta} \Delta)}{\rho(\Delta)}\right).
\end{align*}
After $\lambda$ has been estimated, we can also estimate $c(l\Delta)$ for $l=1, \ldots, \tilde{\tau}_{\Delta}$, using the relation
\begin{align*}
c(l \Delta) =\rho(l \Delta)e^{\lambda l \Delta}.
\end{align*}
Altogether, we can consider the following estimator
\begin{multline*}
(\widehat{\lambda}, \widehat{c( \Delta)}, \ldots, \widehat{c(\tilde{\tau}_{\Delta}\Delta)})^{\top}
\\:= 
\left(\frac{1}{\Delta(1- \mathcal{T}_{\Delta})}\log\left( \frac{\rho_{n;\Delta}^*(\mathcal{T}_{\Delta} \Delta)}{\rho_{n;\Delta}^*(\Delta)}\right),
\left(
\rho_{n;\Delta}^*(l\Delta)
\exp\left\{  \frac{l}{(1- \mathcal{T}_{\Delta})}\log\left( \frac{\rho_{n;\Delta}^*(\mathcal{T}_{\Delta} \Delta)}{\rho_{n;\Delta}^*(\Delta)}\right)\right\}
\right)_{l=1, \ldots, \tilde{\tau}_{\Delta}}
\right).
\end{multline*}

We note that the vector $(\widehat{\lambda}, \widehat{c( \Delta)}, \ldots, \widehat{c(\tilde{\tau}_{\Delta}\Delta)})^{\top}$ 
can be represented as a function of the empirical autocorrelation functions $(\rho_{n;\Delta}^*(l\Delta))_{l=1, \ldots, \mathcal{T}_{\Delta}}$. As such, we can apply the delta method to derive a joint central limit theorem as follows.
More precisely, we have
\begin{align*}
(\widehat{\lambda}, \widehat{c( \Delta)}, \ldots, \widehat{c(\tilde{\tau}_{\Delta}\Delta)})^{\top}
=F((\rho_{n;\Delta}^*(l\Delta))_{l=1, \ldots, \mathcal{T}_{\Delta}})
\\
=(F_1((\rho_{n;\Delta}^*(l\Delta))_{l=1, \ldots, \mathcal{T}_{\Delta}}), \ldots, 
F_{\mathcal{T}_{\Delta}}((\rho_{n;\Delta}^*(l\Delta))_{l=1, \ldots, \mathcal{T}_{\Delta}})
)^{\top},
\end{align*}
where
the functions $F_1, \ldots, F_{\mathcal{T}_{\Delta}}$ are given by
\begin{align*}
F_1(\rho_1, \ldots, \rho_{\mathcal{T}_{\Delta}})
&=F_1(\rho_1, \rho_{\mathcal{T}_{\Delta}})=\frac{1}{\Delta(1- \mathcal{T}_{\Delta})}\log\left( \frac{\rho_{\mathcal{T}_{\Delta}}}{\rho_1}\right),\\
F_2(\rho_1, \ldots, \rho_{\mathcal{T}_{\Delta}})
&=F_2(\rho_1, \rho_{\mathcal{T}_{\Delta}})
= \rho_1 \exp(\Delta F_1(\rho_1, \rho_{\mathcal{T}_{\Delta}})),\\
F_3(\rho_1, \ldots, \rho_{\mathcal{T}_{\Delta}})
&=F_3(\rho_1, \rho_2, \rho_{\mathcal{T}_{\Delta}})
= \rho_2 \exp(2\Delta F_1(\rho_1, \rho_{\mathcal{T}_{\Delta}})),\\
F_4(\rho_1, \ldots, \rho_{\mathcal{T}_{\Delta}})
&=F_4(\rho_1, \rho_3, \rho_{\mathcal{T}_{\Delta}})
= \rho_3 \exp(3\Delta F_1(\rho_1, \rho_{\mathcal{T}_{\Delta}})),\\
\vdots\\
F_{\mathcal{T}_{\Delta}}(\rho_1, \ldots, \rho_{\mathcal{T}_{\Delta}})
&= F_{\mathcal{T}_{\Delta}}(\rho_1,  \rho_{\mathcal{T}_{\Delta}-1}, \rho_{\mathcal{T}_{\Delta}}) =\rho_{\mathcal{T}_{\Delta}-1} \exp((\mathcal{T}_{\Delta}-1)\Delta F_1(\rho_1, \rho_{\mathcal{T}_{\Delta}})).
\end{align*}
Under the assumptions of Theorem \ref{thm:acflimits},
\begin{align*}
\sqrt{n}((\widehat{\lambda}, \widehat{c( \Delta)}, \ldots, \widehat{c(\tilde{\tau}_{\Delta}\Delta)})^{\top}-
(\lambda, c(\Delta), \ldots, c(\mathcal{T}_{\Delta}\Delta))
\end{align*}
converges to a multivariate normal random vector with zero mean  and variance given by
$D {W}_{\Delta}D'$, where 
${W}_{\Delta}$ is the $\mathcal{T}_{\Delta}\times \mathcal{T}_{\Delta}$-matrix whose elements are the ones defined in Theorem \ref{thm:acflimits}.
The matrix $D$ is a $\mathcal{T}_{\Delta}\times \mathcal{T}_{\Delta}$-matrix 
$[(\partial F_i/\partial \rho_j)(\rho(\Delta),\ldots, \rho(\mathcal{T}_{\Delta}\Delta))^{\top}]$,
where
\begin{align*}
\frac{\partial F_1}{\partial \rho_1}(\rho_1, \ldots, \rho_{\mathcal{T}_{\Delta}})
&=\frac{\partial F_1}{\partial \rho_1}(\rho_1, \rho_{\mathcal{T}_{\Delta}})
=\frac{1}{\Delta(\mathcal{T}_{\Delta}-1) \rho_1},\\
\frac{\partial F_1}{\partial \rho_2}(\rho_1, \ldots, \rho_{\mathcal{T}_{\Delta}})
&=\cdots=\frac{\partial F_1}{\partial \rho_{\mathcal{T}_{\Delta}-1}}(\rho_1, \ldots, \rho_{\mathcal{T}_{\Delta}})=0,\\
\frac{\partial F_1}{\partial \rho_{\mathcal{T}_{\Delta}}}(\rho_1, \ldots, \rho_{\mathcal{T}_{\Delta}})
&=\frac{\partial F_1}{\partial \rho_{\mathcal{T}_{\Delta}}}(\rho_1, \rho_{\mathcal{T}_{\Delta}})
=\frac{1}{\Delta(1- \mathcal{T}_{\Delta}) \rho_{\mathcal{T}_{\Delta}}},\\
%
%
%
%
%
%
\frac{\partial F_2}{\partial \rho_1}(\rho_1, \ldots, \rho_{\mathcal{T}_{\Delta}})
&=\frac{\partial F_2}{\partial \rho_1}(\rho_1, \rho_{\mathcal{T}_{\Delta}})
=
\frac{\mathcal{T}_{\Delta}}{\mathcal{T}_{\Delta}-1}
\exp(\Delta F_1(\rho_1, \rho_{\mathcal{T}_{\Delta}})),\\
\frac{\partial F_2}{\partial \rho_2}(\rho_1, \ldots, \rho_{\mathcal{T}_{\Delta}})
&=\cdots=\frac{\partial F_2}{\partial \rho_{\mathcal{T}_{\Delta}-1}}(\rho_1, \ldots, \rho_{\mathcal{T}_{\Delta}})=0,\\
\frac{\partial F_2}{\partial \rho_{\mathcal{T}_{\Delta}}}(\rho_1, \ldots, \rho_{\mathcal{T}_{\Delta}})
&=\frac{\partial F_2}{\partial \rho_{\mathcal{T}_{\Delta}}}(\rho_1, \rho_{\mathcal{T}_{\Delta}})
=\frac{1}{1-\mathcal{T}_{\Delta}}
\frac{\rho_1}{\rho_{\mathcal{T}_{\Delta}}}
\exp(\Delta F_1(\rho_1, \rho_{\mathcal{T}_{\Delta}})),\\
%
%
%
%
%
\frac{\partial F_3}{\partial \rho_1}(\rho_1, \ldots, \rho_{\mathcal{T}_{\Delta}})
&=\frac{\partial F_3}{\partial \rho_1}(\rho_1,\rho_2,\rho_{\mathcal{T}_{\Delta}})
= \frac{2}{ \mathcal{T}_{\Delta}-1}
\frac{\rho_2}{\rho_1}\exp(2\Delta F_1(\rho_1, \rho_{\mathcal{T}_{\Delta}})),\\
\frac{\partial F_3}{\partial \rho_2}(\rho_1, \ldots, \rho_{\mathcal{T}_{\Delta}})
&=\frac{\partial F_3}{\partial \rho_2}(\rho_1,\rho_2,\rho_{\mathcal{T}_{\Delta}})
=
\exp(2\Delta F_1(\rho_1, \rho_{\mathcal{T}_{\Delta}})),\\
\frac{\partial F_3}{\partial \rho_3}(\rho_1, \ldots, \rho_{\mathcal{T}_{\Delta}})
&=\cdots=\frac{\partial F_3}{\partial \rho_{\mathcal{T}_{\Delta}-1}}(\rho_1, \ldots, \rho_{\mathcal{T}_{\Delta}})=0,\\
\frac{\partial F_3}{\partial \rho_{\mathcal{T}_{\Delta}}}(\rho_1, \ldots, \rho_{\mathcal{T}_{\Delta}})
&=\frac{\partial F_3}{\partial \rho_{\mathcal{T}_{\Delta}}}(\rho_1,\rho_2, \rho_{\mathcal{T}_{\Delta}})
=
 \frac{2}{1- \mathcal{T}_{\Delta}}
\frac{\rho_2}{\rho_{\mathcal{T}_{\Delta}}}\exp(2\Delta F_1(\rho_1, \rho_{\mathcal{T}_{\Delta}})),\\
\vdots\\
\frac{\partial F_{\mathcal{T}_{\Delta}}}{\partial \rho_1}(\rho_1, \ldots, \rho_{\mathcal{T}_{\Delta}})
&= \frac{\partial F_{\mathcal{T}_{\Delta}}}{\partial \rho_1}(\rho_1,\rho_{\mathcal{T}_{\Delta}-1}  \rho_{\mathcal{T}_{\Delta}}) =\frac{\rho_{\mathcal{T}_{\Delta}-1}}{\rho_1} \exp((\mathcal{T}_{\Delta}-1)\Delta F_1(\rho_1, \rho_{\mathcal{T}_{\Delta}})),\\
\frac{\partial F_{\mathcal{T}_{\Delta}}}{\partial \rho_2}(\rho_1, \ldots, \rho_{\mathcal{T}_{\Delta}})
&= 
\cdots= \frac{\partial F_{\mathcal{T}_{\Delta}}}{\partial \rho_{\mathcal{T}_{\Delta}-2}}(\rho_1, \ldots, \rho_{\mathcal{T}_{\Delta}})
=0,\\
\frac{\partial F_{\mathcal{T}_{\Delta}}}{\partial \rho_{\mathcal{T}_{\Delta}-1}}(\rho_1, \ldots, \rho_{\mathcal{T}_{\Delta}})
&=\frac{\partial F_{\mathcal{T}_{\Delta}}}{\partial \rho_{\mathcal{T}_{\Delta}-1}}(\rho_1, \rho_{\mathcal{T}_{\Delta}-1},  \rho_{\mathcal{T}_{\Delta}})
=\exp((\mathcal{T}_{\Delta}-1)\Delta F_1(\rho_1, \rho_{\mathcal{T}_{\Delta}})), 
\\
\frac{\partial F_{\mathcal{T}_{\Delta}}}{\partial \rho_{\mathcal{T}_{\Delta}}}(\rho_1, \ldots, \rho_{\mathcal{T}_{\Delta}})
&= \frac{\partial F_{\mathcal{T}_{\Delta}}}{\partial \rho_{\mathcal{T}_{\Delta}}}(\rho_1,\rho_{\mathcal{T}_{\Delta}-1},   \rho_{\mathcal{T}_{\Delta}}) =-\frac{\rho_{\mathcal{T}_{\Delta}-1}}{\rho_{\mathcal{T}_{\Delta}}}\exp((\mathcal{T}_{\Delta}-1)\Delta F_1(\rho_1, \rho_{\mathcal{T}_{\Delta}})).
\end{align*}

\subsection{SupGamma trawl function}
Let us now focus on the case of a supGamma trawl function that allows for both short- and long-memory settings. 

We note that the inference for the long-memory case appears numerically unstable when 
approximating the corresponding integrals appearing in the autocorrelation function (in simulation experiments not reported here). Hence we suggest, as before,  using the mean value theorem result for inference on the memory parameter.
 Let $g(x)=\left(1+\frac{x}{\alpha}\right)^{-H}$, for $\alpha>0, H>1$. 
 As before, there exists a $\tau$-periodic function $c$ such that, for $t\geq 0$,  
\begin{align*}
\rho(t)=	\Cor(Y_{0},Y_{t})&=
c(t)\left(1+\frac{t}{\alpha}\right)^{1-H}.
\end{align*}
	For $H\in(1,2]$, we are in the long-memory case and for $H>2$ in the short-memory case.

We are interested in estimating the Hurst index $H$ and therefore in the following assume that the parameter $\alpha$ is known.

Using the analogous procedure as in the short-memory case, we can write that, for all $t\geq 0$, 
\begin{align*}
\frac{\rho(t)}{c(t)}=\exp((1-H)\log(1+t/\alpha)) 
\Leftrightarrow \log\left(\frac{\rho(t)}{c(t)} \right)=(1-H)\log(1+t/\alpha).
\end{align*}
This suggests the following estimator of the Hurst index:
\begin{align*}
\widehat{H} =1-\log\left(\frac{\rho_{n;\Delta}^*(\Delta)}{c(\Delta)} \right)/\log(1+\Delta/\alpha),
\end{align*}
assuming that $c(\Delta)$ and $\alpha$ are known.

Now, define the function
$\phi:(0, \infty) \to \R,  \phi(x):=1-\log\left(\frac{x}{c(\Delta)} \right)/\log(1+\Delta/\alpha)$, which is continuously differentiable with $\phi'(x)=-\frac{1}{\log(1+\Delta/\alpha)}\frac{1}{x}$
and set
\begin{align*}
\widehat{H} := \phi(\rho_{n;\Delta}^*(\Delta)).
\end{align*}
An application of the delta method, see \citet[Proposition 6.4.3]{BrockwellDavis1987}, 
leads to 
\begin{align*}
\sqrt{n}(\widehat{H}-H)
=\sqrt{n}(\phi(\rho_{n;\Delta}^*(\Delta))-\phi(\rho(\Delta))) \stackrel{\mathrm{d}}{\to}\mathrm{N}(0, \Sigma_{H;\Delta}),
\end{align*}
where
\begin{align*} \Sigma_{H;\Delta}&:=\phi'(\rho(\Delta))w_{11;\Delta}\phi'(\rho(\Delta))
=\frac{w_{11;\Delta}}{
(\log(1+\Delta/\alpha))^2\rho^2(\Delta)
}\\
&=
\frac{w_{11;\Delta}}{
(\log(1+\Delta/\alpha))^2c^2(\Delta)\left(1+\frac{\Delta}{\alpha}\right)^{2-2H}
}
=\frac{w_{11;\Delta}\left(1+\frac{\Delta}{\alpha}\right)^{2H-2}}{
(\log(1+\Delta/\alpha))^2c^2(\Delta)
}.
\end{align*}

Now suppose that $\tau$ is known and proceed as in the exponential case: Assume that 
there exists a $\widetilde{\tau}_{\Delta} \in \N$ such that $\tau =\widetilde{\tau}_{\Delta}\Delta$
and set 
$\mathcal{T}_{\Delta}:=(1+\widetilde{\tau}_{\Delta})$.
Then, we have
\begin{align*}
c(\mathcal{T}_{\Delta}\Delta)=c(\Delta).
\end{align*}
We can now construct an estimator for the memory parameter based on two empirical autocorrelations. More precisely, note that
\begin{align*}
\rho(\Delta)=(1+\Delta/\alpha)^{H-1}c(\Delta),&&
\rho(\mathcal{T}_{\Delta}\Delta)=(1+\mathcal{T}_{\Delta}\Delta/\alpha)^{H-1}c(\mathcal{T}_{\Delta}\Delta)=(1+\mathcal{T}_{\Delta}\Delta/\alpha)^{H-1}c(\Delta),
\end{align*}
which leaves us with two equations to estimate the two unknown parameters $H$ and $c(\Delta)$, assuming that $\alpha$ is known. 

We can solve both equations for $c(\Delta)$
and obtain
\begin{align*}
c(\Delta) = \rho(\Delta) (1+\Delta/\alpha)^{1-H},&&
c(\Delta) = \rho(\mathcal{T}_{\Delta}\Delta) (1+\mathcal{T}_{\Delta}\Delta/\alpha)^{1-H}.
\end{align*}
Setting the two equations equal and solving for $H$ leads to 
\begin{align*}
H =1+\log\left( \frac{\rho(\mathcal{T}_{\Delta} \Delta)}{\rho(\Delta)}\right)/
\log\left( \frac{\alpha+\Delta}{\alpha +\mathcal{T}_{\Delta} \Delta}\right).
\end{align*}

After $H$ has been estimated, we can also estimate $c(l\Delta)$ for $l=1, \ldots, \tilde{\tau}_{\Delta}$, using the relation
\begin{align*}
c(l \Delta) =\rho(l \Delta)(1+l\Delta/\alpha)^{1-H}.
\end{align*}
Altogether, we can consider the following estimator
\begin{multline*}
(\widehat{H}, \widehat{c( \Delta)}, \ldots, \widehat{c(\tilde{\tau}_{\Delta}\Delta)})^{\top}
\\:= 
\left(1+\log\left( \frac{\rho_{n;\Delta}^*(\mathcal{T}_{\Delta} \Delta)}{\rho_{n;\Delta}^*(\Delta)}\right)/
\log\left( \frac{\alpha+\Delta}{\alpha +\mathcal{T}_{\Delta} \Delta}\right),\right.\\
\left.
\left(
\rho_{n;\Delta}^*(l\Delta)
(1+l\Delta/\alpha)^{\log\left( \frac{\rho_{n;\Delta}^*(\mathcal{T}_{\Delta} \Delta)}{\rho_{n;\Delta}^*(\Delta)}\right)/
\log\left( \frac{\alpha+\Delta}{\alpha +\mathcal{T}_{\Delta} \Delta}\right)}
\right)_{l=1, \ldots, \tilde{\tau}_{\Delta}}
\right).
\end{multline*}

We note that the vector $(\widehat{H}, \widehat{c( \Delta)}, \ldots, \widehat{c(\tilde{\tau}_{\Delta}\Delta)})^{\top}$ 
can be represented as a function of the empirical autocorrelation functions $(\rho_{n;\Delta}^*(l\Delta))_{l=1, \ldots, \mathcal{T}_{\Delta}}$. That is, we have
\begin{align*}
(\widehat{H}, \widehat{c( \Delta)}, \ldots, \widehat{c(\tilde{\tau}_{\Delta}\Delta)})^{\top}
=F((\rho_{n;\Delta}^*(l\Delta))_{l=1, \ldots, \mathcal{T}_{\Delta}})
\\
=(F_1((\rho_{n;\Delta}^*(l\Delta))_{l=1, \ldots, \mathcal{T}_{\Delta}}), \ldots, 
F_{\mathcal{T}_{\Delta}}((\rho_{n;\Delta}^*(l\Delta))_{l=1, \ldots, \mathcal{T}_{\Delta}})
)^{\top},
\end{align*}
where
the functions $F_1, \ldots, F_{\mathcal{T}_{\Delta}}$ are given by
\begin{align*}
F_1(\rho_1, \ldots, \rho_{\mathcal{T}_{\Delta}})
&=F_1(\rho_1, \rho_{\mathcal{T}_{\Delta}}) =1+\log\left( \frac{\rho_{\mathcal{T}_{\Delta}}}{\rho_1}\right)/
\log\left( \frac{\alpha+\Delta}{\alpha +\mathcal{T}_{\Delta} \Delta}\right),\\
F_2(\rho_1, \ldots, \rho_{\mathcal{T}_{\Delta}})
&=F_2(\rho_1, \rho_{\mathcal{T}_{\Delta}})
= \rho_1  (1+\Delta/\alpha)^{1-F_1(\rho_1,  \rho_{\mathcal{T}_{\Delta}})}\\
F_3(\rho_1, \ldots, \rho_{\mathcal{T}_{\Delta}})
&=F_3(\rho_1, \rho_2, \rho_{\mathcal{T}_{\Delta}})
= \rho_2 (1+2\Delta/\alpha)^{1-F_1(\rho_1, \rho_{\mathcal{T}_{\Delta}})},\\
F_4(\rho_1, \ldots, \rho_{\mathcal{T}_{\Delta}})
&=F_4(\rho_1, \rho_3, \rho_{\mathcal{T}_{\Delta}})
= \rho_3 (1+3\Delta/\alpha)^{1-F_1(\rho_1, \rho_{\mathcal{T}_{\Delta}})},\\
\vdots\\
F_{\mathcal{T}_{\Delta}}(\rho_1, \ldots, \rho_{\mathcal{T}_{\Delta}})
&= F_{\mathcal{T}_{\Delta}}(\rho_1,  \rho_{\mathcal{T}_{\Delta}-1}, \rho_{\mathcal{T}_{\Delta}}) =\rho_{\mathcal{T}_{\Delta}-1}  (1+(\mathcal{T}_{\Delta}-1)\Delta/\alpha)^{1-F_1(\rho_1, \rho_{\mathcal{T}_{\Delta}})}.
\end{align*}
Under the assumptions of Theorem \ref{thm:acflimits}
\begin{align*}
\sqrt{n}((\widehat{H}, \widehat{c( \Delta)}, \ldots, \widehat{c(\tilde{\tau}_{\Delta}\Delta)})^{\top}-
(H, c(\Delta), \ldots, c(\mathcal{T}_{\Delta}\Delta))
\end{align*}
converges to a multivariate normal random vector with zero mean  and variance given by
$D {W}_{\Delta}D'$, where 
${W}_{\Delta}$ is the $\mathcal{T}_{\Delta}\times \mathcal{T}_{\Delta}$-matrix whose elements are the ones defined in Theorem \ref{thm:acflimits}.
The matrix $D$ is a $\mathcal{T}_{\Delta}\times \mathcal{T}_{\Delta}$-matrix 
$[(\partial F_i/\partial \rho_j)(\rho(\Delta),\ldots, \rho(\mathcal{T}_{\Delta}\Delta))^{\top}]$.

\begin{remark}
  Note that the assumptions of Theorem \ref{thm:acflimits} are not satisfied in the case of a long-memory supGamma trawl function, i.e.~when $H\in (1,2]$, see Subsection \ref{sec:asscheck} in the Appendix. 
\end{remark}

\section{Inference for periodic trawl processes using a generalised-method-of-moments approach}\label{sec:inferenceGMM}
In the previous section, we showed how the memory parameter and periodic function of a periodic trawl process can be inferred if the period is known. 
More generally, if we would like to estimate all model parameters, including the parameters of the L\'{e}vy basis, we can proceed by using 
the generalised method of moments (GMM). 

Hence, we will now develop the asymptotic theory for estimating the parameters of a periodic trawl process using a generalised method of moments (GMM). This extends the work presented in \cite{BLSV2023} for (integer-valued) trawl processes to the case of periodic trawl processes.

Recall that we denote the periodic trawl process by 
\begin{align*}
	Y_t & =
\int_{\R\times \R}p(t-s)\ind_{(0,g(t-s))}(x)\ind_{[0,\infty)}(t-s)L(dx,ds),
		\end{align*}
  and we denote the corresponding trawl process, where $p\equiv 1$, by 
\begin{align}
	X_t & =
 \int_{\R\times \R}\ind_{(0,g(t-s))}(x)\ind_{[0,\infty)}(t-s)L(dx,ds).
\end{align}

\subsection{Weak dependence}
We will first show that periodic trawl processes are  $\theta$-weakly dependent, see \citet[Definition 3.2]{CuratoStelzer2019}.

We note that a periodic trawl process is  
a special case of a causal mixed moving average process as defined in 
\citet[Definition 3.3]{CuratoStelzer2019}.
Hence, under the assumption that $\int_{|\xi|>1}|\xi|^2\lev(d\xi)<\infty$, we can deduce from \citet[Corollary 3.4]{CuratoStelzer2019} 
 the periodic trawl process is $\theta$-weakly dependent, see \citet[Definition 3.2]{CuratoStelzer2019} with coefficient, for $r\geq 0$, 
 \begin{align*}
\theta_Y(r)&=\left( \mathrm{Var}(L')\int_{(-\infty, -r)\times \R}p^2(-s)\ind_{(0,g(-s))}^2(x)\ind_{[0,\infty)}^2(-s) dxds
\right.\\
&\left. 
+
\left| \E(L')\int_{(-\infty, -r)\times \R}p(-s)\ind_{(0,g(-s))}(x)\ind_{[0,\infty)}(-s) dxds
\right|^2
\right)^{1/2}
\\
&=
\left( \mathrm{Var}(L')\int_{-\infty}^{-r}p^2(-s)g(-s)ds+ (\E(L'))^2\left(\int_{-\infty}^{-r}p(-s)g(-s)ds\right)^2\right)^{1/2}
\\
&=
\left( \mathrm{Var}(L')\int_{r}^{\infty}p^2(s)g(s)ds+ (\E(L'))^2\left(\int_{r}^{\infty}p(s)g(s)ds\right)^2\right)^{1/2}.
\end{align*}
Since the periodic function is continuous and bounded, we note that the weak-dependence coefficient can be bounded by a function $\widetilde{\theta}_Y$ such that
$\theta_Y(r)\leq \widetilde{\theta}_Y(r)$ for all $r\geq 0$, where 
\begin{align*}
\widetilde{\theta}_Y(r)&=
\left( c_1\mathrm{Cov}(X_0, X_r)
+c_2 (\E(L'))^2 (\mathrm{Cov}(X_0, X_r))^2\right)^{1/2},
 \end{align*}
 where $\E(L')=\gamma+\int_{|\xi|>1}\lev(d\xi)$, and 
 $\mathrm{Var}(L')=a+\int_{\R}\xi^2\lev(d\xi)$ and $c_1, c_2>0$ are constants.

Let us briefly consider the special case when $L'$ is of finite variation, i.e.~when the characteristic triplet is given by $(\gamma, 0, \lev)$ with $\int_{\R}|\xi|\lev(\xi)< \infty$.
This case is of interest since it includes the well-known class of integer-valued trawl processes. 
Here, the weak dependence coefficient is, for $r\geq 0$,  given by
\begin{align*}
\theta_Y(r)&=\int_{(-\infty, -r)\times \R}\int_{\R}|p(-s)\ind_{(0,g(-s))}(x)\ind_{[0,\infty)}(-s)\xi|\lev(d\xi) dxds\\
& +\int_{(-\infty, -r)\times \R}|p(-s)\ind_{(0,g(-s))}(x)\ind_{[0,\infty)}(-s)\gamma_0| dxds\\
&=\int_{\R}|\xi|\lev(d\xi) \int_{-\infty}^{-r}|p(-s)|g(-s) ds +|\gamma_0|\int_{-\infty}^{-r}|p(-s)|g(-s)ds\\
&=\left(\int_{\R}|\xi|\lev(d\xi)+|\gamma_0|\right) \int_{r}^{\infty}|p(s)|g(s)ds.
\end{align*}
As before, 
we note that the weak-dependence coefficient can be bounded by a function $\widetilde{\theta}_Y$ such that
$\theta_Y(r)\leq \widetilde{\theta}_Y(r)$ for all $r\geq 0$, where 
\begin{align*}
\widetilde{\theta}_Y(r)&=
c\mathrm{Cov}(X_0, X_r),
\end{align*}
for a positive constant $c>0$ and $\gamma_0=\gamma-\int_{|\xi|\leq 1}\xi\lev(d\xi)$.

\begin{remark}
    We note that the $\theta$-dependence coefficient of the periodic trawl process $Y$  can be related to the $\theta$-dependence coefficient of $X$ via
    \begin{align*}
        \theta_Y(r)\leq \widetilde{\theta}_Y(r)=O(\theta_X(r)).
    \end{align*}
  \end{remark}
  
  \subsection{GMM estimation for periodic trawl processes}\label{sec:GMM}
  In this section, we will describe how the parameters of a periodic trawl process can be estimated via the generalised method of moments (GMM).

As before, we consider the equidistantly sampled process $Y_{\Delta}, Y_{2\Delta}, \ldots, Y_{n\Delta}$, for $\Delta=T/n >0, T>0, n \in \N$. We will consider a GMM estimator, which is based on the sample mean, sample variance and sample autocovariances up to lag $m\geq 2$.

Consider the vector
$$
Y_t^{(m)}=(Y_{t\Delta}, Y_{(t+1)\Delta}, \ldots, Y_{(t+m)\Delta}),
$$
for $t=1, \ldots, n-m$. We denote by  $\Theta$  the parameter space of the periodic trawl process and set $\mu:=\mu(\theta)=\E(Y_0)$ and $D(k):=D(k,\theta):=\E(Y_0Y_{k\Delta})$, for $k=0, \ldots, m$. So, as soon as we specify a  parametric model for $Y$, then $D(k)$ is  a function of the model parameter(s) $\theta$.

Next, we define the measurable function $h:\R^{m+1}\times \Theta\to \R^{m+2}$ by
\begin{align*}
h(Y_t^{(m)}, \theta)&=
\left(
\begin{array}{c}
h_E(Y_t^{(m)}, \theta)\\
h_0(Y_t^{(m)}, \theta)\\
h_1(Y_t^{(m)}, \theta)\\
\vdots
\\
h_m(Y_t^{(m)}, \theta)\\
\end{array}
\right)
=
\left(
\begin{array}{c}Y_{t\Delta}-\mu(\theta) \\
Y_{t\Delta}^2-D(0,\theta)\\
Y_{t\Delta}Y_{(t+1)\Delta}-D(1,\theta)\\
\vdots
\\
Y_{t\Delta}Y_{(t+m)\Delta}-
D(m,\theta)
\end{array}
\right).
\end{align*}
Moreover, define the corresponding sample moments  as 
\begin{align*}
g_{n,m}(\theta)&=
\frac{1}{n-m}\sum_{t=1}^{n-m}h(Y_t^{(m)}, \theta)=
\left(
\begin{array}{c}
\frac{1}{n-m}\sum_{t=1}^{n-m}h_E(Y_t^{(m)}, \theta)\\
\frac{1}{n-m}\sum_{t=1}^{n-m}h_0(Y_t^{(m)}, \theta)\\
\vdots
\\
\frac{1}{n-m}\sum_{t=1}^{n-m}h_m(Y_t^{(m)}, \theta)
\end{array}
\right).
\end{align*}
Suppose that the true parameter (vector) is denoted by $\theta_0$, say, which can be estimated by minimising the objective function of the GMM, i.e.~the GMM estimator is given by 
\begin{align}\label{eq:GMMest}
\widehat\theta_{0,\mathrm{GMM}}^{n,m}
=\mathrm{argmin} 
g_{n,m}(\theta)^{\top} A_{n,m}
g_{n,m}(\theta),
\end{align}
for a positive-definite weight matrix $A_{n,m}$.

We now  derive the weak consistency and asymptotic normality of the GMM estimator under suitable (standard) assumptions.

First, we present the assumptions which guarantee the weak consistency of the estimator, cf.~\citet[Assumptions 1.1-1.3]{Matyas1999}
\begin{assumption}\label{as:1.1}
\begin{enumerate}
\item[(i)] Suppose that the expectation $\E(h(Y_t^{(m)}, \theta))$ exists and is finite for all $\theta\in \Theta$ and for all $t$.
\item[(ii)] Set $h_t^{(m)}(\theta)=\E(h(Y_t^{(m)}, \theta))$. There exists a $\theta_0\in \Theta$ such that $h_t^{(m)}(\theta)=0$ for all $t$ if and only if $\theta=\theta_0$. 
\end{enumerate}
\end{assumption}
\begin{assumption}\label{as:1.2}
Let $h^{(m)}(\theta)=\sum_{t=1}^{n-m}h_t^{(m)}(\theta)$ and denote the $j$th component of the $m+2$-dimensional vectors $h^{(m)}(\theta)$ and $g_{n,m}(\theta)$ by $h^{(m)}_j(\theta)$ and $g_{n,m; j}(\theta)$, respectively.
Suppose that, for $j=1, \ldots, m+2$, as $n\to \infty$, 
$$
\sup_{\theta \in \Theta}|h^{(m)}_j(\theta) -g_{n,m; j}(\theta)|\stackrel{\mathbb{P}}{\rightarrow} 0,
$$
\end{assumption}

\begin{assumption}\label{as:1.3}
There exists a sequence of non-random, positive definite matrices $\overline{A}_{n,m}$ such that, as $n \to \infty$, 
$|A_{n,m}-\overline{A}_{n,m}|\stackrel{\mathbb{P}}{\rightarrow} 0$
\end{assumption}

From \citet[Theorem 1.1]{Matyas1999}, we deduce the following result.
\begin{theorem}\label{thm:consistency}
Assume that Assumptions \ref{as:1.1}, \ref{as:1.2}, \ref{as:1.3} hold. Then the GMM estimator $\widehat\theta_{0,\mathrm{GMM}}^{n,m}$ defined in \eqref{eq:GMMest}
is weakly consistent.
\end{theorem}

Next, we formulate the assumptions needed for the central limit theorem.

\begin{assumption}\label{as:1}
$\Theta$ is a compact  parameter space which includes the true parameter $\theta_0$. 
\end{assumption}
\begin{assumption}\label{as:2}
The weight matrix $A_{n,m}$ converges in probability to a positive definite matrix $A$.
\end{assumption}
\begin{assumption}\label{as:3}
The covariance matrix $\Sigma_a$ defined in \eqref{eq:Sigmaa} below is positive definite.
\end{assumption}

\begin{theorem}\label{prop:clt-gmm}
Consider a periodic trawl process $Y$ with characteristic triplet $(\drift, a, \lev)$ and suppose that
$\int_{|\xi|>1}|\xi|^{4+\delta}\lev(d\xi)<\infty$, for some $\delta>0$ and and suppose that the $\theta$-weak dependence coefficient of the periodic trawl process satisfies  $\theta_Y(r)\leq O(r^{-\alpha})$, for $\alpha>\left( 1+\frac{1}{\delta}\right)\left(1+\frac{1}{2+\delta} \right)$.
Suppose that Assumptions \ref{as:1.1},  \ref{as:1}, \ref{as:2}, \ref{as:3} hold.
Then, as $n\to \infty$, 
$$
\sqrt{n}(\widehat\theta_{0,\mathrm{GMM}}^{n,m}-\theta_0)
\stackrel{d}{\to}\mathrm{N}(0, M\Sigma_aM^{\top}),
$$
where
\begin{align}\label{eq:Sigmaa}
\Sigma_{a}&=\sum_{l\in \Z}\mathrm{Cov}(h(Y_0^{(m)}, \theta_0), h(Y_l^{(m)}, \theta_0)),\\ \nonumber
M&=(G_0^{\top}AG_0)^{-1}G_0^{\top}A, \quad \mathrm{where} \; \;\;
G_0=\E\left[\frac{\partial h(Y_t^{(m)}, \theta)}{\partial \theta^{\top}} \right]_{\theta=\theta_0}.
\end{align}
\end{theorem}
\begin{proof}[Proof of Theorem \ref{prop:clt-gmm}]
The proof follows, with very minor modifications, the steps of the proof presented for trawl processes in \cite{BLSV2023}, which is based on the arguments of the proofs of  \citet[Theorem 1.2]{Matyas1999}, see also \citet[Proof of Theorem 6.2]{CuratoStelzer2019} for the case of a supOU process.
\end{proof}

\begin{remark}
    We note that the assumption on the $\theta$-weak dependence index rule out long-memory settings, see \cite{BLSV2023} for a more detailed discussion in the case of trawl processes. 
\end{remark}

\section{Empirical illustration  on  electricity day-ahead prices}\label{sec:empirics}
In this section, we will illustrate how the proposed methodology for estimating the kernel function of periodic trawl processes developed in Section \ref{sec:inference} can be used in practice, see  \cite{PeriodicTrawl-Energy} for the corresponding   data and {\tt R} code. 

We consider  day-ahead electricity baseload prices	for Germany and Luxembourg from 1st October 2018 to 1st January 2023 recorded in EUR/MWh. The data have been downloaded from the website \url{https://www.smard.de/en}, which is maintained by the 
 Bundesnetzagentur in Germany. The time series is depicted in Figure \ref{fig:Data}.

Given the turmoil experienced by the electricity market in recent years, a single stationary stochastic process model would not fit such data appropriately. Hence, we rather split the dataset into two parts, representing the relatively calm period from 01.10.2018--31.12.2020 (823 observations, referred to as TS1) and the volatile period from 01.01.2021--01.01.2023 (731 observations, referred to as TS2). We remark that more sophisticated methods, e.g.~from change-point-detection, could be applied to split the dataset. Figure \ref{fig:Data12}
depicts the time series, the empirical autocorrelation functions and the empirical densities of both time series. We note a rather slowly decaying serial correlation with a distinct weakly periodic pattern. The marginal distributions are rather different for the calm and the volatile regime: In the calm regime, we observe the well-known fact of a slightly skewed distribution, which has slightly heavier tails than the normal distribution and includes negative prices. For the volatile regime, the distribution appears to be skewed, multimodal, heavy-tailed and with significantly larger empirical mean and variance than in the calm regime.

\begin{figure}[htbp]
\centering
\includegraphics[scale=0.25,angle=270]{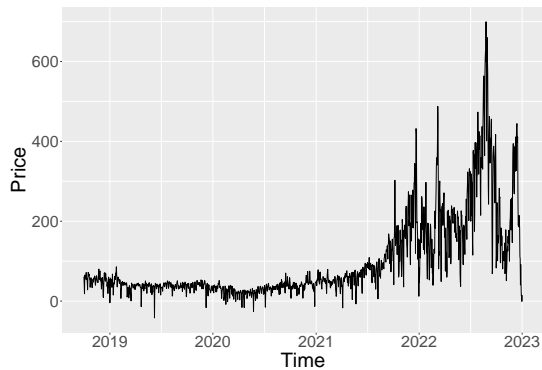}  
\caption{Time series plot of the  day-ahead electricity  baseload prices	for Germany and Luxembourg from 1st October 2018 to 1st January 2023 recorded in EUR/MWh. }
 \end{figure}

\begin{figure}[htbp]
	\captionsetup[subfigure]{aboveskip=-4pt,belowskip=-4pt}
	\subfloat[TS1: Daily prices from 01.10.2018--31.12.2020\label{fig:data1}]{	\includegraphics[scale=0.25, angle=270]{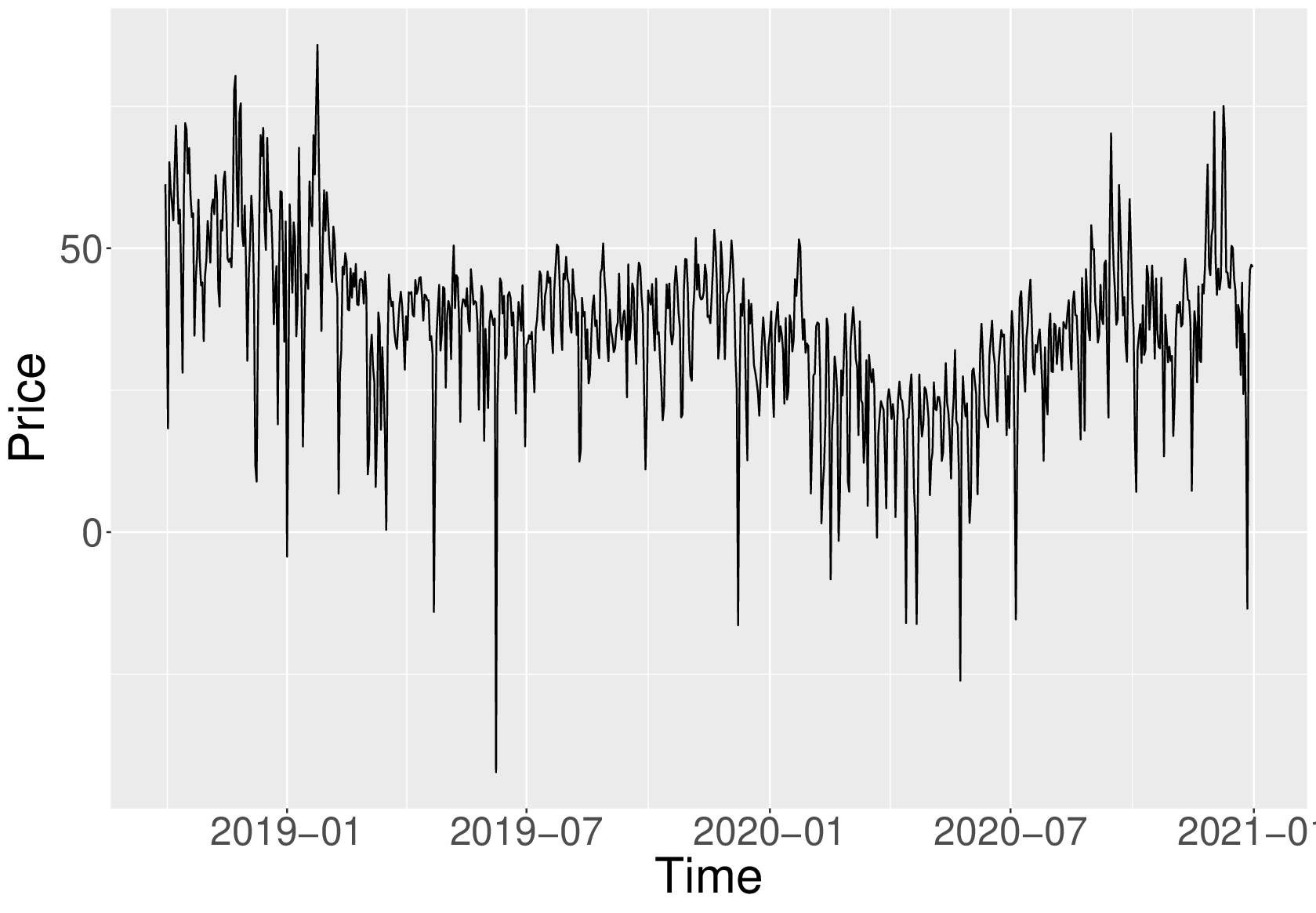} } 
	\captionsetup[subfigure]{aboveskip=-4pt,belowskip=-4pt}
	\subfloat[TS2: Daily prices from 01.01.2021--01.01.2023\label{fig:data2}]{	\includegraphics[scale=0.25, angle=270]{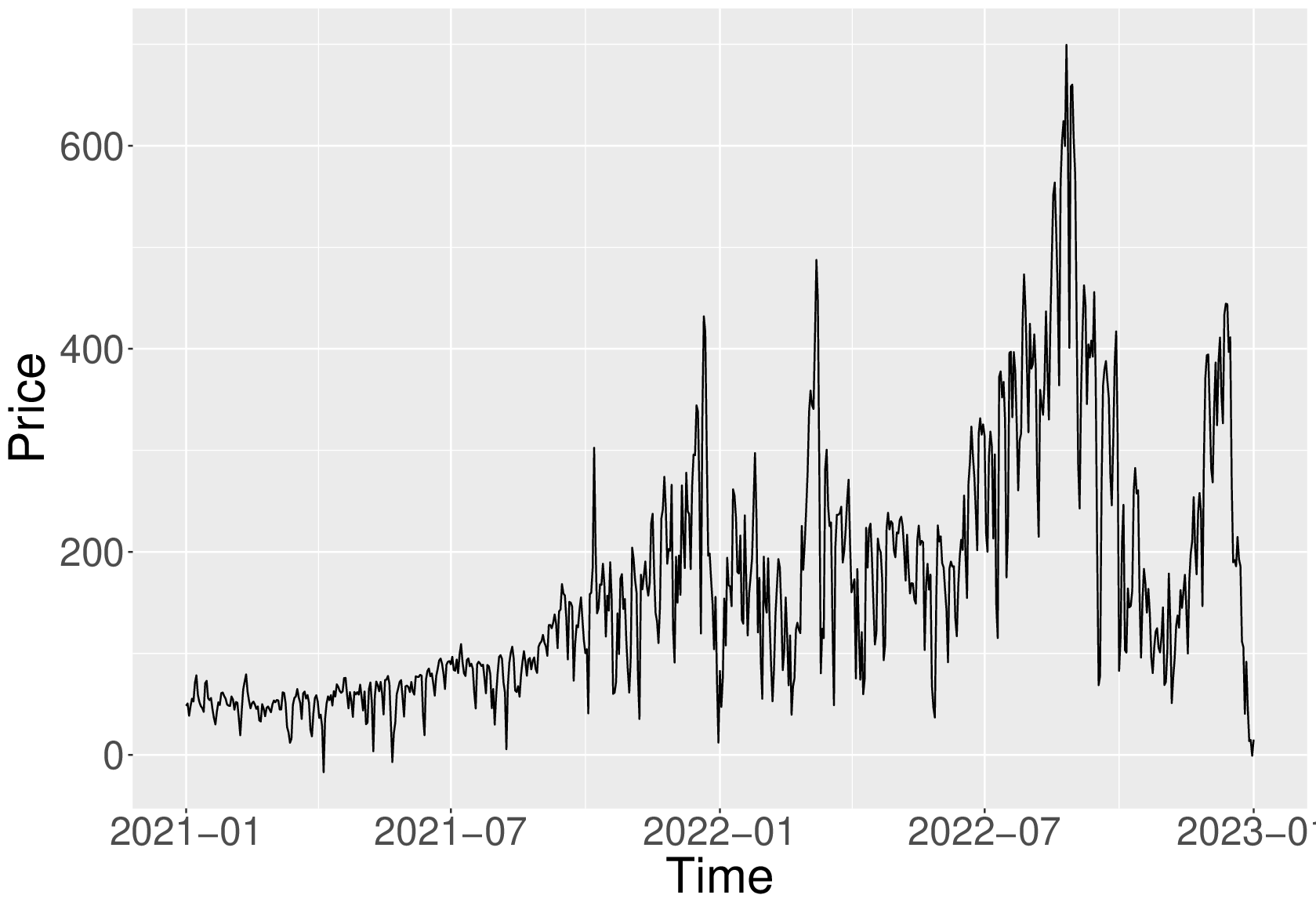}} 	\\
		\subfloat[Empirical ACF of TS1\label{fig:acf1}]{	\includegraphics[scale=0.25, angle=270]{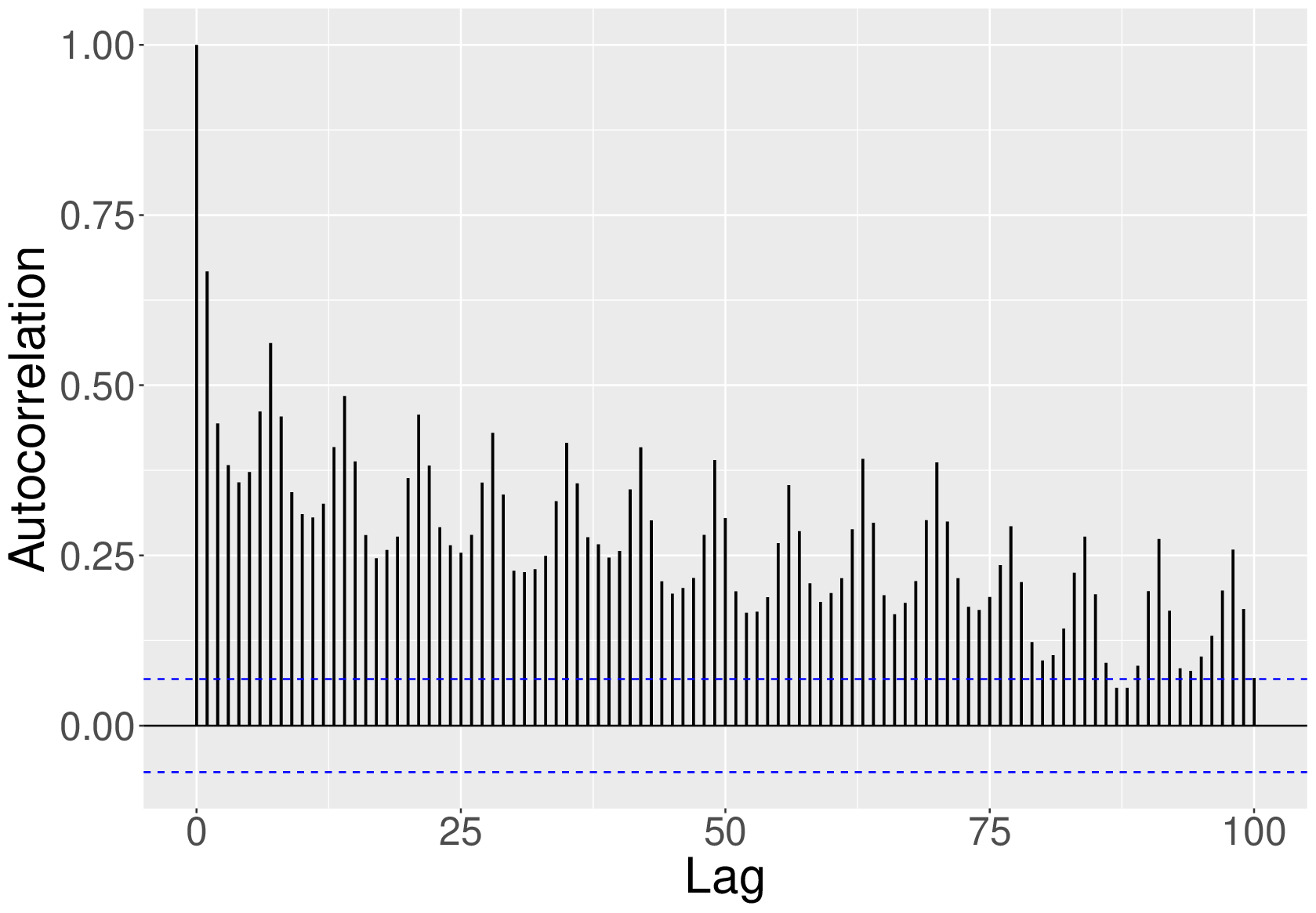} } 
	\captionsetup[subfigure]{aboveskip=-4pt,belowskip=-4pt}
	\subfloat[Empirical ACF of TS2\label{fig:acf2}]{	\includegraphics[scale=0.25, angle=270]{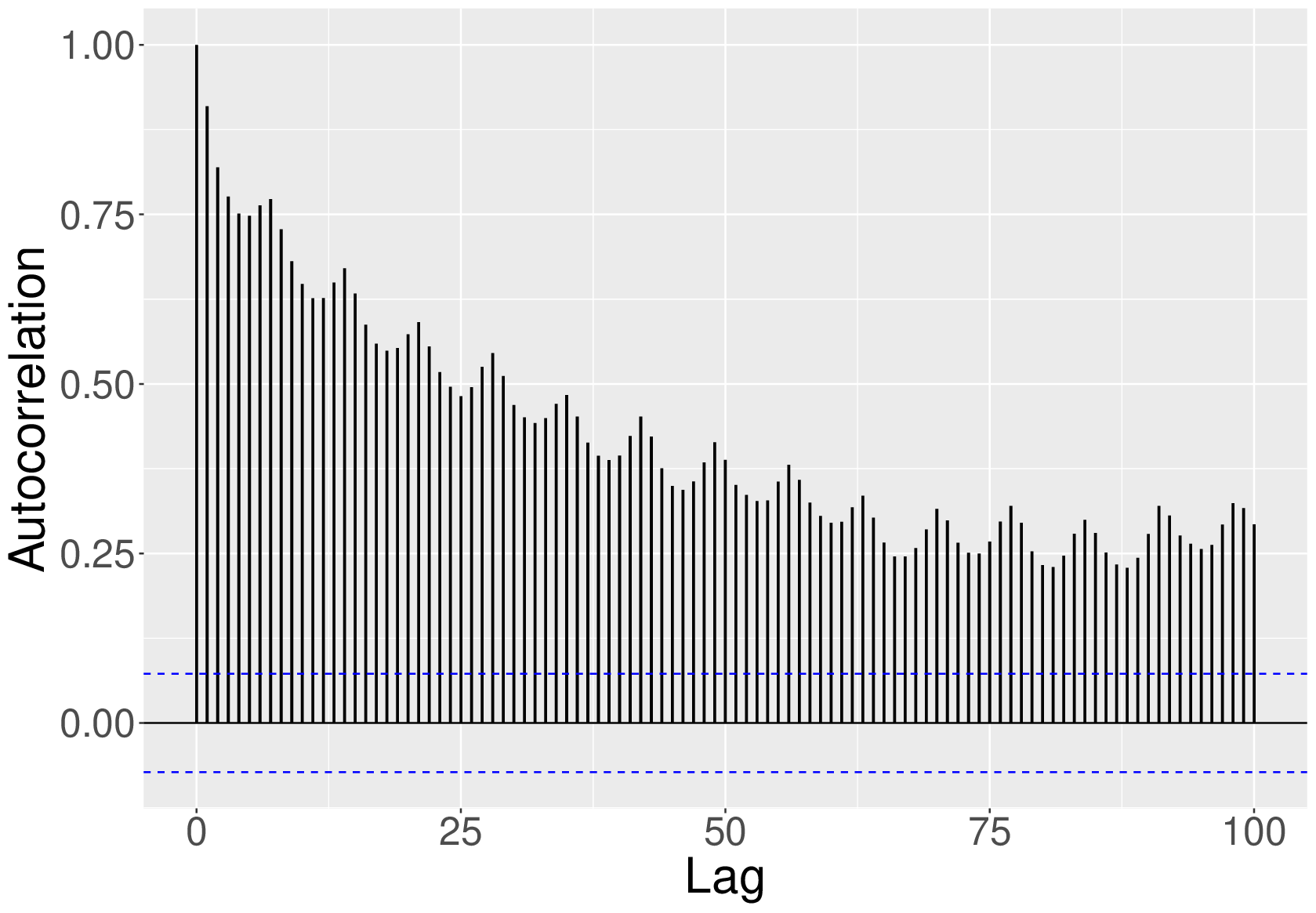}} 	\\
		\subfloat[Empirical density and scaled histogram of TS1	\label{fig:hist1}]{	\includegraphics[scale=0.25, angle=270]{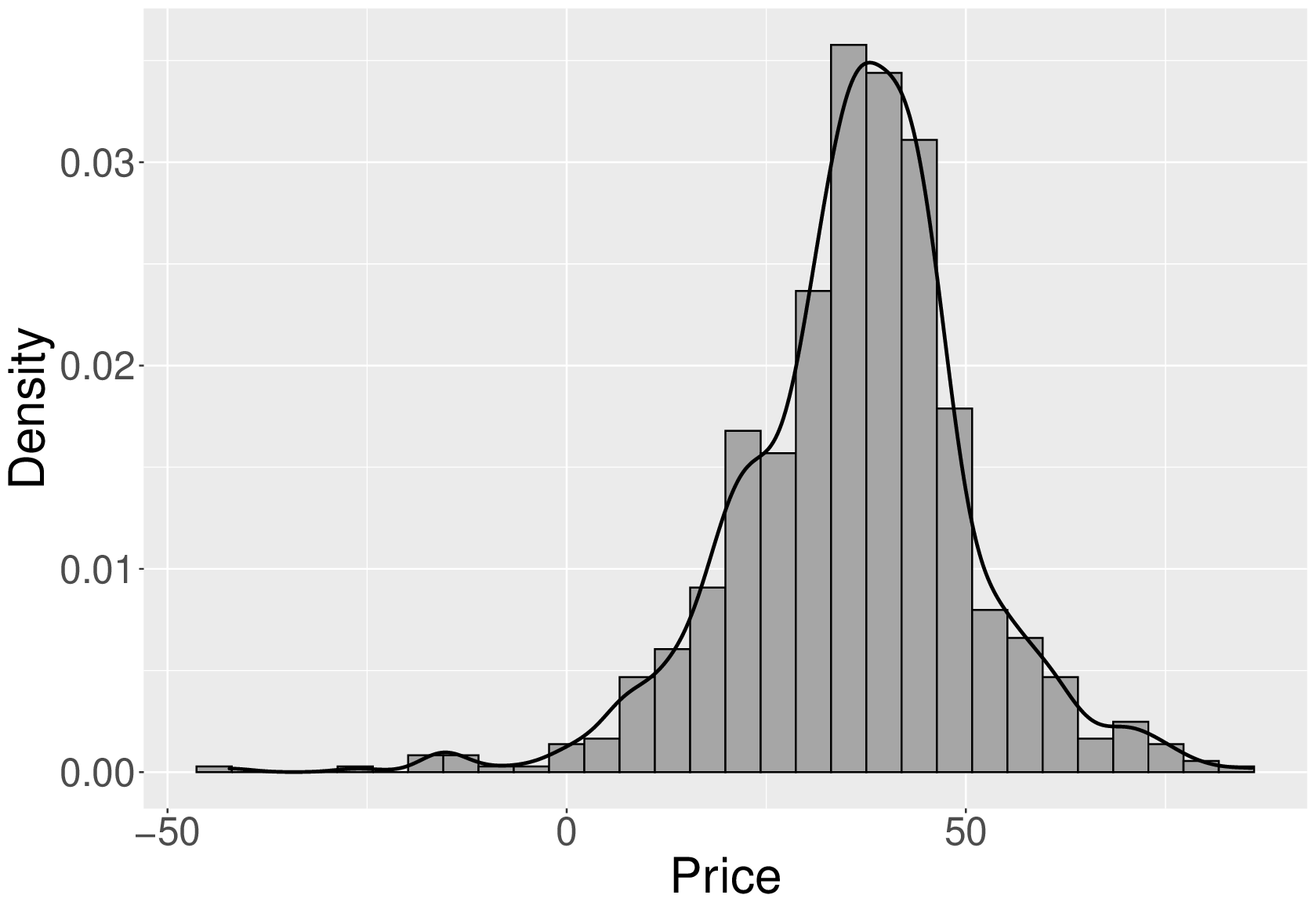} } 
	\captionsetup[subfigure]{aboveskip=-4pt,belowskip=-4pt}
	\subfloat[Empirical density and scaled histogram of TS2\label{fig:hist2}]{	\includegraphics[scale=0.25, angle=270]{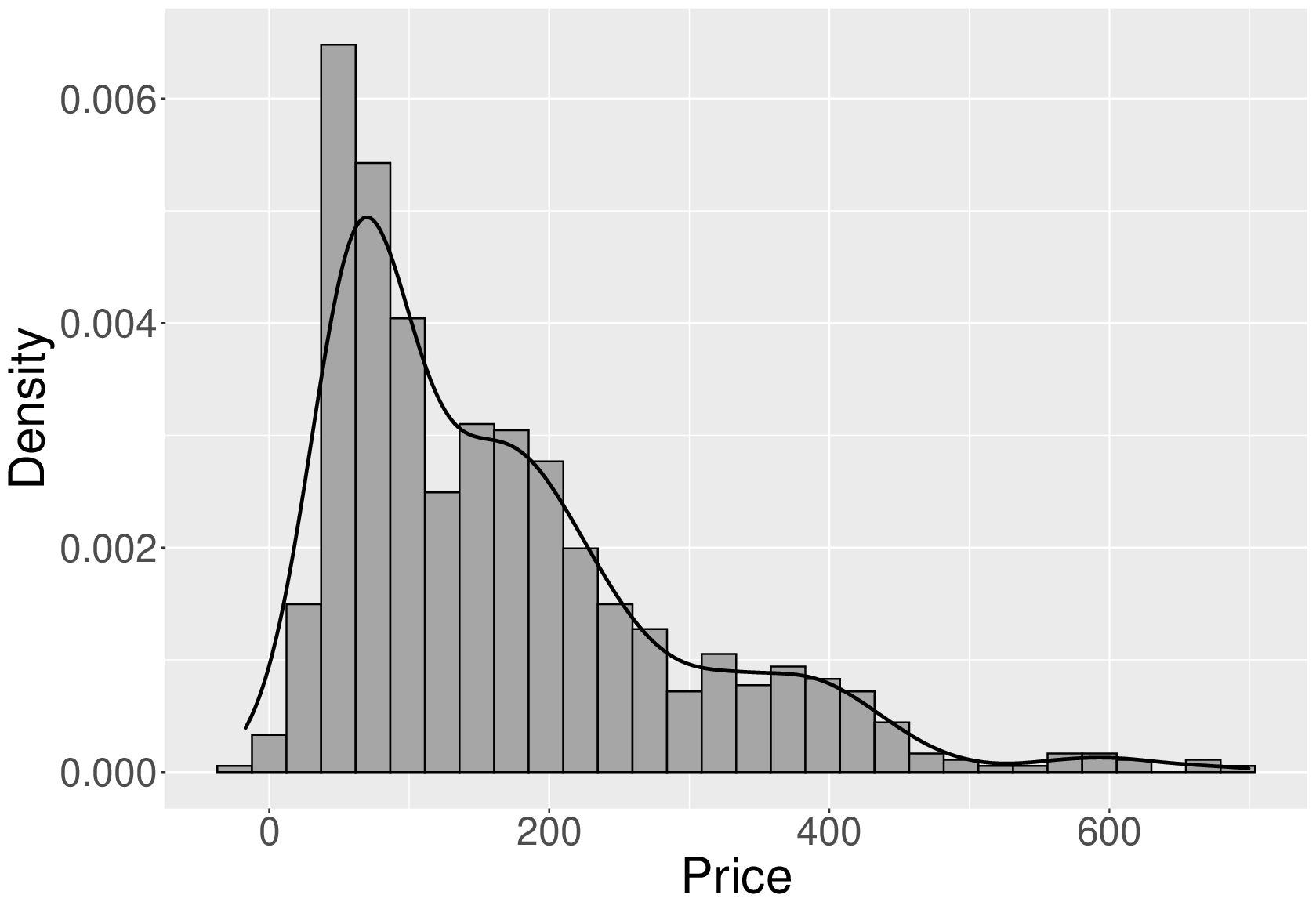}} 	\\
	\caption{\it Exploratory data analysis of the  day-ahead electricity  baseload prices	for Germany and Luxembourg from 1st October 2018 to 1st January 2023 recorded in EUR/MWh, which has been split into two times series (TS1 and TS2).}
	\label{fig:Data12}
\end{figure}

We will now focus on estimating the kernel function of a periodic trawl process.
 First, we consider the relatively calm period from 01.10.2018--31.12.2020  and estimate an exponential trawl function with  period $\tau=7$. We set $\Delta=1$, representing one day. We obtain the following estimates:
 $\widehat{\lambda} =0.055$
 and $(\widehat{c(0)}, \widehat{c(1)}, \ldots, \widehat{c(6)})=(1.000, 0.705, 0.496, 0.451, 0.445, 0.490, 0.642)$.
 The empirical and fitted autocorrelation functions are depicted in Figure \ref{fig:acf-exp-1}.
We note that the exponential decay appears to be too fast to capture the empirical autocorrelation function well. Hence we also fit a supGamma periodic trawl function. We first ran a GMM estimation on the empirical autocorrelation function to estimate $\alpha$ and then set  
$\alpha=1.2$. We then employ our method of moment estimator and obtain  $\widehat{H}=1.264$, indicating a long-memory regime, and  
$(\widehat{c(0)}, \widehat{c(1)}, \ldots, \widehat{c(6)}) = (1.000, 0.789, 0.582, 0.539, 0.534, 0.583, 0.752)$. We note that the model fit looks very good, see Figure \ref{fig:acf-lm-1}. Note that the sensitivity to the particular choice of $\alpha$ appeared low, but this needs to be investigated in more detail in the future. 

We repeat the analysis on the more volatile time period from 01.01.2021--01.01.2023. For the exponential-periodic trawl process we obtain the estimates  $\widehat{\lambda} =0.032$ and 
 $(\widehat{c(0)}, \widehat{c(1)}, \ldots, \widehat{c(6)}) = (1.000, 0.689, 0.473, 0.421, 0.406, 0.437, 0.558)$, which results in a decay which is too fast, see Figure \ref{fig:acf-exp-2}. Again, the fit of the supGamma-periodic trawl model appears better, where we set
 $\alpha = 5.297$ (based on a GMM-estimation) and then find that  $\widehat{H}=1.298$, which as before falls into the long-memory setting, and 
 $(\widehat{c(0)}, \widehat{c(1)}, \ldots, \widehat{c(6)}) = (1.000, 0.958, 0.901, 0.887, 0.888, 0.912, 0.957)$, see Figure \ref{fig:acf-lm-2}.

\begin{figure}[htbp]
	\captionsetup[subfigure]{aboveskip=-4pt,belowskip=-4pt}
	\subfloat[Exponential-periodic trawl for TS1	\label{fig:acf-exp-1}]{	\includegraphics[scale=0.25, angle=270]{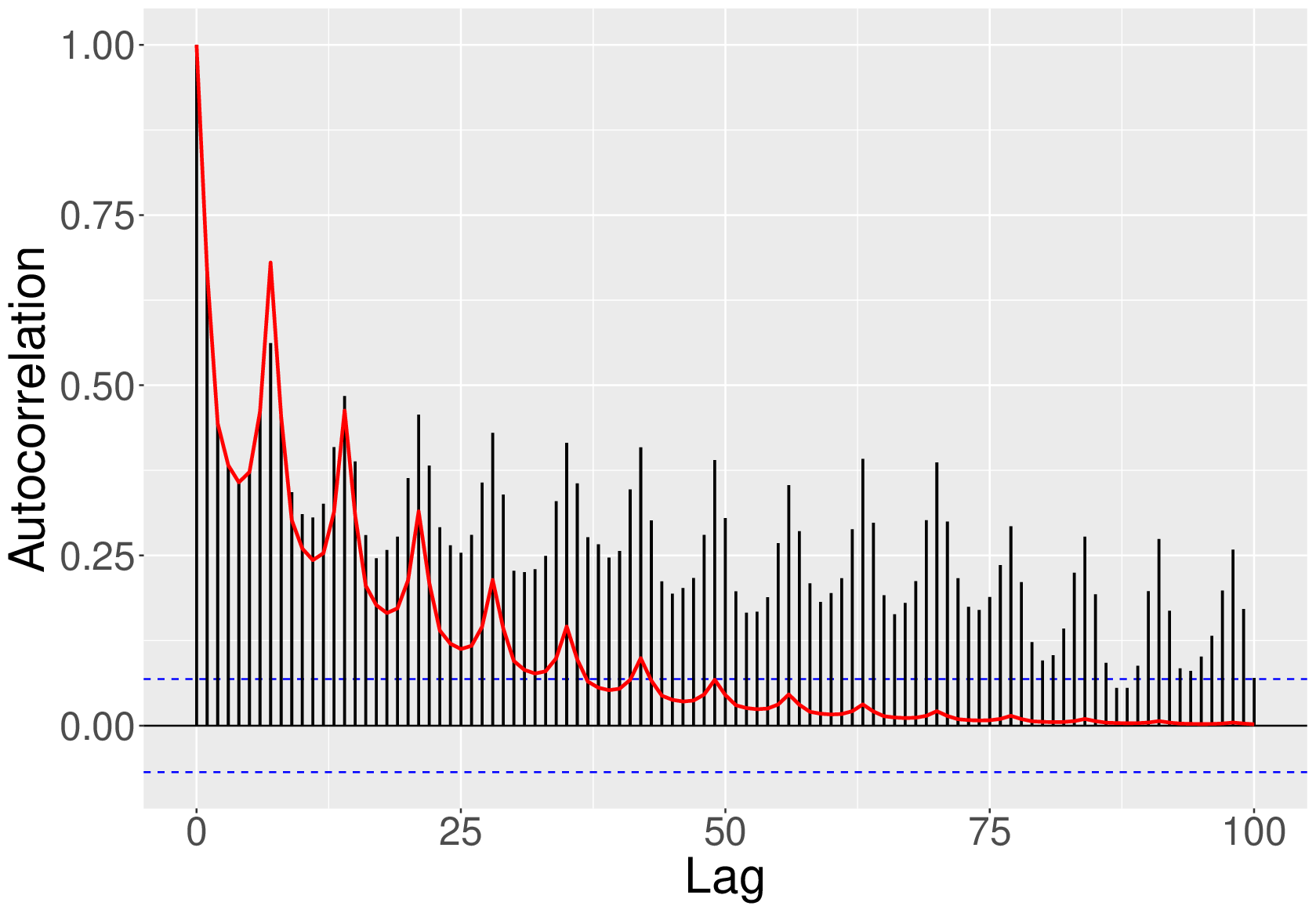} } 
	\captionsetup[subfigure]{aboveskip=-4pt,belowskip=-4pt}
	\subfloat[SupGamma-periodic trawl for TS1\label{fig:acf-lm-1}]{	\includegraphics[scale=0.25, angle=270]{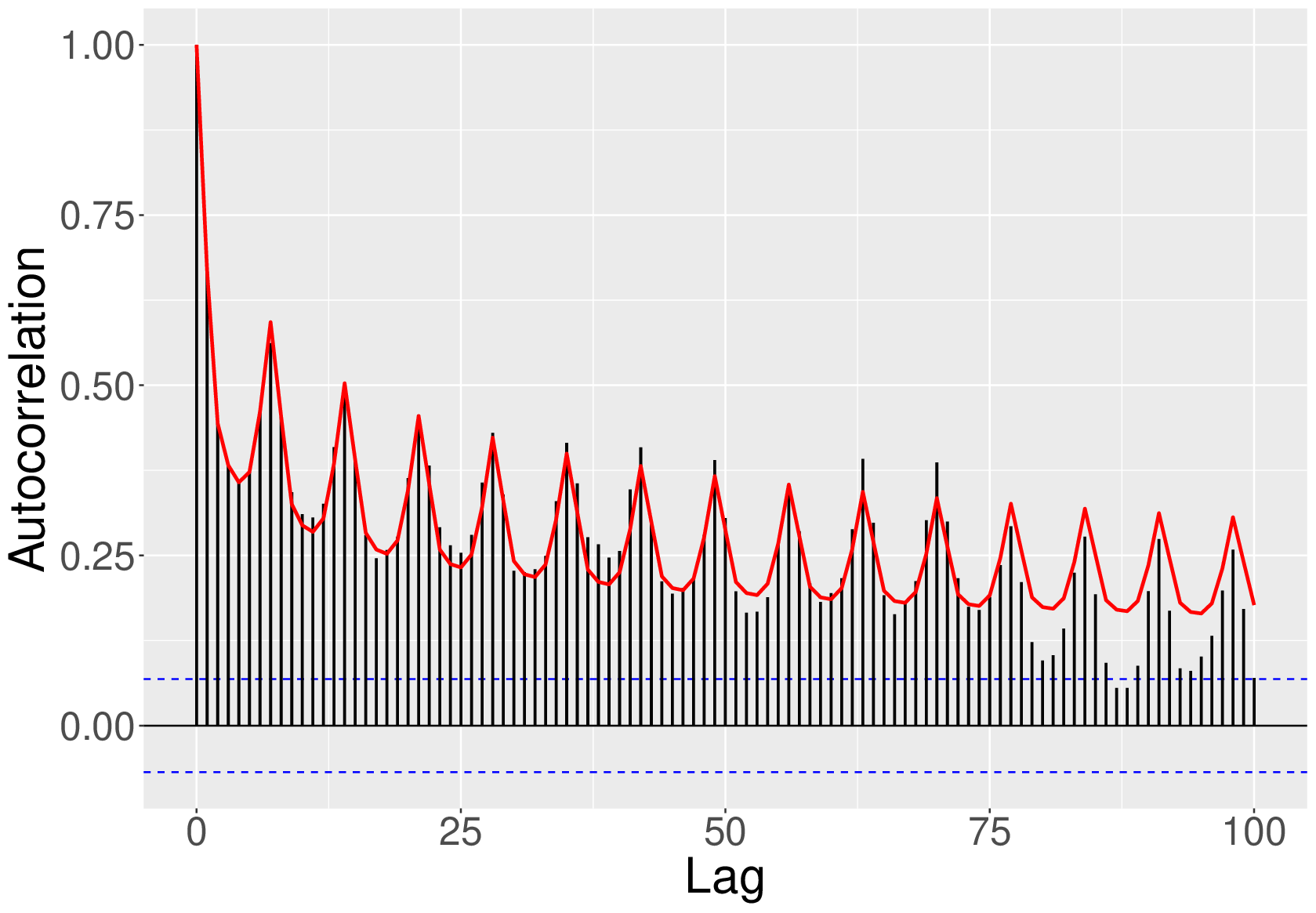}} 	\\
		\subfloat[Exponential-periodic trawl for TS2\label{fig:acf-exp-2}]{	\includegraphics[scale=0.25, angle=270]{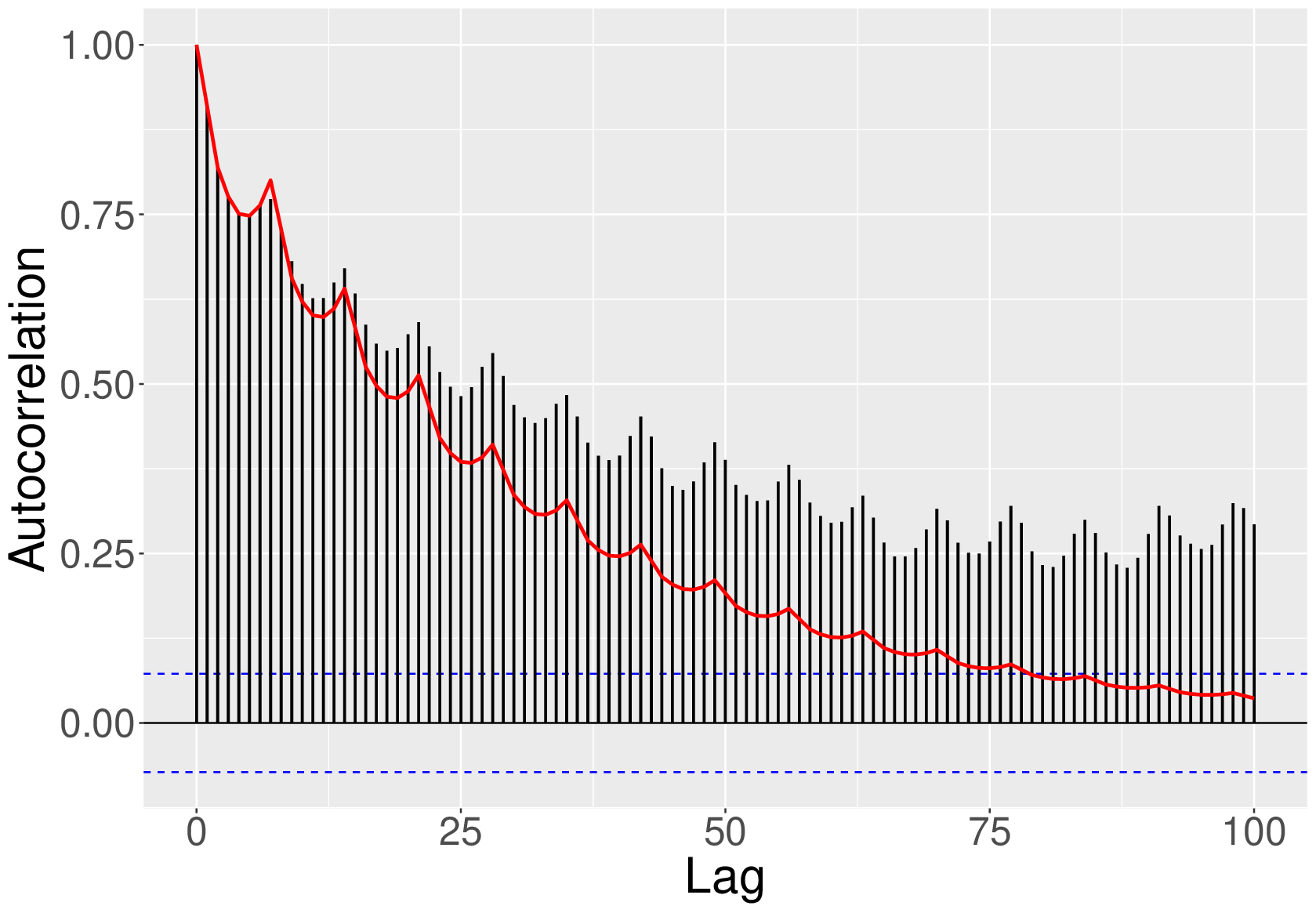} } 
	\captionsetup[subfigure]{aboveskip=-4pt,belowskip=-4pt}
	\subfloat[SupGamma-periodic trawl for TS2\label{fig:acf-lm-2}]{	\includegraphics[scale=0.25, angle=270]{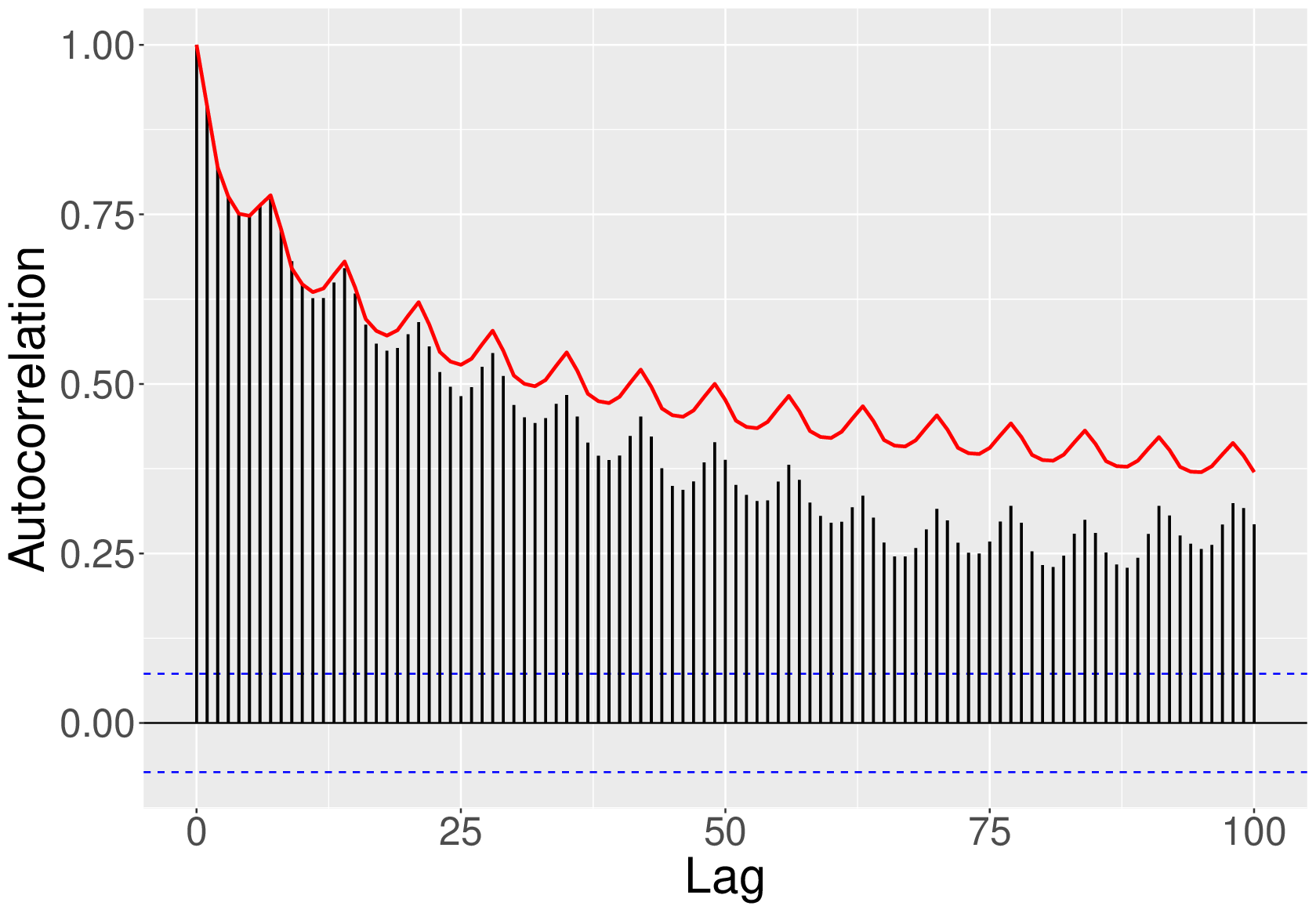}} 	
	\caption{\it Empirical and fitted exponential-periodic and supGamma periodic  trawl for the time series  TS1 and TS2.}
	\label{fig:Data}
\end{figure}

This small illustration demonstrates how easily the kernel function of a periodic trawl process can be fitted to periodic autocorrelation functions. 

In future work, it will be interesting to further extend and fine-tune the methodology, in particular also the GMM methodology derived in Section \ref{sec:GMM} to also estimate the parameters of the driving  L\'{e}vy noise (under suitable identifiability conditions) in practical settings. Here, a multistep approach, where the parameters of the kernel functions are inferred first (as in this article), followed by the parameters of the L\'{e}vy basis, might be a useful direction to investigate further. We also remark that in cases when the period is not known, preliminary simulation experiments not reported here indicate that the period can be estimated using a smoothed empirical periodogram, which is in line with existing literature.

It will also be worthwhile to explore extensions beyond the stationary setting, by for instance considering a linear combination of an  additive seasonal function and a periodic trawl process to allow for non-stationarities in the mean.
\section{Conclusion and outlook}\label{sec:conclusion}
This article introduced an alternative definition of periodic trawl processes compared to the one given in \cite{BNLSV2014}, which appears to have slightly higher analytical tractability.
We derived some of the key probabilistic properties of periodic trawl processes and have studied relevant examples including both short- and long-memory settings. We showed that a slice method can be used to simulate periodic trawl processes effectively. 
Under suitable technical conditions, which currently only cover short-memory scenarios, we have proved the asymptotic normality of the sample mean, sample autocovariances, sample autocorrelations and the GMM-estimator. 
We have implemented the proposed methodology in {\tt R} and showed that the serial correlation of periodic trawl processes 
describes the empirical one  of electricity spot prices very well. 
These are promising results, which suggest that it will be worthwhile to further advance the theory and statistical methodology for periodic trawl processes. 
For instance, for the GMM approach to be applicable in practice, we need to specify suitable identifiability conditions
for fully parametric settings. It will also be interesting to account for multiple periods, for instance, including  weekly and yearly periodicities, and tailor methods from spectral analysis to such settings. 
Moreover, model selection tools are needed for choosing 
the appropriate L\'{e}vy basis. For the particular application considered here, we note that our framework allows for a smooth
(Gaussian) noise terms as well as for jumps, where the latter can model the positive and even more dominant negative jumps, as in \cite{VV2014}, but also the more volatile behaviour observed in the most recent electricity prices. That said, in future research, we should also investigate whether including a stochastic volatility term, either via temporal or spatial stochastic  scaling of the periodic trawl process is needed to describe the data even better, which would bring us into the general framework of ambit fields and processes, see \cite{BNBV2018_book}. 

From a theoretical point of view, it will be interesting to develop an asymptotic theory for periodic trawl processes which can allow for long memory settings. We would then need to check how well such an asymptotic theory works in finite samples to decide whether confidence bounds for the obtained parameters can be constructed based on asymptotic guarantees or whether bootstrap approaches might be more useful in practice.

\begin{appendix}
\section{Appendix}\label{ap:proofs}
The appendix contains the proofs of all the technical results presented in the main paper, additional examples and a discussion of when the technical assumptions needed in our main theorems hold for periodic trawl processes.  
\subsection{Proof of the second order properties}\label{ap:SP}
First, we derive the joint characteristic/cumulant function.
\begin{proposition} Let $t_1<t_2$ and $\theta_1, \theta_2 \in \R$. Then 
	\begin{align*}
		&\Log(\mathbb{E}(\exp(i (\theta_1, \theta_2)(Y_{t_1}, Y_{t_2})^{\top})))=\Log(\mathbb{E}(\exp(i \theta_1 Y_{t_1} + i\theta_2 Y_{t_2})))
		\\
		&= \int_{(-\infty,t_1]\times \R}C_{L'}( \theta_1 p(t_1-s)\ind_{(0,g(t_1-s))}(x)+\theta_2 p(t_2-s)\ind_{(0,g(t_2-s))}(x))dx ds
		\\
		&+
		\int_{(t_1,t_2]\times \R}C_{L'}(\theta_2 p(t_2-s)\ind_{(0,g(t_2-s))}(x))dx ds,
	\end{align*}
	where $C_{L'}$ denotes the cumulant function of the L\'{e}vy seed $L'$ associated with the L\'{e}vy basis $L$.
\end{proposition}
\begin{proof}
	Let $t_1<t_2$ and $\theta_1, \theta_2 \in \R$. Then the joint characteristic function is given by 
	\begin{align*}
		&\mathbb{E}(\exp(i (\theta_1, \theta_2)(Y_{t_1}, Y_{t_2})^{\top}))=\mathbb{E}(\exp(i \theta_1 Y_{t_1} + i\theta_2 Y_{t_2}))
		\\
		&=\mathbb{E}\left[\exp\left(i \theta_1 \int_{(-\infty,t_1]\times \R}p(t_1-s)\ind_{(0,g(t_1-s))}(x)L(dx,ds)
		\right.\right.\\
		& \left.\left.+ i\theta_2 \int_{(-\infty,t_1]\times \R}p(t_2-s)\ind_{(0,g(t_2-s))}(x)L(dx,ds)+i\theta_2\int_{(t_1,t_2]\times \R}p(t_2-s)\ind_{(0,g(t_2-s))}(x)L(dx,ds)\right)\right]\\
		&=\mathbb{E}\left[\exp\left( i\int_{(-\infty,t_1]\times \R}\{ \theta_1 p(t_1-s)\ind_{(0,g(t_1-s))}(x)+\theta_2 p(t_2-s)\ind_{(0,g(t_2-s))}(x)\}L(dx,ds)
		\right.\right.\\
		& \left.\left.+i\theta_2\int_{(t_1,t_2]\times \R}p(t_2-s)\ind_{(0,g(t_2-s))}(x)L(dx,ds)\right)\right]\\
		&=\exp\left(
		\int_{(-\infty,t_1]\times \R}C( \theta_1 p(t_1-s)\ind_{(0,g(t_1-s))}(x)+\theta_2 p(t_2-s)\ind_{(0,g(t_2-s))}(x); L')dx ds
		\right)\\
		& \quad \cdot
		\exp\left(
		\int_{(t_1,t_2]\times \R}C(\theta_2 p(t_2-s)\ind_{(0,g(t_2-s))}(x); L')dx ds
		\right).
	\end{align*}
	I.e.
	\begin{align*}
		&\log(\mathbb{E}(\exp(i \theta_1 Y_{t_1} + i\theta_2 Y_{t_2})))
		\\
		&= \int_{(-\infty,t_1]\times \R}C( \theta_1 p(t_1-s)\ind_{(0,g(t_1-s))}(x)+\theta_2 p(t_2-s)\ind_{(0,g(t_2-s))}(x); L')dx ds
		\\
		&+
		\int_{(t_1,t_2]\times \R}C(\theta_2 p(t_2-s)\ind_{(0,g(t_2-s))}(x); L')dx ds.
	\end{align*}
\end{proof}
We can now easily derive the second-order properties of the periodic trawl process:
\begin{proof}[Proof of Proposition \ref{prop:SecondOrder}] 
	For $t, t_1, t_2\in \R, t_1<t_2$, we have
\begin{align*}
	\E(Y_t)&= \E(L')\int_{-\infty}^{t}p(t-s) g(t-s)ds=\E(L')\int_{0}^{\infty}p(u) g(u)du,\\
		\Var(Y_{t})&= \Var(L')\int_{0}^{\infty}p^2(u)g(u)du,\\
	\Cov(Y_{t_1},Y_{t_2})&=-\left.\frac{\partial^2}{\partial \theta_1 \partial \theta_2}\log(\E(\exp(i\theta_1 Y_{t_1}+i\theta_2 Y_{t_2})))\right|_{\theta_1=\theta_2=0}
	\\
	&= \Var(L')\int_{-\infty}^{t_1}p(t_1-s)p(t_2-s)\min(g(t_1-s),g(t_2-s))ds,\\
	\Cor(Y_{t_1},Y_{t_2})
	&= \frac{\int_{-\infty}^{t_1}p(t_1-s)p(t_2-s)\min(g(t_1-s),g(t_2-s))ds}{\int_{0}^{\infty}p^2(u)g(u)du}.
\end{align*}
Recall that we assume that $g$ is monotonically decreasing, i.e.~if $x\leq y$, then $g(x)\geq g(y)$. 

Since $t_1<t_2$, we have $t_1-s<t_2-s$  for $s<t_1$ and 
\begin{align*}
	\min(g(t_1-s),g(t_2-s))=g(t_2-s),
\end{align*}
hence 
the above expressions simplify to 
\begin{align*}
	\Cov(Y_{t_1},Y_{t_2})
	&=\Var(L')\int_{-\infty}^{t_1}p(t_1-s)p(t_2-s)g(t_2-s)ds\\
	&=\Var(L')\int_{0}^{\infty}p(u)p(t_2-t_1+u)g(t_2-t_1+u)du,\\
	\Cor(Y_{t_1},Y_{t_2})
	&= \frac{\int_{0}^{\infty}p(u)p(t_2-t_1+u)g(t_2-t_1+u)du}{\int_{0}^{\infty}p^2(u)g(u)du}.
\end{align*}
\end{proof}

\begin{proof}[Proof of Proposition \ref{prop:Cor}]
	Recall that 
	\begin{align*}
		\Cor(Y_{0},Y_{t})=
	 \frac{\int_{0}^{\infty}p(u)p(t+u)g(t+u)du}{\int_{0}^{\infty}p^2(u)g(u)du}.
	\end{align*}
We consider a constant  $M>\tau$.  
Since $p$ is periodic with period $\tau$, there exist $\xi_1, \xi_2 \in [0,\tau]$ such that, 
 \begin{align*}
	\int_{0}^{M}p(u)p(t+u)g(t+u)du&=p(\xi_1)p(\xi_1+t)\int_{0}^{M}g(t+u)du,\\
	\int_{0}^{M}(p(u))^2g(u)du&=(p(\xi_2))^2\int_{0}^{M}g(t+u)du,
\end{align*}
by the mean value theorem. 
We note that $p$ is assumed to be continuous, and since it is also periodic, it is bounded. 
Also, the integrability conditions in \eqref{eq:intcondMA} guarantee the existence of the integrals when taking the limit as $M\to \infty$. 
Taking the limit  and setting $c(t)=p(\xi_1)p(\xi_1+t)/(p(\xi_2))^2$ leads the result; since $Cor(Y_0,Y_0)=1$, we deduce that $c(0)=1$. Also, we observe that  $c$ is proportional to the $\tau$-periodic function $p$ and is hence $\tau$-periodic itself.

\end{proof}

\begin{remark}\label{rem:alt-trawl}
As mentioned in Remark \ref{rem:alt-trawl-main}, 
	\cite{BNLSV2014}
 proposed adding a periodic function as a multiplicative factor to $g$ rather than as kernel function as in \eqref{eq:PT}, which results in a process
$(Z_t)_{t\geq 0}$ with 
\begin{align}\label{eq:PT2}
	Z_t 
	&=
	\int_{\R\times \R}\ind_{(0,p(t-s)g(t-s))}(x)\ind_{[0,\infty)}(t-s) L(dx,ds),
\end{align}
compared to our earlier definition of 
$(Y_t)_{t\geq 0}$ with 
\begin{align*}
	Y_t 
	&=
	\int_{\R\times \R}p(t-s)\ind_{(0,g(t-s))}(x)\ind_{[0,\infty)}(t-s) L(dx,ds).
\end{align*}
The autocorrelation function of the process $Z$ is of the form, for $t_1<t_2$,
\begin{align*}
\Cor(Z_{t_1},Z_{t_2})
	&= \frac{\int_{-\infty}^{t_1}\min(p(t_1-s)g(t_1-s), p(t_2-s)g(t_2-s))ds}{\int_{0}^{\infty}p(u)g(u)du},
 \end{align*}
 which is potentially slightly more difficult to deal with than the autocorrelation function of our proposed periodic trawl  process $Y$.
\end{remark}

\subsection{Proofs of the asymptotic theory}\label{ap:proofsgentheory}
The following proofs extend the ideas presented in the work by   \cite{CohenLindner2013}. Alternatively, we could have deduced the results from the more recent work by \cite{CuratoStelzer2019}. 
\begin{proof}[Proof of Theorem \ref{thm:samplemean}]
The proof is a straightforward extension of the arguments provided in the proof of Theorem 2.1 in \cite{CohenLindner2013}. For the convenience of the reader and to keep this article self-contained, we will present the steps to extend the proof by \cite{CohenLindner2013} to our more general setting of mixed moving average processes driven by  homogeneous L\'{e}vy bases. 

First of all, we continue the function $F_{\Delta}$ periodically on $\R$ by setting
\begin{align*}
F_{\Delta}(x,u)=\sum_{j=-\infty}^{\infty}|f(x, u+j\Delta)|, \quad u \in \R, 
\end{align*}
where $F_{\Delta}(x,u)=F_{\Delta}(x,u+j\Delta)$ for all $j \in \mathbb{Z}$, $u, x \in \R$.

We note that the autocovariance function of $Y$ satisfies
\begin{align*}
|\gamma_f(j\Delta)| \leq \kappa_2\int_{\R \times \R}|f(x, -s)||f(x, j\Delta-s)|dx ds,
\end{align*}
for any $j \in \mathbb{Z}$
and
\begin{align}\begin{split}\label{eq:bound}
\frac{1}{\kappa_2}\sum_{j=-\infty}^{\infty}|\gamma_f(j \Delta)|
&\leq
\sum_{j=-\infty}^{\infty}\int_{\R \times \R}|f(x, -s)||f(x, j\Delta-s)|dx ds\\
&\leq 
\int_{\R \times \R}|f(x, -s)|\sum_{j=-\infty}^{\infty}|f(x, j\Delta-s)|dx ds\\
&=\int_{\R \times \R}|f(x, -s)|F_{\Delta}(x, -s)dx ds
\\
&=\int_{\R \times \R}|f(x, s)|F_{\Delta}(x, s)dx ds
\\
&=\sum_{j=-\infty}^{\infty}\int_{\R \times [0,\Delta]}|f(x, j\Delta +s)|F_{\Delta}(x, s)dx ds
\\
&=\int_{\R \times [0,\Delta]} \sum_{j=-\infty}^{\infty}|f(x, j\Delta +s)|F_{\Delta}(x, s)dx ds\\
&=\int_{\R \times [0,\Delta]} F_{\Delta}^2(x, s)dx ds< \infty.
\end{split}
\end{align}
The above computations can be repeated without the modulus, which implies that
$\sum_{j=-\infty}^{\infty}\gamma_f(j \Delta)=V_{\Delta}$.

To simplify the exposition, we shall now assume that $\mu=0$. We proceed as in \cite{CohenLindner2013}.
Define the function
$f_{m;\Delta}(x,s):=f(x,s)\mathbb{I}_{(-m \Delta, m \Delta)}(s)$,
for $m \in \N, x, s \in \R$, and set
\begin{align*}
Y_{j; \Delta}^m&:=\int_{\R \times \R}f_{m; \Delta}(x,s)L(dx, ds)=
\int_{\R \times (-m \Delta, m \Delta)}f(x, s) L(dx, ds)
\\
&=\int_{\R \times ((-m+j)\Delta, (m+j)\Delta)}f(x, j\Delta-s) L(dx, ds).
\end{align*} 
Since $L$ is independently scattered, we can deduce that $(Y_{j;\Delta}^{(m)})_{j \in \mathbb{Z}}$ is a $(2m-1)$-dependent sequence, which is also strictly stationary.
Hence, by \citet[Theorem 6.4.2]{BrockwellDavis1987}, we know that
\begin{align*}
\sqrt{n} \; \overline{Y}_{n; \Delta}^{(m)}  = n^{-1/2}\sum_{j=1}^n Y_{j;\Delta}^{(m)}\stackrel{\mathrm{d}}{\to}Z^{(m)}_{\Delta},
\quad \mathrm{as} \; n \to \infty,
\end{align*}
where the random variable $Z^{(m)}_{\Delta}$ satisfies
$Z^{(m)}_{\Delta}\stackrel{\mathrm{d}}{=}\mathrm{N}(0, V^{(m)})$, where
\begin{align*}
V^{(m)}_{\Delta}=\sum_{j=-2m}^{2m}\gamma_{f_m}(j\Delta),
\end{align*}
for 
\begin{align*}
\gamma_{f_m}(j\Delta) &=\Cov(Y_{0;\Delta}^{(m)}, Y_{j;\Delta}^{(m)})=\kappa_2\int_{\R \times \R}f_{m;\Delta}(x, -s)f_{m;\Delta}(x, j\Delta-s)dx ds\\
&=\int_{\R \times ((-m+j)\Delta, (m+j)\Delta)}f(x, -s) f(x, j\Delta-s) dx ds.
\end{align*}
We observe that $\lim_{m\to \infty}\gamma_{f_m}(j\Delta)=\gamma_{f}(j\Delta)$ for all $j \in \mathbb{Z}$; also
\begin{align*}
|\gamma_{f_m}(j\Delta)|\leq \kappa_2  \int_{\R \times \R}|f(x, -s)| |f(x, j\Delta-s)| dx ds,
\end{align*}
and
$\sum_{j=-\infty}^{\infty} \int_{\R \times \R}|f(x, -s)| |f(x, j\Delta-s)|<\infty$ by the computations in  \eqref{eq:bound}.
Hence, Lebesgue's Dominated Convergence Theorem implies that $\lim_{m\to \infty}V^{(m)}_{\Delta}=V_{\Delta}$ and we get that
$Z_{\Delta}^{(m)}\stackrel{\mathrm{d}}{\to} Z_{\Delta}$, where
$Z \stackrel{\mathrm{d}}{=} \mathrm{N}(0, V_{\Delta})$.

It remains to control the difference $n^{1/2}(\overline{Y}_{n; \Delta}-\overline{Y}_{n; \Delta}^{(m)})$. We argue as follows. Using similar arguments as above, we note that
$\lim_{m\to \infty}\sum_{j=-\infty}^{\infty}\gamma_{f-f_{m;\Delta}}(j\Delta) =0$. Hence, we have that
\begin{align*}
&\lim_{m\to \infty} \lim_{n\to \infty} \Var(n^{1/2}(\overline{Y}_{n; \Delta}-\overline{Y}_{n; \Delta}^{(m)}))
\\
&=\lim_{m\to \infty} \lim_{n\to \infty} n\Var\left(n^{-1}\sum_{j=1}^n\int_{\R \times \R}(f(x, j\Delta-s)-f_{m;\Delta}(x, j\Delta-s))L(dx, ds)\right)\\
&\stackrel{(\star)}{=}\lim_{m\to \infty} \sum_{j=-\infty}^{\infty} \gamma_{f-f_{m;\Delta}}(j\Delta)=0,
\end{align*}
where the equality $(\star)$ follows from \citet[Theorem 7.1.1]{BrockwellDavis1987}.
Chebychef's inequality allows us to conclude that, for any $\epsilon >0$, 
\begin{align*}
\lim_{m\to \infty} \limsup_{n\to \infty} \mathbb{P}(n^{1/2}|\overline{Y}_{n; \Delta}-\overline{Y}_{n; \Delta}^{(m)}|>\epsilon)=0. 
\end{align*}
As stated in 
\cite{CohenLindner2013}, the final step of the proof consists of an application of a Slutsky-type theorem as presented in \citet[Proposition 6.3.9]{BrockwellDavis1987}.
\end{proof}

\begin{proof}[Proof of Lemma \ref{lem:fourthmoment}]
For $t_1, t_2, t_3, t_4 \in \R$, we have, for any $a_1, a_2, a_3, a_4 \in \R$,
the following expression for the joint characteristic function
\begin{align*}
&\psi((a_1,a_2,a_3,a_4);(Y_{t_1},Y_{t_2},Y_{t_3},Y_{t_4})):=
\mathbb{E}(\exp(i(a_1Y_{t_1}+a_2Y_{t_2}+a_3Y_{t_3}+a_4Y_{t_4}))\\
&=\mathbb{E}\left[\exp\left(i \int_{\R\times \R}\left(\sum_{j=1}^4a_jf(x, t_j-s)\right)L(ds, dx)\right) \right]\\
&=
\exp\left[ 
\int_{\R \times \R} C\left(\sum_{j=1}^4a_jf(x, t_j-s); L' \right) dx ds
\right],
\end{align*}
where $C(\cdot; L')$ denotes the cumulant function of the L\'{e}vy seed $L'$, which we will present next.

Suppose $L'$ has characteristic triplet $(c, A, \nu)$ w.r.t.~the truncation function $\tau(y)=\mathbb{I}_{[-1,1]}(y)$. I.e.~we have the following representation for its characteristic function, for any $\theta \in \R$,
\begin{align*}
\mathbb{E}(\exp(i \theta L'))=\exp\left(i c \theta -\frac{1}{2}A \theta^2 + \int_{\R }(e^{iy \theta}-1-i\theta y \tau(y))\nu(dy)\right).
\end{align*}
We recall that $\mathbb{E}(L')=c +\int_{\R}y(1-\tau(y))\nu(dy)$. Since we are assuming that $\mathbb{E}(L')=0$, we get that $c=-\int_{\R}y(1-\tau(y))\nu(dy)=-\int_{\R}y\mathbb{I}_{[-1,1]^c}(y)\nu(dy)$ and, hence, 
\begin{align*}
\mathbb{E}(\exp(i \theta L'))=\exp\left( -\frac{1}{2}A \theta^2 + \int_{\R }(e^{iy \theta}-1-i\theta y )\nu(dy)\right).
\end{align*}
I.e.~the corresponding cumulant function is given by
\begin{align*}
C(\theta; L')= -\frac{1}{2}A \theta^2 + \int_{\R }(e^{iy \theta}-1-i\theta y )\nu(dy).
\end{align*}
Moreover,
\begin{align*}
&C((a_1,a_2,a_3,a_4);(Y_{t_1},Y_{t_2},Y_{t_3},Y_{t_4}))
:=\int_{\R \times \R} C\left(\sum_{j=1}^4a_jf(x, t_j-s); L' \right) dx ds\\
&=
-\frac{1}{2}A 
\int_{\R \times \R}
\left(\sum_{j=1}^4a_jf(x, t_j-s)\right)^2
dx ds
\\
& \quad +\int_{\R \times \R} \int_{\R }\left(e^{iy \sum_{j=1}^4a_jf(x, t_j-s)}-1-i\sum_{j=1}^4a_jf(x, t_j-s) y \right)\nu(dy)
dx ds\\
&=
-\frac{1}{2}A 
\sum_{j,k=1}^4 a_j a_k \int_{\R \times \R} f(x, t_j-s) f(x, t_k-s)
dx ds
\\
& \quad +\int_{\R \times \R} \int_{\R }\left(e^{iy \sum_{j=1}^4a_jf(x, t_j-s)}-1-i\sum_{j=1}^4a_jf(x, t_j-s) y \right)\nu(dy)
dx ds\\
&=
-\frac{1}{2}A 
\sum_{j=1}^4 a_j^2 \int_{\R \times \R} f^2(x, t_j-s) 
dx ds
-\frac{1}{2}A 
\sum_{j,k=1, j\not = k}^4 a_j a_k \int_{\R \times \R} f(x, t_j-s) f(x, t_k-s)
dx ds
\\
& \quad +\int_{\R \times \R} \int_{\R }\left(e^{iy \sum_{j=1}^4a_jf(x, t_j-s)}-1-i\sum_{j=1}^4a_jf(x, t_j-s) y \right)\nu(dy)
dx ds.
\end{align*}

Next, we compute the fourth moments, where we recall that 
\begin{align*}
\left.\frac{\partial^4}{\partial a_1 \partial a_2 \partial a_3 \partial a_4}\psi((a_1,a_2,a_3,a_4);(Y_{t_1},Y_{t_2},Y_{t_3},Y_{t_4}))\right|_{a_1=a_2=a_3=a_4=0}= \mathbb{E}(Y_{t_1}Y_{t_2}Y_{t_3}Y_{t_4}).
\end{align*}
We now abbreviate the functions to $\psi$ and $C$ without stating their arguments and a subscript denotes the corresponding partial derivative, e.g.~$C_{a_1}=\frac{\partial}{\partial a_1}C((a_1,a_2,a_3,a_4);(Y_{t_1},Y_{t_2},Y_{t_3},Y_{t_4})$ and similarly for higher order partial derivatives.
Since $\psi = \exp(C)$, we have
\begin{align*}
\psi_{a1} &= \psi C_{a_1},\\
\psi_{a_1,a_2}
&= \psi [ C_{a_1, a_2}+ C_{a_1} C_{a_2}],\\
\psi_{a_1,a_2,a_3}
&= \psi [ (C_{a_1, a_2}+ C_{a_1} C_{a_2}) C_{a_3}
+C_{a_1, a_2, a_3} +C_{a_1, a_3}C_{a_2}+C_{a_1}C_{a_2, a_3}
]\\
&= \psi [ C_{a_1} C_{a_2}C_{a_3} 
+C_{a_1}C_{a_2, a_3} 
+C_{a_2} C_{a_1, a_3}
+C_{a_3} C_{a_1, a_2} +  
  +C_{a_1, a_2, a_3}
],\\
\psi_{a_1,a_2,a_3,a_4}
&= \psi [  C_{a_1} C_{a_2}C_{a_3} C_{a_4} 
+C_{a_1}C_{a_2, a_3} C_{a_4} 
+C_{a_2} C_{a_1, a_3}C_{a_4} 
+C_{a_3} C_{a_1, a_2} C_{a_4}   
  +C_{a_1, a_2, a_3}C_{a_4} 
\\
&
+
 C_{a_1, a_4} C_{a_2}C_{a_3} 
 + C_{a_1} C_{a_2, a_4}C_{a_3} 
  + C_{a_1} C_{a_2}C_{a_3, a_4} 
  +
  C_{a_1, a_4}C_{a_2, a_3}
  +
  C_{a_1}C_{a_2, a_3, a_4}
  \\
  &
  +
  C_{a_2, a_4} C_{a_1, a_3}
  +
  C_{a_2} C_{a_1, a_3, a_4}
  +
  C_{a_3,a_4} C_{a_1, a_2}
  +
  C_{a_3} C_{a_1, a_2, a_4}
  +C_{a_1, a_2, a_3, a_4}
].
\end{align*}
Here we have
\begin{align*}
C_{a_i}|_{a_1=a_2=a_3=a_4=0}&=0, \quad \mathrm{for}\; i=1, 2, 3, 4,\\
C_{a_i,a_j}|_{a_1=a_2=a_3=a_4=0}
&=-\left(A+ \int_{\R} y^2 \nu(dy)\right) \int_{\R \times \R} f(x, t_i-s) f(x, t_j-s)
dx ds, \quad \mathrm{for}\, i, j = 1, 2, 3, 4,\\
C_{a_1,a_2,a_3,a_4}&=\int_{\R\times \R}\prod_{j=1}^4 f(x, t_j-s)dxds \int_{\R}y^4\nu(dy).
\end{align*}
The above results imply that 
\begin{align*}
&\mathbb{E}(Y_{t_1}Y_{t_2}Y_{t_3}Y_{t_4})\\
&= \left(A+ \int_{\R} y^2 \nu(dy)\right)^2
\\
&\left(
\int_{\R \times \R} f(x, t_1-s) f(x, t_2-s)dx ds 
\int_{\R \times \R} f(x, t_3-s) f(x, t_4-s)dx ds
\right .\\
&
+
\int_{\R \times \R} f(x, t_1-s) f(x, t_3-s)dx ds 
\int_{\R \times \R} f(x, t_2-s) f(x, t_4-s)dx ds \\
&\left.
+\int_{\R \times \R} f(x, t_1-s) f(x, t_4-s)dx ds
\int_{\R \times \R} f(x, t_2-s) f(x, t_3-s)dx ds
\right)\\
& + \int_{\R} y^4 \nu(dy)
\int_{\R \times \R} f(x, t_1-s) f(x, t_2-s)
f(x, t_3-s) f(x, t_4-s) dx ds.
\end{align*}
We note that $\kappa_4:=\int_{\R}y^4\nu(dy)=(\eta-3)\kappa_2^2$ and $\kappa_2=A+\int_{\R}y^2\nu(dy)$

We can further simplify the above formula as follows:
\begin{align*}
&\mathbb{E}(Y_{t_1}Y_{t_2}Y_{t_3}Y_{t_4})\\
&= \left(A+ \int_{\R} y^2 \nu(dy)\right)^2
\\
&\left(
\int_{\R \times \R} f(x, t_1-t_2+s) f(x, s)dx ds 
\int_{\R \times \R} f(x, t_3-t_4+s) f(x, s)dx ds
\right .\\
&
+
\int_{\R \times \R} f(x, t_1-t_3+s) f(x, s)dx ds 
\int_{\R \times \R} f(x, t_2-t_4+s) f(x, s)dx ds \\
&\left.
+\int_{\R \times \R} f(x, t_1-t_4+s) f(x, s)dx ds
\int_{\R \times \R} f(x, t_2-t_3+s) f(x, s)dx ds
\right)\\
& + \int_{\R} y^4 \nu(dy)
\int_{\R \times \R} f(x, t_1-t_3+s) f(x, t_2-t_3)
f(x, s) f(x, t_4-t_3+s) dx ds\\
&= \gamma(t_1-t_2)\gamma(t_3-t_4) 
+
\gamma(t_1-t_3) \gamma(t_2-t_4)
+\gamma(t_1-t_4)\gamma(t_2-t_3) \\
& + \kappa_4
\int_{\R \times \R}f(x, t_1-t_3+s) f(x, t_2-t_3)
f(x, s) f(x, t_4-t_3+s) dx ds.
\end{align*}

\end{proof}

\begin{proof}[Proof of Proposition \ref{prop:covlimit}]
We first expand the covariance of the sample autocovariances as follows
\begin{align*}
\Cov(\widehat{\gamma}_{n;\Delta}^*(\Delta p), \widehat{\gamma}_{n,\Delta}^*(\Delta q))
=\mathbb{E}(\widehat{\gamma}_{n;\Delta}^*(\Delta p) \widehat{\gamma}_{n;\Delta}^*(\Delta q)) - \mathbb{E}(\widehat{\gamma}_{n;\Delta}^*(\Delta p))
\mathbb{E}(\widehat{\gamma}_{n;\Delta}^*(\Delta q)),
\end{align*}
where
\begin{align*}
\mathbb{E}(\widehat{\gamma}_{n;\Delta}^*(\Delta p))
&=\frac{1}{n}\sum_{j=1}^{n}\mathbb{E}(Y_{j\Delta}Y_{(j+p)\Delta})= \frac{1}{n}\sum_{j=1}^{n}\gamma(p\Delta)=\gamma(p\Delta),\\
\mathbb{E}(\widehat{\gamma}_{n;\Delta}^*(\Delta q))
&=\frac{1}{n}\sum_{k=1}^{n}\mathbb{E}(Y_{k\Delta}Y_{(k+q)\Delta})= \frac{1}{n}\sum_{k=1}^{n}\gamma(q\Delta)=\gamma(q\Delta).
\end{align*}
Also,
\begin{align*}
&\mathbb{E}(\widehat{\gamma}_{n;\Delta}^*(\Delta p) \widehat{\gamma}_{n;\Delta}^*(\Delta q))
=\frac{1}{n^2}\sum_{j=1}^{n}\sum_{k=1}^{n}\mathbb{E}(
Y_{j\Delta}Y_{(j+p)\Delta} Y_{k\Delta}Y_{(k+q)\Delta}
),
\end{align*}
where
\begin{align*}
&\mathbb{E}(
Y_{j\Delta}Y_{(j+p)\Delta} Y_{k\Delta}Y_{(k+q)\Delta}
)\\
&=\gamma(p\Delta)\gamma(q\Delta) 
+
\gamma((k-j)\Delta) \gamma((j+p-k-q)\Delta)
+\gamma((j-k-q)\Delta)\gamma((j+p-k)\Delta) \\
& \quad + \kappa_4
\int_{\R \times \R} f(x, (j-k)\Delta+s) f(x, (j-k+p)\Delta+s)
f(x, s) f(x, q\Delta+s) dx ds\\
&=\gamma(p\Delta)\gamma(q\Delta) 
+
\gamma((j-k)\Delta) \gamma((j-k+p-q)\Delta)
+\gamma((j-k-q)\Delta)\gamma((j-k+p)\Delta) \\
& \quad + \kappa_4
\int_{\R \times \R}  f(x, (j-k)\Delta+s) f(x, (j-k+p)\Delta+s)
f(x, s) f(x, q\Delta+s) dx ds\\
&\stackrel{l=j-k}{=}\gamma(p\Delta)\gamma(q\Delta) 
+
\gamma(l\Delta) \gamma((l+p-q)\Delta)
+\gamma((l-q)\Delta)\gamma((l+p)\Delta) \\
& \quad + \kappa_4
\int_{\R \times \R}  f(x, l\Delta+s) f(x, (l+p)\Delta+s)
f(x, s) f(x, q\Delta+s) dx ds.
\end{align*}
Now we subtract $\gamma(p\Delta)\gamma(q\Delta)$, we set $l=j-k$, interchange the order of summation and use the stationarity to obtain
\begin{align*}
&\Cov(\widehat{\gamma}_{n;\Delta}^*(\Delta p), \widehat{\gamma}_{n;\Delta}^*(\Delta q))\\
&=\frac{1}{n^2}\sum_{j=1}^{n}\sum_{k=1}^{n}\mathbb{E}(
Y_{j\Delta}Y_{(j+p)\Delta} Y_{k\Delta}Y_{(k+q)\Delta}
)
-\gamma(p\Delta)\gamma(q\Delta)\\
&
=\frac{1}{n^2}\sum_{j=1}^{n}\sum_{k=1}^{n}\left(
\gamma((j-k)\Delta) \gamma((j-k+p-q)\Delta)
+\gamma((j-k-q)\Delta)\gamma((j-k+p)\Delta) \right.\\
& \quad + \left.\kappa_4
\int_{\R \times \R}  f(x, (j-k)\Delta+s) f(x, (j-k+p)\Delta+s)
f(x, s) f(x, q\Delta+s) dx ds\right)\\
&
=\frac{1}{n^2}\sum_{|l|<n}\sum_{k=1}^{n-|l|}\left(
\gamma(l\Delta) \gamma((l+p-q)\Delta)
+\gamma((l-q)\Delta)\gamma((l+p)\Delta)\right. \\
& \quad + \left. \kappa_4
\int_{\R \times \R}  f(x, l\Delta+s) f(x, (l+p)\Delta+s)
f(x, s) f(x, q\Delta+s) dx ds\right)\\
&
=\frac{1}{n^2}\sum_{|l|<n}(n-|l|)T_{l,p,q;\Delta}
=\frac{1}{n}\sum_{|l|<n}\left(1-\frac{|l|}{n}\right)T_{l,p,q;\Delta},
\end{align*}
where 
\begin{align*}
    T_{l,p,q;\Delta}&:=
\gamma(l\Delta) \gamma((l+p-q)\Delta)
+\gamma((l-q)\Delta)\gamma((l+p)\Delta) \\
& \quad +  \kappa_4
\int_{\R \times \R}  f(x, l\Delta+s) f(x, (l+p)\Delta+s)
f(x, s) f(x, q\Delta+s) dx ds.
\end{align*}
Hence, we have
\begin{align*}
\lim_{n\to \infty}n\Cov(\widehat{\gamma}_{n;\Delta}^*(\Delta p), \widehat{\gamma}_{n;\Delta}^*(\Delta q))
=\lim_{n\to \infty}\sum_{|l|<n}\left(1-\frac{|l|}{n}\right)T_{l,p,q;\Delta}
=\sum_{l=-\infty}^{\infty}T_{l,p,q;\Delta},
\end{align*}
where 
\begin{align*}
\sum_{l=-\infty}^{\infty}T_{l,p,q;\Delta}
&=
\sum_{l=-\infty}^{\infty}[\gamma(l\Delta) \gamma((l+p-q)\Delta)
+\gamma((l-q)\Delta)\gamma((l+p)\Delta)]
\\
&\quad +
\kappa_4
\int_{\R \times \R}\sum_{l=-\infty}^{\infty}
  f(x, l\Delta+s) f(x, (l+p)\Delta+s)
f(x, s) f(x, q\Delta+s) dx ds,
\end{align*}
by the Dominated Convergence Theorem since \eqref{eq:sumgamma2-cov} holds. More precisely, let us justify why $\sum_{l=-\infty}^{\infty}|T_{l,p,q;\Delta}|<\infty$.
The finiteness of 
$\sum_{l=-\infty}^{\infty}|\gamma(l\Delta) \gamma((l+p-q)\Delta)
+\gamma((l-q)\Delta)\gamma((l+p)\Delta)|$ follows from \eqref{eq:sumgamma2-cov}.
For the second term,  for $q \in \mathbb{Z}$, define
\begin{align*}
\left(\widetilde{G}_{q;\Delta}:\mathbb{R}\times [0, \Delta] \to \R, (x, u) \mapsto \widetilde{G}_{q;\Delta}(x,u)=\sum_{j=-\infty}^{\infty}|f(x, u+j\Delta)||f(x, u+(j+q)\Delta)|\right),
\end{align*}
 which is in $L^2(\R \times [0, \Delta])$ 
due to \eqref{as:sumf2b}.  We  consider the periodic continuation of $\widetilde{G}_{q;\Delta}$ and set 
\begin{align*}
\left(\widetilde{G}_{q;\Delta}:\mathbb{R}\times \R \to \R, (x, u) \mapsto \widetilde{G}_{q;\Delta}(x,u)=\sum_{j=-\infty}^{\infty}|f(x, u+j\Delta)||f(x, u+(j+q)\Delta)|\right).
\end{align*}
Since $\widetilde{G}_{q;\Delta}$ is periodic and, restricted to $\mathbb{R}\times [0, \Delta]$ square-integrable, we have
\begin{align*}
&\int_{\R \times \R} \sum_{l=-\infty}^{\infty}
  |f(x, l\Delta+s) f(x, (l+p)\Delta+s)|
|f(x, s) f(x, q\Delta+s)| dx ds\\
&=\int_{\R \times \R}\widetilde{G}_p(s)
|f(x, s) f(x, q\Delta+s)| dx ds
=\sum_{j=-\infty}^{\infty}\int_{\R \times [j\Delta, (j+1)\Delta]}\widetilde{G}_p(s)
|f(x, s) f(x, q\Delta+s)| dx ds\\
&=\sum_{j=-\infty}^{\infty}\int_{\R \times [0, \Delta]}\widetilde{G}_p(s+j\Delta)
|f(x, s+j\Delta) f(x, (q+j)\Delta+s)| dx ds\\
&=\sum_{j=-\infty}^{\infty}\int_{\R \times [0, \Delta]}\widetilde{G}_p(s)
|f(x, s+j\Delta) f(x, (q+j)\Delta+s)| dx ds\\
&=\int_{\R \times [0, \Delta]}\widetilde{G}_p(s)
 \sum_{j=-\infty}^{\infty}|f(x, s+j\Delta) f(x, (q+j)\Delta+s)| dx ds\\
 &=\int_{\R \times [0, \Delta]}\widetilde{G}_p(s)
 \widetilde{G}_q(s) dx ds<\infty.
\end{align*}
Equation \eqref{eq:sumT} follows from the same calculations as above without the modulus sign in the definition of $\widetilde{G}$.

\end{proof}
\begin{proof}[Proof of Theorem \ref{thm:acflimits}]
\begin{enumerate}
    \item 

For a function $f$ with compact support, the result can be deduced as in \citet[Proposition 7.3.2]{BrockwellDavis1987}. The general case can be handled as follows, where we adapt the proof of   \citet[Theorem 3.5]{CohenLindner2013} to our more general setting. 
As in the proof for the sample mean, define the function
$f_{m;\Delta}(x,s):=f(x,s)\mathbb{I}_{(-m \Delta, m \Delta)}(s)$,
for $m \in \N, x, s \in \R$, and set
\begin{align*}
Y_{j; \Delta}^m&:=\int_{\R \times \R}f_{m; \Delta}(x,s)L(dx, ds)
=\int_{\R \times ((-m+j)\Delta, (m+j)\Delta)}f(x, j\Delta-s) L(dx, ds).
\end{align*} 
We denote by $\gamma_m$ the autocovariance function of the process $(Y_{j; \Delta}^m)_{j\in \mathbb{Z}}$. We set
\begin{align*}
\gamma_{n;\Delta;m}^*(p\Delta)=\sum_{j=1}^n Y_{j; \Delta}^mY_{j+p; \Delta}^m, \quad p=0, \ldots, h.
\end{align*}
Then, we have
\begin{align*}
        \sqrt{n}(\gamma_{n;\Delta;m}^*(0)-\gamma_m(0), \ldots, \gamma_{n;\Delta;m}^*(h\Delta)-\gamma_m(h\Delta))^{\top} \stackrel{\mathrm{d}}{\to}Z_{\Delta;m}\sim \mathrm{N}(0,V_{\Delta;m}), \quad n \to \infty,
            \end{align*}
            where the asymptotic covariance matrix is given by $V_{\Delta;m}=(v_{pq; \Delta;m})_{p,q=0,\ldots,h}\in \R^{h+1,h+1}$ with $v_{pq;\Delta;m}$ defined as 
\begin{align*}
    v_{pq;\Delta;m}&:=(\eta-3)\kappa_2^2\int_{\R\times[0,\Delta]}G_{p;\Delta;m}(x,u)G_{q;\Delta;m}(x,u)dxdu \\
&+\sum_{l=-\infty}^{\infty}[\gamma_m(l\Delta)\gamma_m((l+p-q)\Delta)
+
\gamma_m((l-q)\Delta)\gamma_m((l+p)\Delta),
] \\
G_{q;\Delta;m}(x,u)&:=\sum_{j=-\infty}^{\infty}f_m(x, u+j\Delta)f_m(x, u+(j+q)\Delta); \quad u \in [0, \Delta].
\end{align*}

We would like to show that $\lim_{m\to \infty}V_{\Delta;m}=V_{\Delta}$.        
For this, we note that 
\begin{align*}
   G_{q;\Delta;m}(x,u)&=\sum_{j=-\infty}^{\infty}f_m(x, u+j\Delta)f_m(x, u+(j+q)\Delta)\\
   &
  \to  G_{q;\Delta}(x,u)=\sum_{j=-\infty}^{\infty}f(x, u+j\Delta)f(x, u+(j+q)\Delta),
\end{align*}
 uniformly in $u \in [0, \Delta]$, as $m\to \infty$, by Lebesgue's Dominated Convergence Theorem, since the function
 $(x,u)\mapsto \sum_{j=-\infty}^{\infty}|f(x, u+j\Delta)||f(x, u+(j+q)\Delta)|$ is in $L^2(\R \times [0, \Delta])$ by 
\eqref{as:sumf2b} 
 and is therefore almost surely finite.
 Moreover, we note that
 \begin{align*}
|G_{q;\Delta;m}(x,u)|&\leq\sum_{j=-\infty}^{\infty}|f(x, u+j\Delta)||f(x, u+(j+q)\Delta)|, 
 \end{align*}
uniformly in $u$ and $m$. Hence, an application of the Dominated Convergence Theorem leads to $G_{q;\Delta;m}\to G_{q;\Delta}$ in $L^2(\R \times [0, \Delta])$ as $m \to \infty$.
We also note that 
  \begin{align*}
|\gamma_{m}(j\Delta)|&\leq \int_{\R\times \R}|f(x, u)||f(x, u+j\Delta)|dxdu, \quad \forall m \in \mathbb{N}, \forall j \in \mathbb{Z}.
 \end{align*}
 Moreover, $\lim_{m\to \infty}\gamma_m(j\Delta)=\gamma(j\Delta)$ for all $j \in \mathbb{Z}$. 
 Assumption \eqref{as:finsumsquint} together with the Dominated Convergence Theorem allows us to conclude that $(\gamma_m(j\Delta))_{j \in \mathbb{Z}}$ converges in $L^2(\mathbb{Z})$ to $(\gamma(j\Delta))_{j \in \mathbb{Z}}$. Combining this result with our earlier finding of the convergence of $G_{q;\Delta;m}$ implies that $\lim_{m\to \infty}V_{\Delta;m}=V_{\Delta}$ and 
 \begin{align*}
Z_{\Delta;m}\stackrel{\mathrm{d}}{\to} Z_{\Delta}, \quad m \to \infty, 
 \end{align*}
 where $Z_{\Delta}\stackrel{\mathrm{d}}{=}\mathrm{N}(0, V_{\Delta})$.

 Now, using the same arguments as in the proof of \citet[Equation (7.3.9)]{BrockwellDavis1987}, we can show that 
 \begin{align*}
\lim_{m\to \infty}\limsup_{n\to \infty}\mathbb{P}(n^{1/2}|\gamma_{n;\Delta;m}^*(q\Delta)-\gamma_m(q\Delta)-\gamma^*(q \Delta)+\gamma(q\Delta)|>\epsilon)=0,  
 \end{align*}
 for all $\epsilon >0, q\in \{0, \ldots, h\}$.
 An application of  a variant of Slutsky's theorem, see \citet[Proposition 6.3.9]{BrockwellDavis1987} completes the proof.
\item 
  This part can be proven in similar way as the proof of \citet[Proposition 7.3.4]{BrockwellDavis1987}. Also, as in the proof of \citet[Theorem 3.5 b)]{CohenLindner2013}, we observe that $\sqrt{n} \overline{Y}_{n;\Delta}$ converges to a Gaussian random variable as $n \to \infty$ due to Theorem  \ref{thm:samplemean} and $\overline{Y}_{n;\Delta}$ converges to 0 in probability as $n \to \infty$ (since we assume here that $\mu=0).$
\item 
For the final part of the theorem, we can argue as in the proof of \citet[Theorem 7.2.1]{BrockwellDavis1987}, where the $w_{pq;\Delta}$ are obtained via the Bartlett formula
\begin{align*}
w_{pq;\Delta}=(v_{pq;\Delta}-\rho(p\Delta)v_{0q;\Delta}
-\rho(q\Delta)v_{p0;\Delta}
+\rho(p\Delta)\rho(q\Delta)v_{00;\Delta})/\gamma^2(0).
\end{align*}
We can simplify the above formula and write
\begin{align*}
    w_{pq;\Delta}=w_{pq;\Delta}^{(1)}+w_{pq;\Delta}^{(2)},
\end{align*}
where
\begin{align*}
  w_{pq;\Delta}^{(1)}&:=\frac{(\eta-3)\kappa_2^2}{\gamma^2(0)}  \int_{\R \times [0,\Delta]}
  (G_{p;\Delta}(x,u)-G_{0;\Delta}(x,u)\rho(p\Delta))\\
  &\quad \cdot
  (G_{q;\Delta}(x,u)-G_{0;\Delta}(x,u)\rho(q\Delta))dxdu,
\end{align*}
and
\begin{align*}
 w_{pq;\Delta}^{(2)}&:=
 \sum_{l=-\infty}^{\infty}
 \left[
 \rho(l\Delta)\rho((l+p-q)\Delta)+\rho((l-q)\Delta)\rho((l+p)\Delta) \right.\\
 &
 -2\rho(l\Delta)\rho((l-q)\Delta)\rho(p\Delta)
 -2\rho(l\Delta)\rho((l+p)\Delta)\rho(q\Delta)\\
 &\left. 
 +2\rho(p\Delta)\rho(q\Delta)\rho^2(l\Delta)\right].
\end{align*}
Note that 
\begin{align*}
   \sum_{l=-\infty}^{\infty}
  \rho(l\Delta)\rho((l+p-q)\Delta) 
  &= \sum_{l=-\infty}^{\infty}
  \rho((l+q)\Delta)\rho((l+p)\Delta),\\
  \sum_{l=-\infty}^{\infty}
  \rho(l\Delta)\rho((l-q)\Delta)\rho(p\Delta)
  &=
  \sum_{l=-\infty}^{\infty}
  \rho((l+q)\Delta)\rho(l\Delta)\rho(p\Delta).
\end{align*}
Hence,
    \begin{align*}
 w_{pq;\Delta}^{(2)}&=
 \sum_{l=-\infty}^{\infty}
 \left[
 \rho((l+q)\Delta)\rho((l+p)\Delta)+\rho((l-q)\Delta)\rho((l+p)\Delta) \right.\\
 &
 -2\rho((l+q)\Delta)\rho(l\Delta)\rho(p\Delta)
 -2\rho(l\Delta)\rho((l+p)\Delta)\rho(q\Delta)\\
 &\left. 
 +2\rho(p\Delta)\rho(q\Delta)\rho^2(l\Delta)\right].
\end{align*}
Hence,
\begin{align*}
  w_{pq;\Delta}&=
  \frac{(\eta-3)\kappa_2^2}{\gamma^2(0)}  \int_{\R \times [0,\Delta]}
  (G_{p;\Delta}(x,u)-G_{0;\Delta}(x,u)\rho(p\Delta))\\
  &\quad \cdot
  (G_{q;\Delta}(x,u)-G_{0;\Delta}(x,u)\rho(q\Delta))dxdu\\
  &+\sum_{l=-\infty}^{\infty}
 \left[
 \rho((l+q)\Delta)\rho((l+p)\Delta)+\rho((l-q)\Delta)\rho((l+p)\Delta) \right.\\
 &
 -2\rho((l+q)\Delta)\rho(l\Delta)\rho(p\Delta)
 -2\rho(l\Delta)\rho((l+p)\Delta)\rho(q\Delta)\\
 &\left. 
 +2\rho(p\Delta)\rho(q\Delta)\rho^2(l\Delta)\right].
\end{align*}
\end{enumerate}
\end{proof}

\subsection{Examples}
In this subsection, we show how the asymptotic variances appearing in the asymptotic theory for the sample mean can be computed for trawl processes with either an exponential or a supGamma trawl function. 

Note that, in the case when $p\equiv 1$, i.e.~for the (non-periodic) trawl process, we get
\begin{align*}
\gamma(h)=\Cov(Y_t, Y_{t+h})=\int_{|h|}^{\infty}g(u)du,
\end{align*}
for $t, h \in \R$.

\subsubsection{Exponential trawl} Consider the case of an exponential trawl function with 
$g(x)=\exp(-\lambda x)$, for $x \geq 0$.
The autocovariance function is given by 
\begin{align*}
\gamma(t)=\Cov(Y_t, Y_0)=\frac{1}{\lambda}e^{-\lambda |t|},
\end{align*}
for $t\in \R$,
and the 
autocorrelation function is given by 
\begin{align*}
\rho(t)=\Cor(Y_t, Y_0)=e^{-\lambda |t|},
\end{align*}
for $t\in \R$.

For the sample mean, we have the following result. 
Suppose that $\mathbb{E}(L')=0, \kappa_2=\Var(L')<\infty, \mu \in \mathbb{R}$ and $\Delta >0$.
Then 
\begin{align*}
\left(F_{\Delta}:\mathbb{R}\times [0, \Delta] \to [0, \infty], (x, u) \mapsto F_{\Delta}(x,u)=\sum_{j=-\infty}^{\infty}|f(x, u+j\Delta)|\right) \in L^2(\R \times [0, \Delta])
\end{align*}
since
\begin{align*}
F_{\Delta}(x,u)=\sum_{j=-\infty}^{\infty}|f(x, u+j\Delta)|
=\sum_{j=-\infty}^{\infty}\indicator_{(0,  g(u+j\Delta)}(x),
\end{align*}
and we have that 
\begin{align*}
\int_{\R \times [0, \Delta]}|F_{\Delta}(x,u)|^2dx du
&=\int_{\R \times [0, \Delta]}
\sum_{j=-\infty}^{\infty}\sum_{k=-\infty}^{\infty}\indicator_{(0, g(u+j\Delta)}(x)
\indicator_{(0, g(u+k\Delta)}(x)dx du
\\
&=\sum_{j=-\infty}^{\infty}\sum_{k=-\infty}^{\infty}\int_{\R \times [0, \Delta]}
\indicator_{(0, g(u+j\Delta)}(x)
\indicator_{(0, g(u+k\Delta)}(x)dx du\\
&=\sum_{j=0}^{\infty}\sum_{k=0}^{\infty}\int_{ [0, \Delta]}
\min\{g(u+j\Delta),g(u+k\Delta)\} du\\
&=\sum_{j=0}^{\infty}\int_{ [0, \Delta]}
g(u+j\Delta) du
+2 \sum_{j=0}^{\infty}\sum_{k=j+1}^{\infty}\int_{ [0, \Delta]}
g(u+k\Delta) du
\end{align*}
where we applied Tonelli's theorem.

Then $\sum_{j=-\infty}^{\infty} |\gamma(\Delta j)| < \infty$,
\begin{align}\label{as:sumgamma}
V_{\Delta}:=\sum_{j=-\infty}^{\infty}\gamma( \Delta j) = \kappa_2 \int_{\R \times [0, \Delta]} \left(\sum_{j=1}^{\infty}f(x, u+j\Delta)\right)^2 dx du, 
\end{align}
and the sample mean of $Y_{\Delta i}$, for $i=1, \ldots, n$, is asymptotically Gaussian as $n \to \infty$, i.e.
\begin{align*}
\sqrt{n}\left(\overline{Y}_{n; \Delta} - \mu \right)
\stackrel{\mathrm{d}}{\to} \mathrm{N}\left(0, V_{\Delta}\right), \quad \mathrm{as} \, n \to \infty.
\end{align*}

For the case of an exponential trawl function, we get
\begin{align*}
V_{\Delta}&=\sum_{j=-\infty}^{\infty}\gamma(\Delta j)
=\sum_{j=-\infty}^{\infty}\frac{1}{\lambda}e^{-\lambda |j|\Delta}
=\frac{1}{\lambda}\left(1+2\sum_{j=1}^{\infty}e^{-\lambda j \Delta}\right)
=\frac{1}{\lambda}\left(1+2\frac{e^{-\lambda  \Delta}}{1-e^{-\lambda  \Delta}}\right)\\
&=\frac{1+e^{-\lambda  \Delta}}{\lambda(1-e^{-\lambda  \Delta})}=\frac{1+e^{\lambda  \Delta}}{\lambda(e^{\lambda  \Delta}-1)}.
\end{align*}

\subsubsection{SupGamma trawl}
In the case when
$g(x)=(1+x/\alpha)^{-H}$, $\alpha>0, H>2$, i.e.~we require a short-memory setting, $x \geq 0$, we have
\begin{align*}
\gamma(h)=\int_{|h|}^{\infty}g(x)dx=\frac{\alpha}{H-1}\left(1+\frac{|h|}{\alpha}\right)^{1-H}.
\end{align*}
Then
\begin{align*}
V_{\Delta}&=\sum_{j=-\infty}^{\infty}\gamma(j \Delta)
=\frac{\alpha}{H-1}\sum_{j=-\infty}^{\infty}
\left(1+\frac{|j|\Delta}{\alpha}\right)^{1-H}
\\
&=\frac{\alpha}{H-1}\left(\frac{\Delta}{\alpha}\right)^{1-H} \sum_{j=-\infty}^{\infty}
\left(\frac{\alpha}{\Delta}+|j|\right)^{1-H}
\\
&=\frac{\alpha}{H-1}\left(\frac{\Delta}{\alpha}\right)^{1-H}
\left(\frac{\alpha}{\Delta} \right)^{1-H} \left[ 
2 \left(\frac{\alpha}{\Delta} \right)^{H-1}\zeta(H-1, \alpha/\Delta)-1\right]\\
&=\frac{\alpha}{H-1} \left[ 
2 \left(\frac{\alpha}{\Delta} \right)^{H-1}\zeta(H-1, \alpha/\Delta)-1\right],
\end{align*}
where $\zeta$ denotes the Hurwitz Zeta function defined by
$ \zeta(s, a)=\sum_{k=0}^{\infty}\frac{1}{(k+a)^s}$,
for $\mathrm{Re}(s)>1$.

\subsection{Verifying the assumptions of Theorem \ref{thm:acflimits} for selected periodic trawl processes}\label{sec:asscheck}
For the applications discussed in Section \ref{sec:inference}, we need to verify the condition \eqref{as:sumf2b}
from Proposition \ref{prop:covlimit} 
and Assumption \eqref{as:finsumsquint} from 
Theorem \ref{thm:acflimits}
    assuming that the corresponding moment assumptions for the L\'{e}vy seed hold. 

For both conditions, it is sufficient to check that a (non-periodic) trawl process satisfies the stated conditions since the periodic function is bounded. Hence, in the following, we shall set $p\equiv 1$.

\subsubsection{Verifying Condition \eqref{as:sumf2b}
from Proposition \ref{prop:covlimit} }
We need to check that \begin{align*}
\left(\mathbb{R}\times [0, \Delta] \to \R, (x, u) \mapsto \sum_{j=-\infty}^{\infty}f^2(x, u+j\Delta)\right) \in L^2(\R \times [0, \Delta]). 
\end{align*}
    This condition holds for trawl processes if, for $(x, u)\in \mathbb{R}\times [0, \Delta]$, 
\begin{align*}
 \sum_{j=-\infty}^{\infty}f^2(x, u+j\Delta) 
 =\sum_{j=-\infty}^{\infty}f(x, u+j\Delta)
 =\sum_{j=-\infty}^{\infty}\indicator_{(0, g(u+j\Delta))}(x) \indicator_{[0, \infty)}(u+j\Delta)
 =F_{\Delta}(x,u)\in L^2(\R \times [0, \Delta]).
\end{align*}
    This is equivalent to checking that 
    \begin{align}\begin{split}\label{fin}
       & \int_{\R \times [0, \Delta]}|F_{\Delta}(x,u)|^2dxdu
       \\
       &= \int_{\R \times [0, \Delta]}
        \sum_{j=-\infty}^{\infty}\sum_{k=-\infty}^{\infty}\indicator_{(0, g(u+j\Delta))}(x) \indicator_{[0, \infty)}(u+j\Delta)
        \indicator_{(0, g(u+k\Delta))}(x) \indicator_{[0, \infty)}(u+k\Delta)
        dxdu\\
        &\propto \sum_{j=-\infty}^{\infty}\gamma(j \Delta) < \infty.\end{split}
    \end{align}
    This condition is satisfied both for an exponential trawl function and also for a supGamma trawl function with short memory. In the latter case, we have that $\gamma(x) \propto (1+|x|/\alpha)^{1-H}$ for $\alpha>0, H>2$. Then, the finiteness of \eqref{fin} follows using the   $\zeta$-function representation.

    \subsubsection{Verifying  Assumption \eqref{as:finsumsquint} from 
Theorem \ref{thm:acflimits} }
We need to verify 
\begin{align*}
    \sum_{j=-\infty}^{\infty}\left(\int_{\R \times \R } |f(x, u)||f(x, u+j\Delta)| dxdu\right)^2<\infty.
    \end{align*}
    Using very similar computations as before, we find that the above condition is equivalent to 
    \begin{align*}
    &\sum_{j=-\infty}^{\infty}\left(\int_{\R \times \R } |f(x, u)||f(x, u+j\Delta)| dxdu\right)^2\\
    &=\sum_{j=-\infty}^{\infty}\left(\int_{\R \times \R } 
    \indicator_{(0, g(u)}(x) \indicator_{[0, \infty)}(u)
    \indicator_{(0, g(u+j\Delta))}(x) \indicator_{[0, \infty)}(u+j\Delta)
     dxdu\right)^2
    \propto\sum_{j=-\infty}^{\infty}\gamma^2(j\Delta)<\infty,
    \end{align*}
    which is  satisfied  by  the exponential trawl function and the supGamma trawl functions with $H>3/2$, which includes some long-memory settings.

\end{appendix}

\section*{Acknowledgement}
I would like to thank Paul Doukhan for suggesting a study of periodic trawl processes and for helpful discussions, as well as Michele Nguyen for commenting on an earlier version of this article. 
\bibliographystyle{agsm}
\bibliography{trawlbib}

\end{document}